\documentclass[11pt]{article}

\usepackage{graphicx}
\usepackage{amsmath,amssymb,amsfonts}
\usepackage{amsthm}
\usepackage{booktabs}
\usepackage{longtable}
\usepackage{tabularx}
\usepackage{xspace}
\usepackage[version=4]{mhchem}
\usepackage[numbers,sort&compress]{natbib}
\usepackage{tikz}
\usetikzlibrary{arrows.meta,positioning,calc,backgrounds,fit,shapes.geometric}
\usepackage{tikz-cd}
\usepackage{hyperref}
\usepackage{arxiv}

\providecolor{ioBlue}{HTML}{0072B2}      
\providecolor{ioVermilion}{HTML}{D55E00} 
\providecolor{ioGreen}{HTML}{009E73}     
\providecolor{ioOrange}{HTML}{E69F00}    
\providecolor{ioSky}{HTML}{56B4E9}       
\providecolor{ioPurple}{HTML}{CC79A7}    
\providecolor{ioYellow}{HTML}{F0E442}    

\providecolor{cdkC}{HTML}{1A1A1A}        
\providecolor{cdkH}{HTML}{2B2B2B}        
\providecolor{cdkN}{HTML}{2A4BD7}        
\providecolor{cdkO}{HTML}{CC0000}        
\providecolor{cdkBr}{HTML}{A0522D}       
\providecolor{cdkHal}{HTML}{1F9E4A}      
\providecolor{cdkS}{HTML}{C9A100}        
\newcommand{\aC}{\textcolor{cdkC}{C}}
\newcommand{\aH}{\textcolor{cdkH}{H}}
\newcommand{\aN}{\textcolor{cdkN}{N}}
\newcommand{\aO}{\textcolor{cdkO}{O}}
\newcommand{\aBr}{\textcolor{cdkBr}{Br}}

\tikzset{
  iobadge/.style={rounded corners=1.3pt, fill=black!78, text=white,
    font=\bfseries\scriptsize, minimum size=4.4mm, inner sep=1pt, anchor=center},
  iocbadge/.style={circle, fill=black!80, text=white, font=\bfseries\scriptsize,
    minimum size=4.6mm, inner sep=0.5pt, anchor=center},
  iocell/.style={rounded corners=3pt, draw=black!26, fill=black!1,
    line width=0.5pt},
  iotitle/.style={font=\bfseries\footnotesize, text=black!85, anchor=west,
    inner sep=0pt},
  iosubtitle/.style={font=\scriptsize, text=black!55, anchor=north west,
    inner sep=0pt},
  iopanel/.style={rounded corners=3pt, draw=black!16, fill=black!1,
    line width=0.4pt},
  iocard/.style={rounded corners=2.6pt, draw=black!14, fill=black!2,
    line width=0.4pt, inner sep=0pt},
  iomap/.style={-{Stealth[length=1.8mm,width=1.4mm]}, draw=ioVermilion!88!black,
    line width=0.65pt, shorten >=1pt, shorten <=1pt},
  ioimg/.style={draw=ioBlue, line width=1.9pt},
  ioflow/.style={-{Latex[length=2.4mm]}, line width=1.1pt, draw=ioBlue!85!black},
  iofish/.style={-{Latex[harpoon,length=2.0mm]}, line width=0.9pt,
    draw=ioVermilion!85!black},
  iobond/.style={line width=0.85pt, black!80},
  iohash/.style={line width=0.85pt, black!80, dash pattern=on 1pt off 1.2pt},
  iolp/.style={fill=ioBlue, circle, inner sep=0.75pt},        
  iorad/.style={fill=ioVermilion, circle, inner sep=0.95pt},  
  iolobe/.style={fill=ioGreen!75, draw=ioGreen!55!black, opacity=0.9}, 
  celp/.style={fill=black!82, circle, inner sep=0.72pt},      
  crad/.style={fill=ioVermilion, circle, inner sep=0.95pt},   
  cbond/.style={line width=0.9pt, black!82},                  
  cpair/.style={-{Latex[length=2.2mm]}, line width=1.0pt, draw=ioBlue!88!black},  
  cfish/.style={-{Latex[harpoon,length=2.0mm]}, line width=0.95pt,
    draw=ioVermilion!88!black}, 
  cchg/.style={font=\scriptsize, inner sep=0.5pt},            
  catom/.style={inner sep=1.1pt},                             
}
\newcommand{\chplus}{\ensuremath{\oplus}}
\newcommand{\chminus}{\ensuremath{\ominus}}

\newcommand{\iosub}[4][]{%
  \node[iocbadge, anchor=north west] (iosubB) at (#2) {#3};%
  \node[iotitle] (iosubT) at ($(iosubB.east)+(0.12,0)$) {#4};%
  \ifx\relax#1\relax\else
    \node[iosubtitle] at ($(iosubT.south west)+(0,-0.02)$) {#1};%
  \fi
}

\newcommand{\ioboxbadge}[3]{%
  \draw[iocell] (#1) rectangle (#2);%
  \node[iocbadge] at ($(#1 |- #2)+(0.40,-0.40)$) {#3};%
}

\newcommand{\synkit}{\texttt{SynKit}\xspace}
\newcommand{\rdkit}{\texttt{RDKit}\xspace}

\newcommand{\modfw}{M\O D\xspace}
\newcommand{\code}[1]{\texttt{#1}}

\newcommand{\Ie}{I_{\mathrm e}}
\newcommand{\ValG}{\operatorname{Val}}
\newcommand{\Loc}{\operatorname{Loc}}
\newcommand{\gout}{\operatorname{out}}
\newcommand{\gin}{\operatorname{in}}
\newcommand{\derives}[1]{\mathrel{\overset{#1}{\Longrightarrow}}}

\theoremstyle{plain}
\newtheorem{theorem}{Theorem}
\newtheorem{proposition}[theorem]{Proposition}
\newtheorem{lemma}[theorem]{Lemma}
\theoremstyle{remark}
\newtheorem{remark}{Remark}
\newtheorem{example}{Example}
\theoremstyle{definition}
\newtheorem{definition}{Definition}

\hypersetup{
  colorlinks=true,
  linkcolor=blue,
  citecolor=blue,
  urlcolor=blue,
  filecolor=blue,
  breaklinks=true
}

\begin{document}

\title{Lewis-labeled graphs: curly arrows and fishhooks as executable electron
transfers}
\shorttitle{Curly arrows and fishhooks as executable electron transfers}
\providecommand{\correspondingauthor}[1]{%
  \thanks{Corresponding author: \texttt{#1}}}
\author{%
  Tieu-Long Phan\correspondingauthor{tieu@bioinf.uni-leipzig.de}\\[0.3em]
  \normalsize Bioinformatics Group, Department of Computer Science,
  Interdisciplinary Center for Bioinformatics, and School for Embedded and Composite Artificial Intelligence (SECAI),
  Leipzig University, D-04107 Leipzig, Germany\\
  \normalsize Department of Mathematics and Computer Science,
  University of Southern Denmark, DK-5230 Odense M, Denmark%
}

\date{}
\maketitle

\begin{abstract}
The curly-arrow formalism is the lingua franca of organic reaction mechanisms,
but it is not executable.  Molecular graphs used in rule-based modeling encode
atomic connectivity while omitting lone pairs, radical electrons, and distinct
\(\sigma\) and \(\pi\) components.  Topological matching therefore cannot determine
whether a reactive center holds the electrons a step consumes, while matching
on derived formal charge can reject centers that do hold them.  We introduce the
Lewis-labeled graph (LLG), whose atom labels carry lone-pair and radical
populations and whose bond labels separate \(\sigma\)- and \(\pi\)-bond
occupancies.  Bond order, formal charge, and valence electron inventory are
derived from these fields.  Chemical transformations become resource-constrained double-pushout
rules, and curly arrows and coupled fishhooks become locus-sorted electron
transfers committed atomically from a common pre-state.  For a specified event
group, we prove that its execution, application of its induced rule, and the
corresponding integral occupancy update are equivalent.  Every admissible event
conserves valence electrons and net formal charge, and bond-centered fishhook
coupling follows from integrality rather than drawing convention.  Bidirectional
replay recovers all 39,732 mapped reference endpoints, while LLG admits 96 fewer
forward and 818 fewer inverse unique outcomes than conventional atom--bond
rules.  Ten radical records require manual annotation corrections.  After
review, transition construction succeeds for
101,313 of 101,314 records.  The remaining case requires an endpoint atom-map
correction rather than an arrow edit.  Strict replay accepts all 160 reviewed steps and rejects
all 1,120 controlled corruptions.  These results establish a common executable
state space for molecular graphs, reaction rules, and electron-flow annotations,
prior to questions of kinetic or thermodynamic feasibility.
\end{abstract}


\section{Introduction}
\label{sec:introduction}

A curly arrow~\cite{KermackRobinson1922} instructs a redistribution of electrons, yet it is inert on graphs that only understand bonds. When molecular representations strip away the valence inventories of the Lewis electron-pair model~\cite{Lewis1916AtomAndMolecule,Langmuir1919ArrangementOfElectrons}, such as lone pairs, unpaired radicals, and $\pi$ systems, rewrite rules~\cite{corradini1997algebraic,ehrig2006,andersen2016} cannot verify electron availability before they fire. Algorithms then match too permissively or too selectively, and lose atomicity. They might protonate a saturated \ce{NH4+} cation, reject an amide anion carrying exactly the pair a substitution needs because formal charge is an accounting convention rather than an observed resource~\cite{Parkin2006}, or execute coupled homolytic fishhooks as two independent half actions~\cite{Hoffman2004}. Section~\ref{sec:prelim} formalizes these matching failures and the need to record local electron capacity.

Chemical transformations are recorded at three semantic levels, as mapped reactant and product pairs~\cite{schwaller2021}, as graph rewrite rules~\cite{andersen2016,Andersen2018RuleComposition}, and as electron-flow annotations~\cite{Tavakoli2023,Tavakoli2024}. Representations that make electrons explicit close that gap~\cite{benko2003graph}. The most complete is the half-edge graph of~\citet{holzschuh2026halfedge}, whose electron-anonymous quotient, the electron-feature graph, carries lone-pair and radical counts on atoms and a bond multiplicity on bonds, and admits double-pushout rules acting on those electron resources. What that work establishes is \emph{existence}, that every balanced mapped reaction admits some sequence of electron pushes and that feature-level rules are quotients of half-edge rules. The operational converse remains open, and it is what a curator or a mechanism generator~\cite{Bradshaw2019Electro,Chen2023MechFinder,Joung2025FlowER} needs. Given a supplied and grouped annotation, decide whether it executes from a given state and reaches the declared endpoint. Existence of an explanation is not conformance of a candidate. Two details bear on it. A bond multiplicity does not distinguish a \(\sigma\) pair from a \(\pi\) pair, so the reacting pair of a multiple bond is not canonically selectable. Pushes applied one at a time supply no synchronized group assessed against one common pre-state, which is what makes coupled fishhooks obligate rather than conventional.

The Lewis-labeled graph (LLG) supplies the state that decision requires. Vertex labels carry lone-pair and radical populations, edge labels decouple $\sigma$ from $\pi$ bond occupancy, and bond order, formal charge, and valence-electron inventory are derived from these four fields rather than stored. The same fields serve at once as molecular state, as the matchable resources of a rule, and as the coordinates of electron flow, which aligns the three levels. A mapped reaction fixes two states and their correspondence, a rule declares a local resource delta, and a mechanism factors that delta into synchronized admissible events. The construction places these fields inside standard double-pushout rewriting~\cite{corradini1997algebraic,ehrig2006} and extends the attributed-graph architecture of \synkit~\cite{phan2025synkit}.

Section~\ref{sec:lsg} makes open-shell species first-class states on the ordinary atom--bond skeleton, so standard graph~\cite{hagberg2008} and cheminformatics~\cite{rdkit} libraries still apply. Section~\ref{sec:rewriting} turns the electron fields into matchable resources, so availability becomes an inequality. Partial atom maps used in the experiments are expanded by the method of Ref.~\cite{Laffitte2026}. Sections~\ref{sec:arrows} and~\ref{sec:grammar} show the three levels to be one decision on one coordinate set, and coupled fishhooks to be obligate by integrality rather than by convention. Section~\ref{sec:experiments} validates the representation by bidirectional rule replay, the transition semantics by controlled corruption, and the grammar by a corpus audit. Ten radical annotations require manual correction, and one further record requires an endpoint atom-map correction that is not applied.

\section{Notation and preliminaries}
\label{sec:prelim}

\subsection{Molecular graphs, rules, and imaginary transition states}
\label{sec:prelim-graphs}

Our notation and foundational definitions follow prior work~\cite{gonzalezlaffitte2024, phan2025syntemp,Laffitte2026}.
Unless stated otherwise, all graphs are finite, undirected, and simple. A molecular graph has vertices for atoms and edges for chemical bonds, with attributes on both carrying chemical information.
A \emph{homomorphism} $\varphi: G \to H$ is a vertex map with $\{u,v\}\in E(G)\Rightarrow\{\varphi u,\varphi v\}\in E(H)$, that is, one preserving adjacency.
A \emph{monomorphism} is an injective homomorphism, equivalently an isomorphism of $G$ onto a subgraph of $H$ that need not be induced, so exhibiting one is an instance of subgraph isomorphism~\cite{Ullmann1976SubgraphIsomorphism,GareyJohnson1979Computers}.
An \emph{isomorphism} is a bijective homomorphism that also reflects adjacency, equivalently one whose inverse is again a homomorphism.
These notions concern the underlying unlabeled graphs. Label compatibility is imposed on top of them and is defined separately for each label family below.

Chemical reactions act locally and are naturally formalized by a
double-pushout (DPO) rewriting rule~\cite{corradini1997algebraic,ehrig2006,andersen2016}, $p=(L \xleftarrow{\,l\,} K \xrightarrow{\,r\,} R)$, where the left-hand side \(L\) is the reactant-side pattern, the right-hand
side \(R\) is the product-side pattern, and the interface \(K\) records the
substructure that persists.  A rule is applied to a host graph \(G\) by
choosing a match \(m:L\to G\) and constructing the two pushout squares
\[
\begin{tikzcd}[column sep=large,row sep=large]
L \arrow[d,"m"'] &
K \arrow[l,"l"'] \arrow[r,"r"] \arrow[d,"m'"] &
R \arrow[d,"m''"] \\
G &
D \arrow[l,"g"] \arrow[r,"h"'] &
H ,
\end{tikzcd}
\]
yielding the derivation $G \derives{(p,m)} H$.  For atom-preserving
chemical rules, vertices are neither deleted nor created. The transformation
is carried by bond edits and label changes, and the span induces the local
atom-to-atom correspondence~\cite{corradini1997algebraic,ehrig2006,andersen2016}.

A mapped reaction can equivalently be condensed into a single graph.  Given an
atom-balanced reaction \(G\to H\) with an atom-to-atom map (AAM)
\(\alpha: V(G) \to V(H)\), the \emph{imaginary transition state}
(ITS)~\cite{fujita1986, gonzalezlaffitte2024}, later reintroduced as the
\emph{condensed graph of reaction}~\cite{hoonakker2011}, is the graph
\(\Upsilon(G,H,\alpha)\) on \(V(G)\) whose edge set is the union of the reactant and
product bonds transported along \(\alpha\), and whose edge \(\{u,v\}\) carries
the ordered pair \(\bigl(b_G(u,v),\,b_H(\alpha u,\alpha v)\bigr)\) of its
reactant- and product-side bond states (with an absent bond read as the null
state).  The \emph{reaction center} is the subgraph on the edges whose two
components differ together with the vertices whose label fields change.

The ITS and the rule span are two views of one object.  By Proposition~4 of Ref~\cite{phan2025syntemp}, every subgraph
\(\Gamma\subseteq\Upsilon(G,H,\alpha)\) containing the reaction center determines
a unique atom-preserving DPO rule \(p=(L\xleftarrow{l}K\xrightarrow{r}R)\) with
\(\Upsilon(L,R,r\circ l^{-1})\cong\Gamma\), and a match \(m\) making
\(G\derives{(p,m)}H\) well defined.  Its reactant side gives \(L\), its product
side \(R\), and the agreeing sides the interface \(K\), so \(\Gamma\) fixes the
preserved context the rule carries.  Section~\ref{sec:lsg} refines the bond and
vertex labels of this span.  The AAM \(\alpha\) builds the span but is not a
vertex attribute, correspondence within the rule being carried by \(l\)
and \(r\).

\subsection{Limitations of omitting electron representation}
\label{sec:limits}

Structural matching on element-and-bond patterns fails in complementary ways.

\paragraph{Over-permissive matches}
Consider a protonation rule whose reacting nitrogen must donate a lone pair.  If
its left-hand side records only a nitrogen vertex, the pattern also embeds at the
nitrogen of ammonium, which has no lone pair to donate in the two-center Lewis
model (Figure~\ref{fig:matching}A1). What the representation omits is the
resource, not the expressiveness of DPO rewriting.  An
expanded electron-locus graph encodes the precondition with an explicit
lone-pair object~\cite{benko2003graph} (A2) but
enlarges the subgraph-isomorphism
input~\cite{GareyJohnson1979Computers,Ullmann1976SubgraphIsomorphism}. The compact
alternative keeps the atom--bond skeleton and requires \(\ell_N\ge1\) (A3).
Compactness here means graph-object count, not a runtime claim.  A natural
repair is instead to constrain the pattern by formal charge, ammonia being
neutral and ammonium cationic.  It excludes the unwanted match, but only by
trading one failure for its opposite.

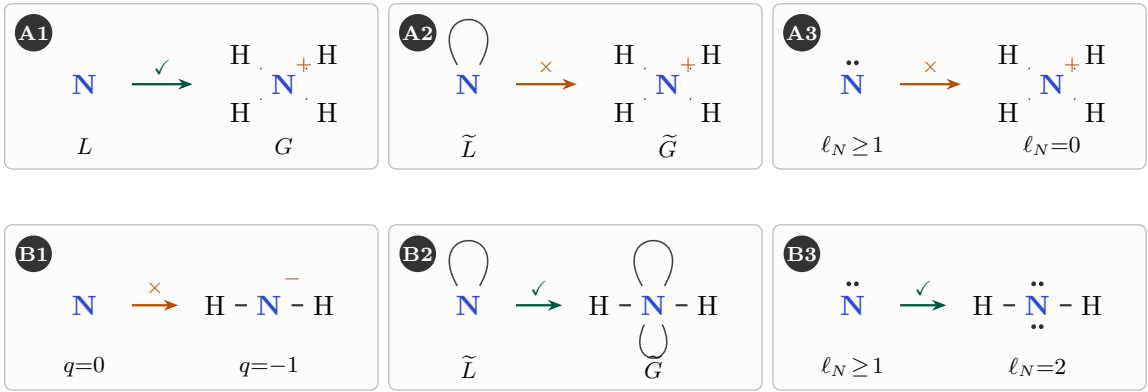
\begin{figure*}[htbp]
  \centering
  \resizebox{\textwidth}{!}{
\begin{tikzpicture}[
  font=\small,
  bond/.style={line width=0.85pt, black!80},
  dbond/.style={line width=0.85pt, double distance=1.0pt, black!80},
  lploop/.style={line width=0.6pt, black!75},
  satom/.style={font=\normalsize\bfseries, text=cdkN},
  atom/.style={font=\normalsize},
  glab/.style={font=\footnotesize, inner sep=1pt},
  goodarrow/.style={-{Stealth[length=2mm]}, line width=0.9pt, ioGreen!55!black,
                    shorten <=3pt, shorten >=3pt},
  badarrow/.style={-{Stealth[length=2mm]}, line width=0.9pt, ioVermilion!85!black,
                   shorten <=3pt, shorten >=3pt},
  vlab/.style={font=\scriptsize},
  colhdr/.style={font=\bfseries\small, align=center}
]
\def\cW{5.3}
\def\rH{3.05}
\def\boxSW{-0.55,-0.66}
\def\boxNE{4.60,1.62}
\def\xL{0.55}\def\xG{3.30}\def\xGb{3.10}
\def\lpup#1{\fill[black!82] (#1)++(100:0.30) circle(0.033); \fill[black!82] (#1)++(80:0.30) circle(0.033);}
\def\lpdn#1{\fill[black!82] (#1)++(260:0.30) circle(0.033); \fill[black!82] (#1)++(280:0.30) circle(0.033);}


\begin{scope}[shift={(0,0)}]
  \ioboxbadge{\boxSW}{\boxNE}{A1}
  \node[satom] (a1N) at (\xL,0.51) {N};
  \node[glab] at (\xL,-0.34) {\(L\)};
  \node[satom] (a1Nc) at (\xG,0.51) {N};
  \node[atom] (a1Hnw) at (2.70,0.92) {H}; \node[atom] (a1Hne) at (3.90,0.92) {H};
  \node[atom] (a1Hsw) at (2.70,0.12) {H}; \node[atom] (a1Hse) at (3.90,0.12) {H};
  \node[vlab,text=ioVermilion!80!black] at (3.58,0.78) {\(+\)};
  \draw[bond] (a1Nc)--(a1Hnw); \draw[bond] (a1Nc)--(a1Hne);
  \draw[bond] (a1Nc)--(a1Hsw); \draw[bond] (a1Nc)--(a1Hse);
  \node[glab] at (\xG,-0.34) {\(G\)};
  \draw[goodarrow] (1.10,0.51) -- node[above,vlab] {\(\checkmark\)} (2.15,0.51);
\end{scope}
\begin{scope}[shift={(\cW,0)}]
  \ioboxbadge{\boxSW}{\boxNE}{A2}
  \node[satom] (a2N) at (\xL,0.51) {N};
  \draw[lploop] (a2N.120) to[out=120,in=60,looseness=8] (a2N.60);
  \node[glab] at (\xL,-0.34) {\(\widetilde L\)};
  \node[satom] (a2Nc) at (\xG,0.51) {N};
  \node[atom] (a2Hnw) at (2.70,0.92) {H}; \node[atom] (a2Hne) at (3.90,0.92) {H};
  \node[atom] (a2Hsw) at (2.70,0.12) {H}; \node[atom] (a2Hse) at (3.90,0.12) {H};
  \node[vlab,text=ioVermilion!80!black] at (3.58,0.78) {\(+\)};
  \draw[bond] (a2Nc)--(a2Hnw); \draw[bond] (a2Nc)--(a2Hne);
  \draw[bond] (a2Nc)--(a2Hsw); \draw[bond] (a2Nc)--(a2Hse);
  \node[glab] at (\xG,-0.34) {\(\widetilde G\)};
  \draw[badarrow] (1.10,0.51) -- node[above,vlab] {\(\times\)} (2.15,0.51);
\end{scope}
\begin{scope}[shift={(2*\cW,0)}]
  \ioboxbadge{\boxSW}{\boxNE}{A3}
  \node[satom] (a3N) at (\xL,0.51) {N};
  \lpup{a3N}
  \node[glab] at (\xL,-0.34) {\(\ell_N\!\ge\!1\)};
  \node[satom] (a3Nc) at (\xG,0.51) {N};
  \node[atom] (a3Hnw) at (2.70,0.92) {H}; \node[atom] (a3Hne) at (3.90,0.92) {H};
  \node[atom] (a3Hsw) at (2.70,0.12) {H}; \node[atom] (a3Hse) at (3.90,0.12) {H};
  \node[vlab,text=ioVermilion!80!black] at (3.58,0.78) {\(+\)};
  \draw[bond] (a3Nc)--(a3Hnw); \draw[bond] (a3Nc)--(a3Hne);
  \draw[bond] (a3Nc)--(a3Hsw); \draw[bond] (a3Nc)--(a3Hse);
  \node[glab] at (\xG,-0.34) {\(\ell_N{=}0\)};
  \draw[badarrow] (1.10,0.51) -- node[above,vlab] {\(\times\)} (2.15,0.51);
\end{scope}

\begin{scope}[shift={(0,-\rH)}]
  \ioboxbadge{\boxSW}{\boxNE}{B1}
  \node[satom] (b1N) at (\xL,0.51) {N};
  \node[glab] at (\xL,-0.34) {\(q{=}0\)};
  \node[atom] (b1Hl) at (2.35,0.51) {H};
  \node[satom] (b1Nc) at (\xGb,0.51) {N};
  \node[vlab,text=ioVermilion!80!black] at (3.42,0.88) {\(-\)};
  \node[atom] (b1Hr) at (3.85,0.51) {H};
  \draw[bond] (b1Hl)--(b1Nc); \draw[bond] (b1Nc)--(b1Hr);
  \node[glab] at (\xGb,-0.34) {\(q{=}{-}1\)};
  \draw[badarrow] (1.10,0.51) -- node[above,vlab] {\(\times\)} (1.95,0.51);
\end{scope}
\begin{scope}[shift={(\cW,-\rH)}]
  \ioboxbadge{\boxSW}{\boxNE}{B2}
  \node[satom] (b2N) at (\xL,0.51) {N};
  \draw[lploop] (b2N.120) to[out=120,in=60,looseness=8] (b2N.60);
  \node[glab] at (\xL,-0.34) {\(\widetilde L\)};
  \node[atom] (b2Hl) at (2.35,0.51) {H};
  \node[satom] (b2Nc) at (\xGb,0.51) {N};
  \node[atom] (b2Hr) at (3.85,0.51) {H};
  \draw[bond] (b2Hl)--(b2Nc); \draw[bond] (b2Nc)--(b2Hr);
  \draw[lploop] (b2Nc.120) to[out=120,in=60,looseness=8] (b2Nc.60);
  \draw[lploop] (b2Nc.245) to[out=245,in=295,looseness=7] (b2Nc.295);
  \node[glab] at (\xGb,-0.34) {\(\widetilde G\)};
  \draw[goodarrow] (1.10,0.51) -- node[above,vlab] {\(\checkmark\)} (1.95,0.51);
\end{scope}
\begin{scope}[shift={(2*\cW,-\rH)}]
  \ioboxbadge{\boxSW}{\boxNE}{B3}
  \node[satom] (b3N) at (\xL,0.51) {N};
  \lpup{b3N}
  \node[glab] at (\xL,-0.34) {\(\ell_N\!\ge\!1\)};
  \node[atom] (b3Hl) at (2.35,0.51) {H};
  \node[satom] (b3Nc) at (\xGb,0.51) {N};
  \node[atom] (b3Hr) at (3.85,0.51) {H};
  \draw[bond] (b3Hl)--(b3Nc); \draw[bond] (b3Nc)--(b3Hr);
  \lpup{b3Nc}\lpdn{b3Nc}
  \node[glab] at (\xGb,-0.34) {\(\ell_N{=}2\)};
  \draw[goodarrow] (1.10,0.51) -- node[above,vlab] {\(\checkmark\)} (1.95,0.51);
\end{scope}
\end{tikzpicture}}
  \caption{Two failure modes of purely structural patterns, each in three
  representations: \emph{structural graph} (col.~1), \emph{electron-locus
  multigraph} with lone pairs as self-loops (col.~2), and \emph{compact label
  fields} (col.~3, LLG).  \textbf{Row A (over-permissive):} a topology-only
  donor pattern matches saturated ammonium, which the lone-pair locus and
  \(\ell_N\ge1\) both reject.  \textbf{Row B (over-selective):} strict
  formal-charge matching rejects an anionic amide, which both resource
  conditions admit.  Green \(\checkmark\)/red \(\times\) mark match/no-match.}
  \label{fig:matching}
\end{figure*}

\paragraph{Over-selective matches}
That repair treats a derived bookkeeping value as a rigid, equality-checked
label.  Formal charge is the prime example, an accounting convention derived
from a Lewis assignment rather than an observed
resource~\cite{Hoffman2004, Parkin2006}.  A nucleophilic-substitution rule
written with neutral nitrogen now fails to match an anionic amide such as
\ce{H2N-}, though the anion carries exactly the pair the step needs
(Figure~\ref{fig:matching}, row B).  Charge-annotated
interfaces~\cite{Andersen2018RuleComposition} mitigate but do not resolve this,
since the condition is resource availability, not equality of a derived value.
Relaxing that equality to an inequality on charge rescues this example but not
the general one, because Equation~\eqref{eq:charge} sums lone pairs, radicals,
bonds, and hydrogens into one number and cannot isolate the pair a step
consumes.  Requiring \(\ell_N\ge1\) names that pair directly and settles all
three centers, admitting ammonia and the amide anion while rejecting ammonium.

\subsection{Related frameworks}
\label{sec:related}

Existing frameworks resolve these gaps only partly. \modfw~\cite{andersen2016,Andersen2018RuleComposition}
implements DPO rewriting over graphs labeled with opaque strings or first-order
terms (the category $\mathbf{Graph}_{\mathcal{L}}$).  Because its rule morphisms
match labels by equality, a resource inequality such as \(\ell_N\ge1\) is not
expressible in the label algebra: variation in oxidation state, radical
character, or charge must be enumerated as disjoint label values or pushed
outside the span into constraint code, and in neither case can the rewriting
layer add, subtract, or conserve the electron resource.  This concerns where the
bookkeeping lives, not whether \modfw can express the chemistry.
\code{CGRtools}~\cite{nugmanov2019cgrtools,nugmanov2022} condenses a reaction
into one graph with dynamic bond attributes and atom-level radical marks, but
these are state attributes, not matchable resources under localized conservation
laws.
\rdkit~\cite{rdkit} matches SMARTS templates applying local heuristic edits,
without span-based pushout semantics, serving perception rather than rule
identity.  \synkit~1.0~\cite{phan2025synkit} supported mapped graphs and ITS
representations but modeled no nonbonding electrons or radical resources.

Mechanism-generation systems address a complementary task.  \code{ELECTRO}
learns two-electron paths from mapped endpoints~\cite{Bradshaw2019Electro},
\code{MechFinder} applies expert-coded templates~\cite{Chen2023MechFinder}, and
\code{FlowER} models electron redistribution by flow
matching~\cite{Joung2025FlowER}, each enforcing conservation internally.  The
present work is not a generator but a deterministic and auditable acceptance
relation, so such a system can sit upstream of it without changing its
semantics.

The closest formal comparison is the half-edge graph (HEG) and electron-feature
graph (EFG) of~\citet{holzschuh2026halfedge}.  Their HEG is a tuple
$\mathcal{H}=(H,V,E,\psi)$ with half-edges $H$, atoms $V$, a partition $E$ of a
subset of $H$ into bond pairs and lone-pair loops, and an ownership map
$\psi\colon V\to 2^H$, on which they prove that any balanced mapped reaction is
explained by a sequence of single- and two-electron half-edge moves.  Dropping
electron identity gives the EFG, the category of ACSets over a schema with
sorts $\textsf{Atom}$, $\textsf{Pair}$, $\textsf{LP}$, $\textsf{Rad}$, reached
by a coarsening functor
$C\colon\mathbf{HEG}^{\mathrm{nf}}\to\mathbf{EFG}$.  In labeled form the EFG
already carries lone-pair and radical counts on atoms and a bond multiplicity
on bonds, and they show it to be an exact static quotient of half-edge graphs.
They further give double-pushout rules as spans of these ACSets, so electron
resources are themselves deleted, preserved, and created by the rule, and prove
these to be quotients of electron-anonymous half-edge rules.

Neither an explicit electron representation nor rewriting over one is therefore
what separates the present work, and three other things do.  Their results are
existence and quotient statements, and neither settles whether a supplied and
grouped annotation can execute from a given state and reach a declared
endpoint.  A single bond multiplicity leaves the reacting pair of a multiple
bond non-canonical, whereas separating \(\sigma\) from \(\pi\) makes
\(\pi_e\) a determinate transition locus.  Their pushes also carry no event
group whose availability is read from one common pre-state, which is what
forces coupled fishhooks to commit together.  The two encodings also place the
bookkeeping differently.  Matching one $\textsf{LP}$ object embeds into any atom
carrying at least one, so an ACSet rule obtains resource inequality for free at
the cost of object count, while the scalar encoding keeps the atom--bond
skeleton compact and expresses the same condition as a residual variable with an
application condition (Section~\ref{sec:rewriting}).  Neither dominates, and
their categorical relation is deferred to
Section~\ref{sec:efg-llg-relation}.

\section{Lewis-labeled graphs}
\label{sec:lsg}

\subsection{Definition and policy-valid states}

\begin{definition}
\label{def:lsg}
A \emph{Lewis-labeled graph} (LLG) is an attributed undirected molecular graph
\(G=(V,E,a,b)\), where each atom \(v\) has label
\(a(v)=(el_v,\ell_v,r_v)\in\mathrm{Symbol}\times\mathbb N_0^2\)
(element, lone pairs, and radical electrons) and each
edge \(e=\{u,v\}\) has label \(b(e)=(s_e,p_e)\in\mathbb N_0^2\) giving its
\(\sigma\)- and \(\pi\)-bond orders.  The neutral-atom valence-electron count
\(z_v\) is a cached function of \(el_v\), not an independent variable.  The total
Kekul\'e bond order is
\begin{equation}
  o_e = s_e + p_e.
  \label{eq:bond-order}
\end{equation}
For clear endpoints we write \(s_{uv},p_{uv},o_{uv}\). Thus ammonia has \(\ell_N=1\), a methyl radical has \(r_C=1\), and an ordinary double bond has \((s_e,p_e)=(1,1)\).
\end{definition}

All hydrogens are explicit vertices here.  An implementation may instead store
an optional implicit-hydrogen count as auxiliary state.

By modeling electron resources as scalar label fields rather than pseudo-atoms or electron vertices, the underlying graph handed to standard graph algorithms remains the compact atom--bond graph of Section~\ref{sec:prelim-graphs}.
Conventional two-center bonds carry \(s_e=1\) and \(p_e \in \{0,1,2\}\), representing single, double, and triple bonds respectively, and Figure~\ref{fig:llg-anatomy} shows an example with its derived quantities.
These same four fields are the electron \emph{loci} of the LLG, the sites between which the arrow-pushing calculus of Section~\ref{sec:arrows} transfers electrons.

\begin{figure}[htbp]
  \centering
  \resizebox{\textwidth}{!}{
\begin{tikzpicture}[
  font=\small,
  panel/.style={iopanel, line width=0.55pt},
  stored/.style={rounded corners=2.5pt, align=center, inner sep=5pt,
                 font=\scriptsize},
  vertex/.style={stored, draw=ioBlue!58!black, fill=ioBlue!6},
  edge/.style={stored, draw=ioPurple!65!black, fill=ioPurple!7},
  result/.style={rounded corners=2.5pt, draw=black!20, fill=white,
                 align=center, inner sep=4.5pt, minimum height=1.42cm,
                 font=\scriptsize},
  tag/.style={font=\scriptsize\bfseries, text=black!58},
  leadA/.style={draw=ioBlue!70!black, line width=0.5pt,
                dash pattern=on 1.5pt off 1.5pt, -{Stealth[length=3.2pt]}},
  leadB/.style={draw=ioPurple!72!black, line width=0.5pt,
                dash pattern=on 1.5pt off 1.5pt, -{Stealth[length=3.2pt]}},
  flow/.style={line width=0.9pt, -{Stealth[length=3.6pt,width=3.2pt]}},
  halo/.style={draw=ioBlue!60!black, line width=0.7pt, opacity=0.55}
]

\path[use as bounding box] (-6.25,-2.80) rectangle (6.25,2.60);

\draw[panel] (-6.10,-2.65) rectangle (2.75,2.45);
\draw[panel] ( 3.00,-2.65) rectangle (6.10,2.45);

\iosub{-5.75,2.13}{A}{Formaldehyde encoding}

\node[vertex, text width=2.72cm] (vf) at (-4.48,0.86) {%
  \textcolor{ioBlue!72!black}{\textbf{oxygen vertex}}\\
  \(a(v)=(\mathrm{el},\textcolor{ioBlue!80!black}{\ell},\textcolor{ioBlue!80!black}{r})\)\\
  \(a(\mathrm O)=(\mathrm O,\textcolor{ioBlue!80!black}{\mathbf 2},\textcolor{ioBlue!80!black}{\mathbf 0})\)
};

\node[catom, font=\normalsize\bfseries] (C) at (-1.72,0.90) {\aC};
\node[catom, font=\normalsize\bfseries] (O) at (-0.57,0.90) {\aO};
\node[catom] (H1) at (-2.47,1.39) {\aH};
\node[catom] (H2) at (-2.47,0.41) {\aH};
\draw[cbond] (C)--(H1);
\draw[cbond] (C)--(H2);

\node[halo, circle, minimum size=6.2mm] at (O) {};

\draw[line width=1.55pt, ioBlue!85!black]
  ($(C)+(0.20,0.085)$)--($(O)+(-0.20,0.085)$);
\draw[line width=1.55pt, ioPurple!82!black]
  ($(C)+(0.20,-0.085)$)--($(O)+(-0.20,-0.085)$);
\node[tag, text=ioBlue!75!black] at ($(C)!0.5!(O)+(0,0.31)$) {\(\sigma\)};
\node[tag, text=ioPurple!78!black] at ($(C)!0.5!(O)+(0,-0.31)$) {\(\pi\)};

\node[celp] at ($(O)+(0.28,0.30)$) {};
\node[celp] at ($(O)+(0.10,0.42)$) {};
\node[celp] at ($(O)+(0.28,-0.30)$) {};
\node[celp] at ($(O)+(0.10,-0.42)$) {};

\node[edge, text width=2.35cm] (ef) at (1.35,0.86) {%
  \textcolor{ioPurple!76!black}{\textbf{C--O edge}}\\
  \(b(e)=(\textcolor{ioPurple!80!black}{s},\textcolor{ioPurple!80!black}{p})\)\\
  \(b(\{\mathrm C,\mathrm O\})=(\textcolor{ioPurple!80!black}{\mathbf 1},\textcolor{ioPurple!80!black}{\mathbf 1})\)
};

\draw[leadA] (vf.north east) to[out=22,in=150] ($(O)+(-0.02,0.40)$);
\draw[leadB] (ef.south west) to[out=-150,in=-30] ($(C)!0.58!(O)+(0,-0.17)$);

\node[tag, text=black!45] at (-5.55,-0.34) {decode};
\draw[flow, ioPurple!70!black] (-4.78,-0.52) -- (-4.78,-0.66);   
\draw[flow, ioBlue!70!black]  (-2.15,-0.52) -- (-2.15,-0.66);   
\draw[flow, black!45]         ( 0.92,-0.52) -- ( 0.92,-0.66);   

\node[result, text width=1.85cm] at (-4.78,-1.42) {%
  \textcolor{ioPurple!78!black}{\textbf{bond order}}\\[1pt]
  \(o_{\mathrm{CO}}=s+p\)\\
  \(=1+1=2\)
};

\node[result, text width=2.62cm] at (-2.15,-1.42) {%
  \textcolor{ioBlue!75!black}{\textbf{nonbonding electrons}}\\[1pt]
  \(n_{\mathrm O}=2\ell+r\)\\
  \(=2(2)+0=4\)
};

\node[result, text width=2.68cm] at (0.92,-1.42) {%
  \textbf{formal charge}\\[1pt]
  \(q_{\mathrm O}=z_{\mathrm O}-n_{\mathrm O}-\sum o_e\)\\
  \(=6-4-2=0\)
};

\iosub{3.05,2.13}{B}{Methyl radical}

\node[catom, font=\normalsize\bfseries] (Rc) at (4.55,0.92) {\aC};
\node[crad] (rad) at ($(Rc)+(0,0.44)$) {};
\node[halo, draw=ioBlue!55!black, circle, minimum size=3.2mm] at (rad) {};
\node[catom] (Rh1) at (3.80,0.33) {\aH};
\node[catom] (Rh2) at (5.30,0.33) {\aH};
\node[catom] (Rh3) at (4.55,0.12) {\aH};
\draw[cbond] (Rc)--(Rh1);
\draw[cbond] (Rc)--(Rh2);
\draw[cbond] (Rc)--(Rh3);

\node[vertex, text width=2.62cm] (rb) at (4.55,-1.30) {%
  \textcolor{ioBlue!72!black}{\textbf{carbon vertex}}\\
  \(a(\mathrm C)=(\mathrm C,0,\underset{r}{\underline{\textcolor{ioBlue!85!black}{\mathbf 1}}})\)\\[2pt]
  \(n_{\mathrm C}=2\ell+r=2(0)+1=1\)\\
  \textcolor{black!55}{one unpaired electron}
};

\draw[leadA] (rb.north east) to[out=72,in=-18] ($(rad)+(0.14,-0.02)$);

\end{tikzpicture}}
  \caption{\textbf{(A)} Stored vertex and
  edge fields for formaldehyde, with three quantities decoded from them.
  \textbf{(B)} Radical occupancy represented as a vertex field.}
  \label{fig:llg-anatomy}
\end{figure}
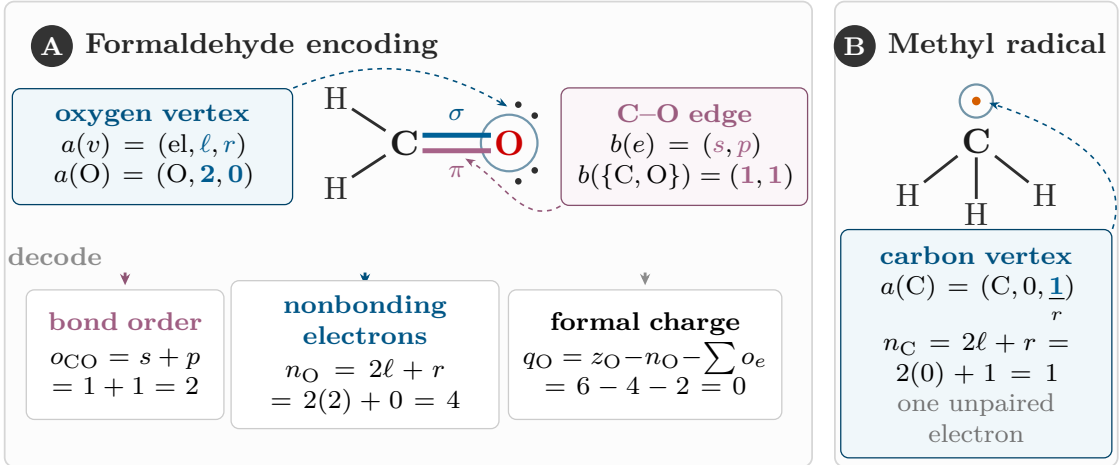

\begin{definition}
\label{def:loci}
The \emph{electron loci} of an LLG \(G\) are
\begin{equation}
  \Loc(G) = \{\operatorname{lp}(v), \rho_v \colon v \in V\}
  \cup \{\sigma_{\{u,v\}}, \pi_{\{u,v\}} \colon \{u,v\}\in\tbinom{V}{2}\},
  \label{eq:locus-set}
\end{equation}
namely the lone-pair and radical loci \(\operatorname{lp}(v),\rho_v\) on each vertex and the
\(\sigma\)- and \(\pi\)-bond loci \(\sigma_{\{u,v\}},\pi_{\{u,v\}}\) on each
unordered pair (abbreviated \(\sigma_{uv},\pi_{uv}\)).  Their \emph{occupancy}
reads off the four label fields,
\[
  c_G(\operatorname{lp}(v))=\ell_v,\quad c_G(\rho_v)=r_v,\quad
  c_G(\sigma_{uv})=s_{uv},\quad c_G(\pi_{uv})=p_{uv},
\]
with \(c_G(\sigma_{uv})=c_G(\pi_{uv})=0\) when \(\{u,v\}\notin E(G)\).  Each locus
carries a \emph{sort} in \(\{\operatorname{lp},\rho,\sigma,\pi\}\), and these four
sorts are exactly the four resource coordinates \(J=\{\ell,r,s,p\}\) used for
rewriting in Section~\ref{sec:rewriting}.  Apart from \(\operatorname{lp}\),
each sort symbol is the Greek counterpart of its occupancy letter, so that
\(\rho\), \(\sigma\), and \(\pi\) carry \(r\), \(s\), and \(p\) respectively.  The locus view therefore adds no
state: it records only that the rewriting coordinates and the electron-flow
coordinates are one and the same.
\end{definition}

\begin{definition}
\label{def:policy-valid}
An LLG \(G\) is \emph{\(\Pi\)-valid} when all \(\ell_v,r_v,s_e,p_e\) are nonnegative integers and satisfy the local constraints selected by \(\Pi\). A main-group policy may require \(s_e\le1\), \(p_e>0\Rightarrow s_e=1\), and
  \[
    2\ell_v+r_v+2\sum_{e\sim v}(s_e+p_e)
    \le \kappa_{\Pi}(v).
  \]
Here \(\kappa_{\Pi}(v)\in\mathbb N_0\cup\{\infty\}\) is the
policy-selected upper bound on the valence-shell electron population at
\(v\).  For example, a duet/octet policy assigns 2 to hydrogen and 8 to an
octet-limited main-group atom, while \(\kappa_{\Pi}(v)=\infty\) disables this
bound.
A closed-shell policy also imposes \(r_v=0\), while open-shell, hypervalent, and transition-metal policies may relax the corresponding bounds.
\end{definition}

All formal results are parameterized by \(\Pi\).  The corpus experiments
instantiate it as
\begin{equation}
  \Pi_{\mathrm{exec}}:\quad
  \ell_v,r_v,s_e,p_e\in\mathbb N_0,\quad
  s_e\in\{0,1\},\quad
  p_e>0\implies s_e=1,\quad
  \kappa_{\Pi}\equiv\infty.
  \label{eq:pi-exec}
\end{equation}
Two relaxations are deliberate and independent: dropping the octet bound admits
hypervalent main-group structures, which appear throughout the input corpora,
and permitting \(r_v>0\) admits the open-shell states the radical corpus
requires.  Element-specific capacities and transition-metal bonding remain
outside the empirical claim.  Being this permissive,
\(\Pi_{\mathrm{exec}}\) rejects only non-integral or negative occupancies and
malformed \(\sigma/\pi\) combinations, so the substantive acceptance conditions
tested in Section~\ref{sec:experiments} are integrality, common-prestate
availability, and endpoint consistency rather than the policy itself.  We report
the permissive case so that no rejection is attributable to a contestable
valence convention.

\subsection{Derived quantities and conservation}
\label{sec:conservation}

\begin{definition}
\label{def:derived}
For an LLG \(G\) and \(v\in V(G)\), the \emph{nonbonding electron count}, the
\emph{degree}, and the derived \emph{formal charge} are
\begin{equation}
  n_v = 2\ell_v + r_v,
  \label{eq:nbe}
\end{equation}
\begin{equation}
  \deg(v) = n_v + \sum_{e \sim v} o_e,
  \label{eq:degree}
\end{equation}
\begin{equation}
 q_v = z_v - \deg(v),
 \label{eq:charge}
\end{equation}
where \(e\sim v\) denotes incidence.  The degree collects the electrons \(v\)
owns, all of its nonbonding electrons together with one per unit of bond order,
as in the standard Lewis convention~\cite{Parkin2006}.  It is the incident half-edge count of
Ref.~\cite{holzschuh2026halfedge}, so \(q_v=\operatorname{ch}(v)\) there and the
two models yield identical charges, the LLG computing them from scalar counts
rather than by partitioning half-edges.
\end{definition}

Boundary formal charges are redundant: derived directly from label fields,
they serve as consistency checks rather than independent variables.
Define the \emph{represented electron inventory} and
\emph{neutral valence pool} by
\begin{equation}
  \Ie(G) = \sum_{v \in V(G)} \deg(v),
  \label{eq:inventory-def}
\end{equation}
\begin{equation}
  \ValG(G) = \sum_{v \in V(G)} z_v.
  \label{eq:valence-pool-def}
\end{equation}
Since each edge is incident with two vertices, expanding
Equation~\eqref{eq:degree} gives
\(\Ie(G)=\sum_{v}n_v+2\sum_{e}o_e\).

\begin{lemma}
\label{lem:inventory}
For any LLG \(G\), let \(Q(G)=\sum_{v\in V(G)}q_v\).  Then
\begin{equation}
  Q(G) = \ValG(G) - \Ie(G).
  \label{eq:charge-inventory-identity}
\end{equation}
\end{lemma}

\begin{proof}
Summing Equation~\eqref{eq:charge} over \(V(G)\) gives
\(Q(G)=\sum_{v}z_v-\sum_{v}\deg(v)=\ValG(G)-\Ie(G)\).
\end{proof}

\begin{proposition}
\label{prop:conservation}
Let \(G\to H\) be atom balanced: there exists an element-preserving bijection
\(\alpha\colon V(G)\to V(H)\).  Then
\begin{equation}
  Q(H) - Q(G) = -\bigl(\Ie(H) - \Ie(G)\bigr),
  \label{eq:charge-inventory-diff}
\end{equation}
so the represented electron inventory is conserved if and only if total
formal charge is conserved.
\end{proposition}

\begin{proof}
The element-preserving bijection gives
\(\ValG(G)=\ValG(H)\).  Subtracting the two instances of
Lemma~\ref{lem:inventory} proves the claim.
\end{proof}

This conservation condition is intentionally weaker than chemical feasibility:
it captures only algebraic stoichiometry and formal redox balance~\cite{dugundji1973algebraic,silverstein2011oxidation},
providing a baseline invariant check while kinetic and thermodynamic
constraints remain outside the transition semantics.

\section{Resource-constrained DPO rewriting of LLGs}
\label{sec:rewriting}

The construction separates permission, availability, and validity: undefined
interface attributes specify which coordinates a rule may change, residual terms
test whether a host has the required resources and instantiate the rule at a
chosen occurrence, ordinary DPO gluing constructs the successor, and the
application predicate with policy \(\Pi\) decides admission.

\subsection{Resource terms and partial attributes}

We use ordinary DPO rewriting with symbolic resource labels and explicit
application conditions~\cite{orejas2011symbolic,habel2009correctness}, every
instantiated attributed match remaining exact.  Let
\(J=\{\ell,r,s,p\}\) index the four resource coordinates, equivalently the four
locus sorts \(\{\operatorname{lp},\rho,\sigma,\pi\}\) of
Definition~\ref{def:loci}, and let
\[
  A=\bigl(A_j,+_j,0_j\bigr)_{j\in J},
  \qquad A_j=\mathbb N_0,
\]
be the \(J\)-sorted commutative-monoid algebra.  Sorting prevents substitution
between unlike resources, and addition separates what a rule consumes or
produces from the residual retained from the host.  For a finite sorted
variable set \(X\), let \(T_A(X)\) be the corresponding term algebra generated
by \(X\), constants in \(A\), and sort-respecting addition.  Every
sort-preserving valuation \(\theta:X\to A\) extends uniquely to the evaluation
homomorphism \(\widehat\theta:T_A(X)\to A\).

Let \(\Sigma\) be the many-sorted attribute signature of
Definition~\ref{def:lsg}, with fixed carriers for the non-resource sorts and
resource carriers supplied by an algebra \(B\).  Write
\(\mathsf{PAttr}_{\Sigma}(B)\) for the category of finite graphs with
partial \(\Sigma\)-attributes.  Its morphisms preserve incidence and every
attribute defined at the domain~\cite{habel2002relabelling}.  Host graphs lie in
the totally labeled subcategory
\(\mathbf{LLG}_{\Pi}^{=}\) of \(\Pi\)-valid LLGs with exact morphisms.
Rule schemata instead use \(B=T_A(X)\) and partial interfaces.  We restrict
rule legs and matches to injective relabeling morphisms.  Natural pushout
complements along this class exist uniquely in the setting used
below~\cite{habel2002relabelling,lack2004adhesive,ehrig2014madhesive}.

\subsection{Rule schemata and partial relabeling interfaces}

\begin{definition}
\label{def:partial-relabeling-interface}
Let \(L\) and \(R\) be totally labeled graphs over \(T_A(X)\).  A graph \(K\) in a span
\(L\xleftarrow{l}K\xrightarrow{r}R\) is a \emph{partial relabeling interface}
when it contains exactly the persistent vertices and edges and, for every
persistent object \(q\) and resource coordinate \(j\),
\begin{equation}
 c_K^j(q)\mathbin{\downarrow}
 \quad\Longrightarrow\quad
 c_L^j(l(q))=c_K^j(q)=c_R^j(r(q)).
 \label{eq:defined-interface-coordinate}
\end{equation}
A defined coordinate is therefore preserved by both rule legs.  The value
\(c_K^j(q)=\bot\) removes that equality constraint.  It permits relabeling but
does not supply the new value.
\end{definition}

\begin{definition}
\label{def:resource-constrained-llg-rule}
A \emph{resource-constrained LLG rule schema} is a span of monomorphisms
\begin{equation}
  p=\bigl(L\xleftarrow{\,l\,}K\xrightarrow{\,r\,}R\bigr)
  \quad\text{in }\mathsf{PAttr}_{\Sigma}\bigl(T_A(X)\bigr),
  \label{eq:resource-constrained-rule-span}
\end{equation}
together with a finite sorted variable set \(X\) and an application predicate
\(\Phi_p(G,m,\theta)\).  Here \(L,R\) are total and \(K\) is a partial
relabeling interface.  Every changed coordinate \(x\) on a persistent carrier
has the residual form
\begin{equation}
  c_L(x)=u_x+d_x^{-},\qquad c_R(x)=u_x+d_x^{+},
  \qquad u_x\in X,\quad d_x^{-},d_x^{+}\in\mathbb N_0,
  \label{eq:residual-rule-label}
\end{equation}
where \(u_x\) is fresh for \(x\).  It denotes the resource left after the rule
consumes \(d_x^{-}\), before the rule contributes \(d_x^{+}\).  Whenever
\(d_x^{-}\neq d_x^{+}\) the two labels differ, so
Equation~\eqref{eq:defined-interface-coordinate} forces \(c_K(x)=\bot\): a
changed coordinate is necessarily undefined in the interface.  Created and
deleted edges are represented by \(R\setminus K\) and \(L\setminus K\),
respectively.  Element identity is defined in \(K\), while charge and total
bond order remain derived.  All rules considered below are atom-preserving,
so \(V(K)=V(L)=V(R)\).
\end{definition}

\begin{example}
\label{ex:beta-scission-lsg}
For neopentoxyl-radical \(\beta\)-scission,
\[
  \ce{(CH3)3C-CH2-O^{.} -> (CH3)3C^{.} + CH2=O},
\]
the interface leaves the changing radical and \(\pi\) coordinates undefined
and deletes the \(c_1{-}c_2\) edge structurally
(Figure~\ref{fig:beta-scission-dpo}).  Its deltas balance,
\(1(-1)+2(+1)+2(-1)+1(+1)=0\), so the rule conserves \(\Ie\) and hence, by
Proposition~\ref{prop:conservation}, total charge.
Definition~\ref{def:resource-constrained-llg-rule} does not force that balance
in general.  Section~\ref{sec:arrows} obtains it for every rule induced by an
event group and factors the same delta into three fishhooks.
\end{example}

\begin{figure*}[htbp]
  \centering
  \resizebox{\textwidth}{!}{
\begin{tikzpicture}[
  font=\small,
  catom/.append style={font=\normalsize, inner sep=1.8pt},
  el/.style={inner sep=1pt, font=\scriptsize},
  vname/.style={font=\scriptsize\itshape, text=black!55, inner sep=0.4pt},
  ccap/.style={font=\footnotesize\itshape, text=black!72, align=center},
  mor/.style={-{Stealth[length=2.2mm]}, line width=0.7pt, draw=black!55},
  ml/.style={font=\itshape, text=black!70, inner sep=1pt},
  free/.style={font=\footnotesize, text=ioVermilion!85!black},
  dltn/.style={font=\scriptsize, text=black!60, align=center}
]
\def\cA{0}\def\cB{5.4}\def\cC{10.8}
\def\rT{0}\def\rB{-4.6}

\begin{scope}[shift={(\cA,\rT)}]
  \node[catom](Lo)  at (-1.25,0){\aO};  \node[crad] at ($(Lo)+(0,0.42)$){};
  \node[catom](Lc20) at (0,0){\aC};
  \node[catom](Lc21) at (1.25,0){\aC};
  \draw[cbond](Lo)--(Lc20); \draw[cbond](Lc20)--(Lc21);
  \node[vname] at ($(Lo)+(0,-0.42)$){$o$};
  \node[vname] at ($(Lc20)+(0,-0.42)$){$c_1$};
  \node[vname] at ($(Lc21)+(0,-0.42)$){$c_2$};
  \node[ccap] at (0,-1.2){$L$};
\end{scope}
\begin{scope}[shift={(\cB,\rT)}]
  \node[catom](Ko)  at (-1.25,0){\aO};
  \node[catom](Kc20) at (0,0){\aC};
  \node[catom](Kc21) at (1.25,0){\aC};
  \draw[cbond](Ko)--(Kc20);
  \draw[dash pattern=on 1.6pt off 1.6pt, draw=black!30, line width=0.7pt] (Kc20)--(Kc21);
  \node[free] at ($(Ko)+(0,0.42)$){$\bot$};
  \node[free] at ($(Ko)!0.5!(Kc20)+(0,0.30)$){$\bot$};
  \node[free] at ($(Kc21)+(0,0.42)$){$\bot$};
  \node[vname] at ($(Ko)+(0,-0.42)$){$o$};
  \node[vname] at ($(Kc20)+(0,-0.42)$){$c_1$};
  \node[vname] at ($(Kc21)+(0,-0.42)$){$c_2$};
  \node[ccap] at (0,-1.2){$K$};
\end{scope}
\begin{scope}[shift={(\cC,\rT)}]
  \node[catom](Ro)  at (-1.25,0){\aO};
  \node[catom](Rc20) at (0,0){\aC};
  \node[catom](Rc21) at (1.25,0){\aC};  \node[crad] at ($(Rc21)+(0,0.42)$){};
  \draw[cbond]($(Ro)+(0.20,0.06)$)--($(Rc20)+(-0.20,0.06)$);
  \draw[cbond]($(Ro)+(0.20,-0.06)$)--($(Rc20)+(-0.20,-0.06)$);
  \node[vname] at ($(Ro)+(0,-0.42)$){$o$};
  \node[vname] at ($(Rc20)+(0,-0.42)$){$c_1$};
  \node[vname] at ($(Rc21)+(0,-0.42)$){$c_2$};
  \node[ccap] at (0,-1.2){$R$};
\end{scope}

\begin{scope}[shift={(\cA,\rB)}]
  \begin{pgfonlayer}{background}
    \fill[ioBlue!8, rounded corners=3pt] (-2.15,-0.34) rectangle (1.05,0.66);
  \end{pgfonlayer}
  \node[catom](Go)  at (-1.65,0){\aO};  \node[crad] at ($(Go)+(0,0.42)$){};
  \node[catom](Gc20) at (-0.45,0){\aC};
  \node[catom](Gc21) at (0.75,0){\aC};
  \draw[cbond](Go)--(Gc20); \draw[cbond](Gc20)--(Gc21);
  \draw[cbond](Gc20)--($(Gc20)+(0,0.52)$); \node[el] at ($(Gc20)+(0,0.68)$){\aH};
  \draw[cbond](Gc20)--($(Gc20)+(0,-0.52)$); \node[el] at ($(Gc20)+(0,-0.68)$){\aH};
  \draw[cbond]($(Gc21)+(0.14,0.06)$)--++(0.52,0.5);
  \draw[cbond]($(Gc21)+(0.14,-0.06)$)--++(0.52,-0.5);
  \draw[cbond]($(Gc21)+(0.20,0)$)--++(0.7,0);
  \node[vname] at ($(Go)+(-0.02,-0.42)$){$o$};
  \node[vname] at ($(Gc21)+(-0.05,-0.42)$){$c_2$};
  \node[ccap] at (-0.4,-1.28){$G$};
\end{scope}
\begin{scope}[shift={(\cB,\rB)}]
  \node[catom](Do)  at (-1.7,0){\aO};  \node[crad] at ($(Do)+(0,0.42)$){};
  \node[catom](Dc20) at (-0.55,0){\aC};
  \draw[cbond](Do)--(Dc20);
  \draw[cbond](Dc20)--($(Dc20)+(0,0.52)$); \node[el] at ($(Dc20)+(0,0.68)$){\aH};
  \draw[cbond](Dc20)--($(Dc20)+(0,-0.52)$); \node[el] at ($(Dc20)+(0,-0.68)$){\aH};
  \node[catom](Dc21) at (0.95,0){\aC};
  \draw[cbond]($(Dc21)+(0.14,0.06)$)--++(0.52,0.5);
  \draw[cbond]($(Dc21)+(0.14,-0.06)$)--++(0.52,-0.5);
  \draw[cbond]($(Dc21)+(-0.14,0.06)$)--++(-0.5,0.5);
  \node[vname] at ($(Do)+(-0.02,-0.42)$){$o$};
  \node[vname] at ($(Dc21)+(-0.05,-0.42)$){$c_2$};
  \node[ccap] at (0,-1.28){$D$};
\end{scope}
\begin{scope}[shift={(\cC,\rB)}]
  \node[catom](Ho)  at (-1.7,0){\aO};
  \node[catom](Hc20) at (-0.55,0){\aC};
  \draw[cbond]($(Ho)+(0.20,0.06)$)--($(Hc20)+(-0.20,0.06)$);
  \draw[cbond]($(Ho)+(0.20,-0.06)$)--($(Hc20)+(-0.20,-0.06)$);
  \draw[cbond](Hc20)--($(Hc20)+(0,0.52)$); \node[el] at ($(Hc20)+(0,0.68)$){\aH};
  \draw[cbond](Hc20)--($(Hc20)+(0,-0.52)$); \node[el] at ($(Hc20)+(0,-0.68)$){\aH};
  \node[el, text=black!55] at (0.15,0){$+$};
  \node[catom](Hc21) at (0.95,0){\aC};  \node[crad] at ($(Hc21)+(0,0.44)$){};
  \draw[cbond]($(Hc21)+(0.14,-0.06)$)--++(0.52,-0.5);
  \draw[cbond]($(Hc21)+(-0.14,-0.06)$)--++(-0.5,-0.5);
  \draw[cbond]($(Hc21)+(0.20,0.02)$)--++(0.68,0.16);
  \node[vname] at ($(Ho)+(-0.02,-0.42)$){$o$};
  \node[vname] at ($(Hc21)+(-0.05,-0.42)$){$c_2$};
  \node[ccap] at (0,-1.28){$H$};
\end{scope}

\draw[mor] (\cB-2.15,\rT+0.05) -- (\cA+2.15,\rT+0.05)
  node[ml, midway, above=0pt]{$l$};
\draw[mor] (\cB+2.15,\rT+0.05) -- (\cC-2.15,\rT+0.05)
  node[ml, midway, above=0pt]{$r$};
\draw[mor] (\cB-2.15,\rB+0.05) -- (\cA+2.15,\rB+0.05)
  node[ml, midway, above=0pt]{$g$};
\draw[mor] (\cB+2.15,\rB+0.05) -- (\cC-2.15,\rB+0.05)
  node[ml, midway, above=0pt]{$h$};
\draw[mor] (\cA,\rT-1.7) -- (\cA,\rB+1.55) node[ml, midway, right=1pt]{$m$};
\draw[mor] (\cB,\rT-1.7) -- (\cB,\rB+1.55) node[ml, midway, right=1pt]{$m'$};
\draw[mor] (\cC,\rT-1.7) -- (\cC,\rB+1.55) node[ml, midway, right=1pt]{$m''$};

\end{tikzpicture}}
  \caption{Resource-constrained DPO derivation for the neopentoxyl-radical
  \(\beta\)-scission of Example~\ref{ex:beta-scission-lsg}.
  Top: the rule span \(L\xleftarrow{l}K\xrightarrow{r}R\), whose partial
  interface leaves \(\bot\)-marked coordinates free and deletes the dashed
  \(c_1{-}c_2\) edge. Bottom: match \(m\) maps the rule into the shaded region
  of host \(G\), producing \(H\) via pushout complement \(D\).}
  \label{fig:beta-scission-dpo}
\end{figure*}
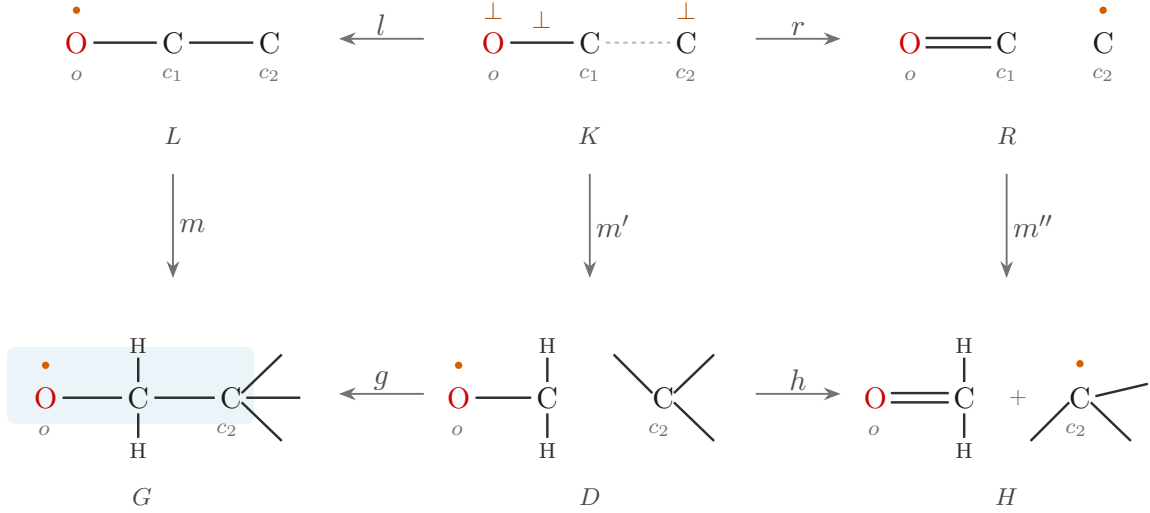

\subsection{Admissible matches and rule instantiation}

A structural match first tests resource availability and then fixes the
residuals.  These are distinct operations because \(\bot\) belongs to the rule
interface, while a residual's value depends on the selected host occurrence.
Let \(|Q|\) denote \(Q\)
without its resource labels but with incidence and non-resource labels retained.

\begin{definition}
\label{def:admissible-dpo-match}
Let \(p\) be a resource-constrained LLG rule schema and let \(G\) be a
\(\Pi\)-valid LLG.  A \emph{structural match} is a monomorphism
\(\mu:|L|\hookrightarrow|G|\).  It is \emph{resource-admissible} when, for
every defined resource label,
\begin{equation}
  \widehat\theta\!\left(c_L^j(q)\right)
  =c_G^j\!\left(\mu(q)\right)
  \qquad\text{for all \(q,j\) with \(c_L^j(q)\mathbin{\downarrow}\)}
  \label{eq:valuation-constraint-system}
\end{equation}
has a sort-preserving solution \(\theta:X\to A\).  For a fresh residual label
\(u_x+d_x^-\), this is equivalent to
\begin{equation}
  c_G(\mu(x))\ge d_x^-.
  \label{eq:resource-admissible-inequality}
\end{equation}
Ground labels still match exactly.  If a schema deliberately reuses a
variable, all corresponding equations must assign it the same value.
\end{definition}

Equation~\eqref{eq:resource-admissible-inequality} is an application condition
on \(\mu\), not a weakened attributed-graph morphism.

\begin{definition}
\label{def:rule-instantiation}
For the fresh-residual form of
Equation~\eqref{eq:residual-rule-label}, a resource-admissible match induces
the unique \emph{instantiation}
\begin{equation}
  \theta_\mu(u_x)=c_G(\mu(x))-d_x^-.
  \label{eq:residual-valuation}
\end{equation}
For \(Q\in\{L,K,R\}\), write \(Q\theta_\mu\) for pointwise evaluation of all
defined resource labels.  This gives the concrete span
\(p\theta_\mu=(L\theta_\mu\leftarrow K\theta_\mu\rightarrow R\theta_\mu)\)
and lifts \(\mu\) to the exact match \(m:L\theta_\mu\hookrightarrow G\).
Incidence, non-resource labels, and \(\bot\) are unchanged.
\end{definition}

\begin{lemma}
\label{lem:evaluation}
For a sort-preserving valuation \(\theta:X\to A\), pointwise evaluation
\(Q\mapsto Q\theta\) is a functor
\(\mathsf{PAttr}_{\Sigma}(T_A(X))\to\mathsf{PAttr}_{\Sigma}(A)\) that fixes the
underlying graph, every non-resource label, and every undefined coordinate.  It
sends the rule schema
\(p=(L\xleftarrow{l}K\xrightarrow{r}R)\) to a span
\(p\theta=(L\theta\xleftarrow{l\theta}K\theta\xrightarrow{r\theta}R\theta)\) of
monomorphisms in \(\mathsf{PAttr}_{\Sigma}(A)\) whose interface \(K\theta\) is
again a partial relabeling interface.
\end{lemma}

\begin{proof}
\(\widehat\theta\) is a sorted algebra homomorphism, so it preserves label
equalities, identities, and composition.  Evaluation fixes graph maps and
therefore preserves monomorphisms.  Applying it to
Equation~\eqref{eq:defined-interface-coordinate} preserves every interface
equality, while \(\bot\) remains undefined.
\end{proof}

\begin{example}
\label{ex:protonation-residual}
For nitrogen protonation, the reacting lone-pair count changes from
\(u_\ell+1\) to \(u_\ell\), while the N--H bond is created structurally.  At
\(\ce{H2N^-}\), where \(\ell_N=2\), Equation~\eqref{eq:residual-valuation}
gives \(\theta_\mu(u_\ell)=1\).  The rule forms neutral \(\ce{NH3}\) and leaves
one lone pair.  At ammonia, the same schema gives \(u_\ell=0\).  Thus one
schema covers both hosts (Figure~\ref{fig:rule-instantiation}).
\end{example}

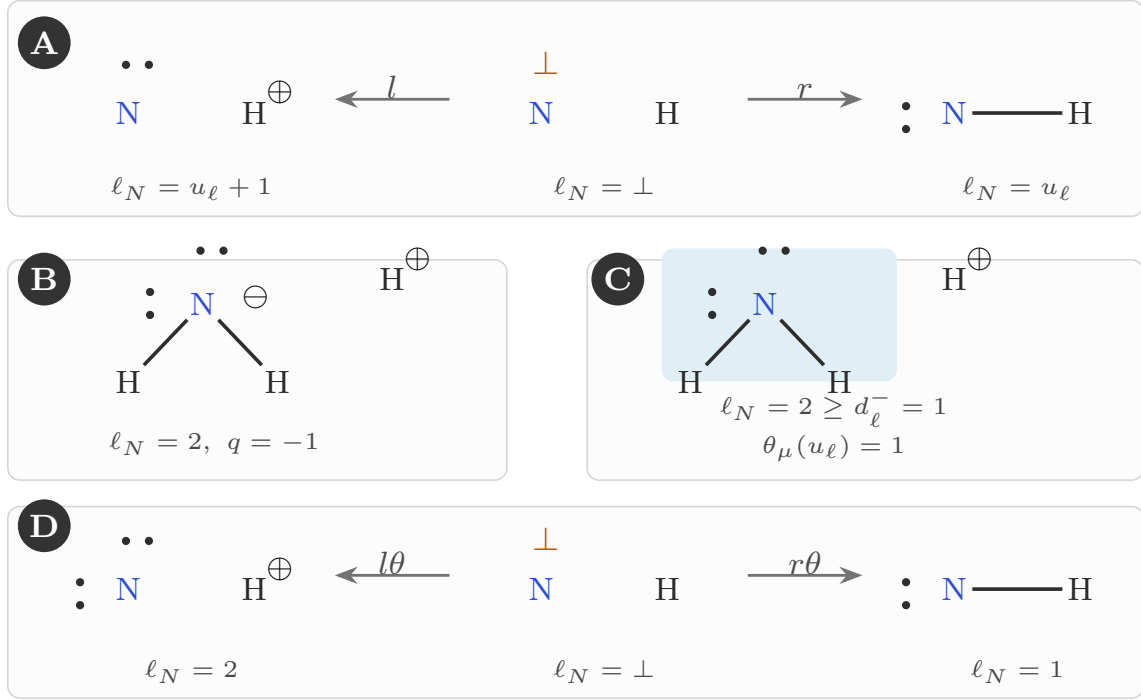
\begin{figure}[!ht]
  \centering
  \resizebox{\textwidth}{!}{
\begin{tikzpicture}[
  font=\scriptsize,
  lbl/.style={font=\tiny, text=black!72, align=center},
  free/.style={font=\scriptsize, text=ioVermilion!88!black},
  mor/.style={-{Stealth[length=2.0mm]}, line width=0.7pt, draw=black!55},
  ml/.style={font=\scriptsize\itshape, text=black!70, inner sep=1pt},
  note/.style={font=\scriptsize, text=black!70, align=left},
  stage/.style={rounded corners=2.4pt, draw=ioBlue!55!black, fill=ioBlue!6,
                align=center, inner sep=3.4pt, font=\scriptsize},
  stagek/.style={stage, draw=ioGreen!58!black, fill=ioGreen!7},
  stagelbl/.style={font=\scriptsize\itshape, text=black!68, inner sep=1pt},
  pstep/.style={-{Latex[length=1.9mm]}, line width=0.7pt, draw=black!52,
               shorten >=1pt, shorten <=1pt},
]

\begin{scope}[shift={(0,0)}]
  \iosub{-0.15,0.78}{A}{}
  \node[catom](LN) at (0.75,0){\aN};
  \node[celp] at ($(LN)+(-0.02,0.42)$){}; \node[celp] at ($(LN)+(0.18,0.42)$){};
  \node[catom](LH) at (1.85,0){\aH}; \node[cchg] at ($(LH)+(0.22,0.18)$){\chplus};
  \node[lbl] at (1.3,-0.66){\(\ell_N=u_\ell+1\)};
  \node[catom](KN) at (4.35,0){\aN}; \node[free] at ($(KN)+(0.04,0.44)$){\(\bot\)};
  \node[catom](KH) at (5.45,0){\aH};
  \node[lbl] at (4.9,-0.66){\(\ell_N=\bot\)};
  \node[catom](RN) at (7.95,0){\aN}; \node[celp] at ($(RN)+(-0.42,0.06)$){}; \node[celp] at ($(RN)+(-0.42,-0.14)$){};
  \node[catom](RH) at (9.05,0){\aH}; \draw[cbond](RN)--(RH);
  \node[lbl] at (8.5,-0.66){\(\ell_N=u_\ell\)};
  \draw[mor](3.55,0.12)--(2.55,0.12) node[ml,midway,above=-1pt]{\(l\)};
  \draw[mor](6.15,0.12)--(7.15,0.12) node[ml,midway,above=-1pt]{\(r\)};
\end{scope}

\begin{scope}[shift={(0,-2.00)}]
  \iosub{-0.15,0.72}{B}{}
  \node[catom](GN) at (1.4,0.34){\aN};
  \node[celp] at ($(GN)+(-0.02,0.46)$){}; \node[celp] at ($(GN)+(0.18,0.46)$){};
  \node[celp] at ($(GN)+(-0.46,0.10)$){}; \node[celp] at ($(GN)+(-0.46,-0.10)$){};
  \node[cchg] at ($(GN)+(0.46,0.06)$){\chminus};
  \node[catom](GH1) at (0.75,-0.34){\aH}; \draw[cbond](GN)--(GH1);
  \node[catom](GH2) at (2.05,-0.34){\aH}; \draw[cbond](GN)--(GH2);
  \node[catom](GH3) at (3.05,0.55){\aH}; \node[cchg] at ($(GH3)+(0.22,0.18)$){\chplus};
  \node[lbl] at (1.5,-0.88){\(\ell_N=2,\ q=-1\)};
\end{scope}

\begin{scope}[shift={(4.9,-2.00)}]
  \iosub{-0.05,0.72}{C}{}
  \fill[ioBlue!12, rounded corners=3pt] (0.5,-0.34) rectangle (2.55,0.82);
  \node[catom](CN) at (1.4,0.34){\aN};
  \node[celp] at ($(CN)+(-0.02,0.46)$){}; \node[celp] at ($(CN)+(0.18,0.46)$){};
  \node[celp] at ($(CN)+(-0.46,0.10)$){}; \node[celp] at ($(CN)+(-0.46,-0.10)$){};
  \node[catom](CH) at (2.05,-0.34){\aH}; \draw[cbond](CN)--(CH);
  \node[catom](CH0) at (0.75,-0.34){\aH}; \draw[cbond](CN)--(CH0);
  \node[catom](CHp) at (3.05,0.55){\aH}; \node[cchg] at ($(CHp)+(0.22,0.18)$){\chplus};
  \node[lbl] at (2.0,-0.72){\(\ell_N=2\ge d_\ell^-=1\)\\[1pt]\(\theta_\mu(u_\ell)=1\)};
\end{scope}

\begin{scope}[shift={(0,-4.15)}]
  \iosub{-0.15,0.72}{D}{}
  \node[catom](DLN) at (0.75,0){\aN};
  \node[celp] at ($(DLN)+(-0.02,0.42)$){}; \node[celp] at ($(DLN)+(0.18,0.42)$){};
  \node[celp] at ($(DLN)+(-0.42,0.06)$){}; \node[celp] at ($(DLN)+(-0.42,-0.14)$){};
  \node[catom](DLH) at (1.85,0){\aH}; \node[cchg] at ($(DLH)+(0.22,0.18)$){\chplus};
  \node[lbl] at (1.3,-0.72){\(\ell_N=2\)};
  \node[catom](DKN) at (4.35,0){\aN}; \node[free] at ($(DKN)+(0.04,0.44)$){\(\bot\)};
  \node[catom](DKH) at (5.45,0){\aH};
  \node[lbl] at (4.9,-0.72){\(\ell_N=\bot\)};
  \node[catom](DRN) at (7.95,0){\aN}; \node[celp] at ($(DRN)+(-0.42,0.06)$){}; \node[celp] at ($(DRN)+(-0.42,-0.14)$){};
  \node[catom](DRH) at (9.05,0){\aH}; \draw[cbond](DRN)--(DRH);
  \node[lbl] at (8.5,-0.72){\(\ell_N=1\)};
  \draw[mor](3.55,0.12)--(2.55,0.12) node[ml,midway,above=-1pt]{\(l\theta\)};
  \draw[mor](6.15,0.12)--(7.15,0.12) node[ml,midway,above=-1pt]{\(r\theta\)};
\end{scope}

\begin{pgfonlayer}{background}
  \draw[iopanel] (-0.3,-0.90) rectangle (9.6,0.98);           
  \draw[iopanel] (-0.3,-3.20) rectangle (4.05,-1.28);         
  \draw[iopanel] (4.75,-3.20) rectangle (9.6,-1.28);          
  \draw[iopanel] (-0.3,-5.10) rectangle (9.6,-3.43);          
\end{pgfonlayer}
\end{tikzpicture}}
  \caption{Rule instantiation in double-pushout form
  (Example~\ref{ex:protonation-residual}). \textbf{(A)} Schema, where \(\bot\)
  permits relabeling. \textbf{(B)} Host. \textbf{(C)} The match gives
  \(\theta_\mu(u_\ell)=1\). \textbf{(D)} Concrete span. Application also
  requires \(\Phi_p\) and a \(\Pi\)-valid successor.}
  \label{fig:rule-instantiation}
\end{figure}

\subsection{Direct derivations and well-definedness}
Write \(\theta=\theta_\mu\).  The instantiated rule and host now lie in
\(\mathsf{PAttr}_{\Sigma}(A)\).  A direct derivation is the two-square
construction
\begin{equation}
\begin{tikzcd}[column sep=large,row sep=large]
L\theta \arrow[d,"m"'] & K\theta \arrow[l,"l\theta"'] \arrow[r,"r\theta"] \arrow[d,"k"] &
R\theta \arrow[d,"n"] \\
G & D \arrow[l,"g"] \arrow[r,"h"'] & H .
\end{tikzcd}
\label{eq:resource-constrained-dpo-diagram}
\end{equation}
The left square is a pushout complement that removes
\(L\theta\setminus K\theta\).  The right square is a pushout that adds
\(R\theta\setminus K\theta\).

\begin{definition}
\label{def:admissible-dpo-derivation}
An instantiated exact match \((\theta_\mu,m)\) is \emph{DPO-applicable} when
the identification and dangling conditions hold.  In addition, after deleting
the matched image of \(L\theta_\mu\setminus K\theta_\mu\), the host must contain
none of the simple edges represented by
\(R\theta_\mu\setminus K\theta_\mu\).  This no-created-edge condition prevents
the second pushout from identifying a new edge with an existing one.

Whenever the two squares exist, they define a direct derivation.  It is
\emph{admissible} when \(\Phi_p(G,m,\theta_\mu)\) holds and the complete
successor satisfies \(H\models\Pi\).  Let \(\bar m\) denote the induced map
from persistent rule carriers into \(H\).  Each changed coordinate satisfies
\begin{equation}
  c_H(\bar m(x))
  =c_G(m(x))-d_x^{-}+d_x^{+}.
  \label{eq:resource-constrained-update}
\end{equation}
The policy is checked on \(H\), not only on \(R\theta_\mu\), because preserved
context can affect the validity of a rewritten atom.
\end{definition}

\begin{proposition}
\label{prop:welldef}
Let \(p\) be a resource-constrained LLG rule schema with \(V(K)=V(L)=V(R)\), and
let a resource-admissible structural match \(\mu\) induce
\((\theta_\mu,m)\).  If this exact match is DPO-applicable, then both squares in
diagram~\eqref{eq:resource-constrained-dpo-diagram} exist and are unique up to
isomorphism.  Equation~\eqref{eq:resource-constrained-update} gives each
changed coordinate, every defined interface coordinate is preserved, and all
coordinates outside the rule image agree with \(G\).  Hence the match and its
induced valuation determine \(H\) uniquely up to isomorphism.
\end{proposition}

\begin{proof}
Since \(V(K)=V(L)\), the left square deletes only the edges of
\(L\theta\setminus K\theta\).  Hence dangling is vacuous, while injectivity of
\(m\) gives identification.  For injective \(l\theta\) and \(m\), the natural
pushout complement exists uniquely~\cite{habel2002relabelling}.  It is the host
with the matched deleted edges removed and the changed interface coordinates
left undefined.  Lemma~\ref{lem:evaluation} preserves every defined interface
label.  The no-created-edge condition then makes the second gluing a pushout in
finite simple graphs.  A changed coordinate is undefined throughout \(K\theta\)
and \(D\), so the left square does not transport its value.  The instantiation
fixes \(\theta(u_x)=c_G(m(x))-d_x^{-}\), so \(R\theta\) carries
\(\theta(u_x)+d_x^{+}\), which the second pushout writes into \(H\).  This
proves Equation~\eqref{eq:resource-constrained-update}.  Both squares preserve
the remaining context.
\end{proof}

\section{Electron-flow transitions on LLGs}
\label{sec:arrows}

\subsection{Locus-sorted moves and event groups}

Arrow pushing becomes a multiset of transfers between the four locus sorts
\(\operatorname{lp},\rho,\sigma,\pi\), checked against one pre-state and
committed as a single net edit.  A two-electron move may form a singleton
group, whereas fishhooks on pair-valued loci must be coupled
(Figure~\ref{fig:group-gates}A--C).

\begin{figure}[!htb]
  \centering
  \resizebox{\textwidth}{!}{
\begin{tikzpicture}[
  font=\scriptsize,
  free/.style={font=\tiny, text=ioVermilion!88!black},
  mor/.style={-{Stealth[length=1.8mm]}, line width=0.6pt, draw=black!55},
  ml/.style={font=\tiny\itshape, text=black!68, inner sep=1pt},
  del/.style={line width=0.75pt, draw=black!35, dash pattern=on 1.5pt off 1.5pt},
  ind/.style={-{Stealth[length=2.4mm]}, line width=0.9pt, draw=ioBlue!72!black},
  bad/.style={-{Stealth[length=2.4mm]}, line width=0.9pt, draw=ioVermilion!75!black},
  sub/.style={font=\tiny, text=black!55},
  wrong/.style={font=\tiny, text=ioVermilion!85!black},
  loc/.style={font=\scriptsize, text=black!88},
  ok/.style={font=\scriptsize\bfseries, text=ioGreen!48!black},
  no/.style={font=\scriptsize\bfseries, text=ioVermilion!80!black},
  empty/.style={draw=black!45, circle, inner sep=1.4pt, line width=0.5pt},
]

\begin{scope}[shift={(0,0)}]
  \iosub{-0.15,0.72}{A}{}
  \node[catom](en) at (0.7,-0.1){\aN};
  \node[celp] at ($(en)+(-0.03,0.36)$){}; \node[celp] at ($(en)+(0.17,0.36)$){};
  \node[catom](ec) at (1.8,-0.1){\aC}; \node[cchg] at ($(ec)+(0.24,0.2)$){\chplus};
  \draw[cbond](en)--(ec);
  \draw[cpair] ($(en)+(0.08,0.38)$) to[bend left=24] ($(en)!0.5!(ec)+(0,0.24)$);
  \draw[ind] (2.55,-0.1) -- (3.4,-0.1) node[ml,midway,above=0pt]{induces};
  \node[catom](ln) at (4.05,-0.1){\aN};
  \node[celp] at ($(ln)+(-0.10,0.36)$){}; \node[celp] at ($(ln)+(0.10,0.36)$){};
  \node[catom](lc) at (4.8,-0.1){\aC}; \node[cchg] at ($(lc)+(0.22,0.18)$){\chplus};
  \draw[cbond](ln)--(lc);
  \node[catom](kn) at (6.2,-0.1){\aN}; \node[free] at ($(kn)+(0.03,0.34)$){\(\bot\)};
  \node[catom](kc) at (6.95,-0.1){\aC}; \draw[cbond](kn)--(kc);
  \node[free] at ($(kn)!0.5!(kc)+(0.13,0.34)$){\(\bot\)};
  \node[catom](rn) at (8.35,-0.1){\aN}; \node[cchg] at ($(rn)+(-0.22,0.22)$){\chplus};
  \node[catom](rc) at (9.1,-0.1){\aC};
  \draw[cbond]($(rn)+(0.18,0.05)$)--($(rc)+(-0.18,0.05)$);
  \draw[cbond]($(rn)+(0.18,-0.05)$)--($(rc)+(-0.18,-0.05)$);
  \draw[mor](5.6,0.0)--(5.1,0.0) node[ml,midway,above=-2pt]{\(l\)};
  \draw[mor](7.55,0.0)--(8.05,0.0) node[ml,midway,above=-2pt]{\(r\)};
  \node[sub] at (4.42,-0.56){\(L\)}; \node[sub] at (6.57,-0.56){\(K\)}; \node[sub] at (8.72,-0.56){\(R\)};
\end{scope}

\begin{scope}[shift={(0,-1.55)}]
  \iosub{-0.15,0.72}{B}{}
  \node[catom](fx) at (0.85,-0.1){X}; \node[catom](fy) at (2.05,-0.1){Y};
  \draw[cbond](fx)--(fy);
  \draw[cfish] ($(fx)!0.5!(fy)+(0,0.09)$) to[bend left=46] ($(fx)+(0.02,0.42)$);
  \draw[bad] (2.6,-0.1) -- (3.45,-0.1);
  \node[wrong, font=\normalsize] at (3.85,-0.08){\(\times\)};
  \node[wrong] at (6.2,0.05){\(\Delta\sigma_{xy}=-\tfrac12\notin\mathbb Z\)};
  \node[wrong] at (6.2,-0.42){no valid rule};
\end{scope}

\begin{scope}[shift={(0,-3.10)}]
  \iosub{-0.15,0.72}{C}{}
  \node[catom](gx) at (0.85,-0.1){X}; \node[catom](gy) at (2.05,-0.1){Y};
  \draw[cbond](gx)--(gy);
  \draw[cfish] ($(gx)!0.5!(gy)+(-0.04,0.09)$) to[bend left=48] ($(gx)+(0.02,0.42)$);
  \draw[cfish] ($(gx)!0.5!(gy)+(0.04,0.09)$) to[bend right=48] ($(gy)-(0.02,-0.42)$);
  \draw[ind] (2.6,-0.1) -- (3.45,-0.1) node[ml,midway,above=0pt]{induces};
  \node[catom](clx) at (4.1,-0.1){X}; \node[catom](cly) at (4.95,-0.1){Y};
  \draw[cbond](clx)--(cly);
  \node[catom](ckx) at (6.25,-0.1){X}; \node[catom](cky) at (7.1,-0.1){Y};
  \draw[del](ckx)--(cky);
  \node[free] at ($(ckx)+(0,0.34)$){\(\bot\)}; \node[free] at ($(cky)+(0,0.34)$){\(\bot\)};
  \node[catom](crx) at (8.4,-0.1){X}; \node[crad] at ($(crx)+(0,0.36)$){};
  \node[catom](cry) at (9.25,-0.1){Y}; \node[crad] at ($(cry)+(0,0.36)$){};
  \draw[mor](5.75,0.0)--(5.25,0.0) node[ml,midway,above=-2pt]{\(l\)};
  \draw[mor](7.7,0.0)--(8.2,0.0) node[ml,midway,above=-2pt]{\(r\)};
  \node[sub] at (4.52,-0.44){\(L\)}; \node[sub] at (6.67,-0.44){\(K\)}; \node[sub] at (8.82,-0.44){\(R\)};
\end{scope}

\begin{scope}[shift={(0.55,-4.55)}]
  \iosub{-0.65,0.72}{D}{}
  \node[loc](az) at (0.75,0){\(z\)};
  \node[loc](ax) at (2.15,0){\(x\)};
  \node[loc](ay) at (3.55,0){\(y\)};
  \node[celp] at ($(ax)+(0,0.44)$){};            
  \draw[cfish] ($(az)+(0.2,0.07)$) to[bend left=30] ($(ax)+(-0.2,0.07)$);
  \draw[cfish] ($(ax)+(0.2,0.07)$) to[bend left=30] ($(ay)+(-0.2,0.07)$);
  \node[sub] at (2.15,-0.46){\(c_G(x)=1\)};
  \node[ok] at (2.15,-0.88){\(\checkmark\) commits};
\end{scope}

\begin{scope}[shift={(5.65,-4.55)}]
  \iosub{-0.6,0.72}{E}{}
  \node[loc](bz) at (0.75,0){\(z\)};
  \node[loc](bx) at (2.15,0){\(x\)};
  \node[loc](by) at (3.55,0){\(y\)};
  \node[empty] at ($(bx)+(0,0.44)$){};           
  \draw[cfish] ($(bz)+(0.2,0.07)$) to[bend left=30] ($(bx)+(-0.2,0.07)$);
  \draw[cfish] ($(bx)+(0.2,0.07)$) to[bend left=30] ($(by)+(-0.2,0.07)$);
  \node[sub] at (2.15,-0.46){\(c_G(x)=0\)};
  \node[no] at (2.15,-0.88){\(\times\) rejected};
\end{scope}

\begin{pgfonlayer}{background}
  \draw[iopanel] (-0.3,-0.72) rectangle (9.8,0.86);           
  \draw[iopanel] (-0.3,-2.18) rectangle (9.8,-0.78);          
  \draw[iopanel] (-0.3,-3.73) rectangle (9.8,-2.33);          
  \draw[iopanel] (-0.3,-5.62) rectangle (4.65,-3.83);         
  \draw[iopanel] (4.85,-5.62) rectangle (9.8,-3.83);          
\end{pgfonlayer}
\end{tikzpicture}}
  \caption{Integrality and availability gates for event groups.
  \textbf{(A)} A two-electron move is integral and induces a rule.
  \textbf{(B)} One fishhook is half-integral and does not.
  \textbf{(C)} Coupled fishhooks induce an integral rule.
  \textbf{(D--E)} Equal net deltas can differ in admissibility because gross
  demand is evaluated on the common pre-state.}
  \label{fig:group-gates}
\end{figure}
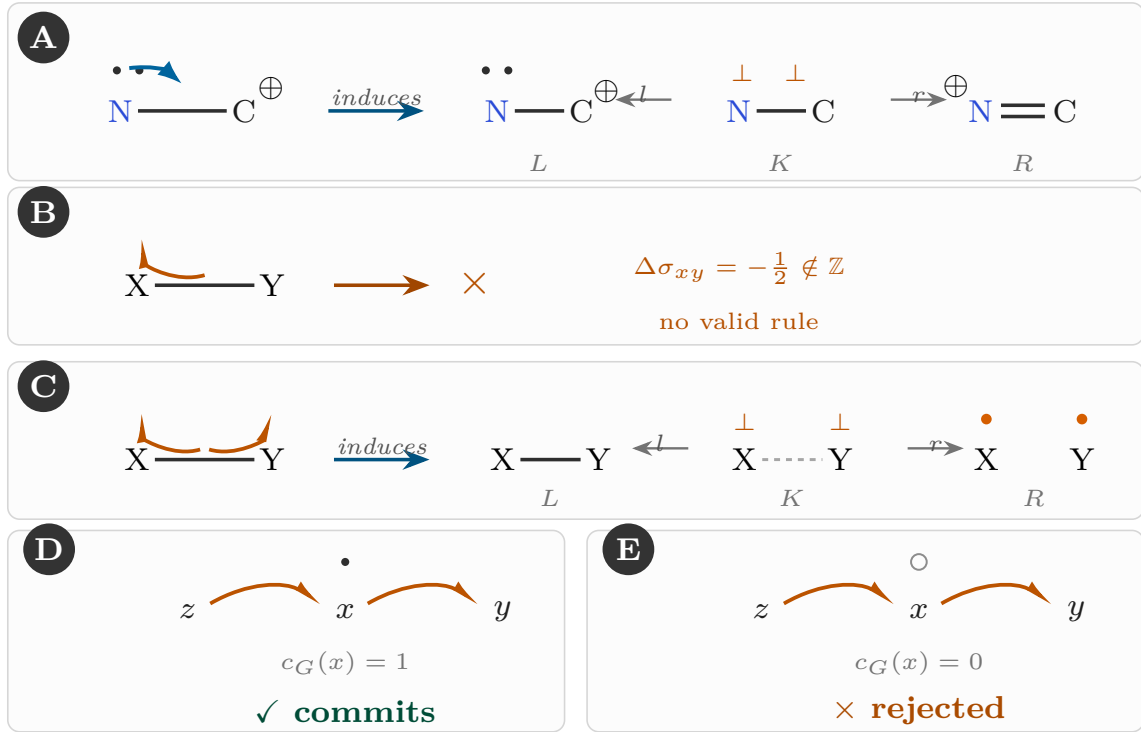

Use the zero extension \(c_G(\sigma_{uv})=c_G(\pi_{uv})=0\) on non-edges, quotas
\(w(\operatorname{lp})=w(\sigma)=w(\pi)=2\) and \(w(\rho)=1\), sort \(t(x)\)
with \(w(x)=w(t(x))\), and display aliases
\(\boldsymbol{\cdot}(v),\sigma(u,v),\pi(u,v)\).  Supports are
\(\operatorname{supp}(\operatorname{lp}(v))=\operatorname{supp}(\rho_v)=\{v\}\)
and \(\operatorname{supp}(\sigma_{uv})=\operatorname{supp}(\pi_{uv})=\{u,v\}\).
The locality predicate \(\Psi_{t(x)\to t(y)}(x,y\mid G)\) requires intersecting
supports.  The move \(\operatorname{lp}(u)\to\operatorname{lp}(v)\) instead requires
\(\{u,v\}\in E(G)\), and named radical macros obey Table~\ref{tab:macros}.

For an integral delta \(\Delta\) with \(c_G+\Delta\ge0\), write
\(G\oplus\Delta\) for the graph with occupancy \(c_G+\Delta\), refreshed
derived labels, and zero-order edges removed.

\begin{definition}
\label{def:moves}
An \emph{electron move} is \(\mu=(x\xrightarrow{k}y)\), where \(x\ne y\) are
loci and \(k\in\{1,2\}\) is the number of transferred electrons.  Its
occupancy delta is
\begin{equation}
  \Delta_{\mu}(z) = \begin{cases} -k / w(x) & \text{if } z = x, \\ +k / w(y) & \text{if } z = y, \\ 0 & \text{otherwise.} \end{cases}
  \label{eq:move-delta-def}
\end{equation}
For a finite multiset \(M\), define its net delta together with its gross
outflow and inflow at each locus by
\begin{equation}
 \Delta_M=\sum_{\mu\in M}\Delta_\mu,
 \quad
 \gout_M(x)=\!\!\sum_{\mu=(x\xrightarrow{k}y)\in M}\!\!\frac{k}{w(x)},
 \quad
 \gin_M(x)=\!\!\sum_{\mu=(y\xrightarrow{k}x)\in M}\!\!\frac{k}{w(x)},
 \label{eq:group-demand}
\end{equation}
so that \(\Delta_M=\gin_M-\gout_M\).  These are group quantities, distinct from
the rule-level amounts \(d_x^{\pm}\) of
Equation~\eqref{eq:residual-rule-label}, which record a net change.  Such a
finite multiset is an \emph{event group}, and it is \emph{admissible} in \(G\)
when (1) every move satisfies $\Psi_{t(x)\to t(y)}(x,y\mid G)$,
(2) $\gout_M(x)\le c_G(x)$ for all loci $x$,
(3) $\Delta_M\in\mathbb Z^{\Loc(G)}$, and (4) $G\oplus\Delta_M$ satisfies $\Pi$.

Condition (2) implies $c_G+\Delta_M \ge 0$, well-defining the atomically committed state $G' = G \oplus \Delta_M$. Because availability is evaluated strictly against the pre-state, resources produced elsewhere in $M$ cannot fund consumption.
\end{definition}

\paragraph{Group boundary and endpoint conformance}
The supplied annotation fixes the group boundary, which the transition relation
neither splits nor merges.  Conditions (1)--(4) attach to it: every move reads
\(G\), and \(\Pi\) is checked once on \(G_M=G\oplus\Delta_M\), never on an
ordering of internal moves.  A singleton polar group may therefore produce a
valid intermediate that a larger group would bypass.

For a declared mapped endpoint \(H_\star\), call \(M\)
\emph{endpoint-conformant} when it is admissible and
\begin{equation}
  G\oplus\Delta_M\cong H_\star
  \quad\text{under the supplied atom correspondence.}
  \label{eq:endpoint-conformance}
\end{equation}
Policy validity is necessary for execution and does not imply conformance,
since it guards the state representation rather than mechanistic plausibility.
Under \(\Pi_{\mathrm{exec}}\) the gap is substantive: \(\kappa_\Pi=\infty\)
admits a hypervalent singleton successor even when that state is an artificial
intermediate, while non-integral or negative occupancies and a \(\pi\)
component without its \(\sigma\) component are still rejected.  For a sequence
\((M_1,\ldots,M_n)\), Equation~\eqref{eq:endpoint-conformance} applies to each
declared intermediate or only to the terminal product, according to the record
schema.

\begin{proposition}
\label{prop:induced-rule}
Let \(M\) be a finite event group in a \(\Pi\)-valid LLG \(G\) with integral
\(\Delta_M\).  On its move supports, form the rule
\(p_M\) whose changed coordinate \(x\) carries \(d_x^-=\max(-\Delta_M(x),0)\) and
\(d_x^+=\max(\Delta_M(x),0)\) in the sense of
Equation~\eqref{eq:residual-rule-label}.  Edge deletion and creation follow
zero-order normalization in \(G\oplus\Delta_M\).  Its application predicate is
\begin{equation}
\begin{aligned}
\Phi_{p_M}\iff{}&
\left[\bigwedge_{\mu=(x\to y)\in M}
\Psi_{t(x)\to t(y)}(x,y\mid G)\right]\\
&\land\left[\gout_M(x)\le c_G(x)\ \text{for every locus }x\right]\\
&\land\left[\text{the declared macro incidence pattern holds}\right].
\end{aligned}
  \label{eq:gross-application-condition}
\end{equation}
The resulting \(p_M\) is the \emph{induced rule}.  Its direct derivation
\(G\derives{(p_M,m)}H\) is admissible exactly when \(M\) satisfies the other
three conditions of Definition~\ref{def:moves}.  Then
\(H\cong G\oplus\Delta_M\) and \(c_H=c_G+\Delta_M\).
\end{proposition}

\noindent
Integrality is required to form \(p_M\): a half-integral \(\Delta_M\) yields no
\(d_x^\pm\in\mathbb N_0\) (Figure~\ref{fig:group-gates}B).

\begin{proof}
\(\Delta_M=d_x^+-d_x^-\), so
Equation~\eqref{eq:resource-constrained-update} gives \(c_H=c_G+\Delta_M\), and
normalization fixes the edge set.  \(\Phi_{p_M}\) is precisely locality, macro
incidence, and common-prestate availability, gross demand included because a
net delta can hide simultaneous consumption and replenishment, while the final
DPO check supplies condition~4.  DPO-applicability adds nothing, since
\(p_M\) creates an edge only where both host bond occupancies vanish and
Proposition~\ref{prop:welldef} covers the rest.
\end{proof}

For \(G\derives{(p,m)}H\), let
\(\delta_{p,m}(x)=c_H(x)-c_G(x)\), extended by zero outside the rule support.
An event group \(M\) \emph{realizes} the derivation when
\(\Delta_M=\delta_{p,m}\).

\begin{proposition}
\label{prop:rule-flow-coherence}
If an admissible event group \(M\) realizes a derivation \(G\derives{(p,m)}H\),
then its induced derivation produces \(H'\cong H\).  The rule \(p_M\) depends
only on the multiset \(M\), produces one successor when all conditions hold,
and none otherwise.
\end{proposition}

\begin{proof}
Both successors have occupancy \(c_G+\Delta_M=c_G+\delta_{p,m}\) and the same
normalized incidence, hence all derived fields agree.
Proposition~\ref{prop:welldef} gives uniqueness, while \(\gout_M,\gin_M,\Delta_M\),
and \(\Phi_{p_M}\) are permutation-invariant.
\end{proof}

\begin{theorem}[Three-level coherence]
\label{thm:coherence}
Let \(G\) be a \(\Pi\)-valid LLG and \(M\) a finite multiset of electron moves
with \(\Delta_M\in\mathbb Z^{\Loc(G)}\).  The following are equivalent.
\begin{enumerate}
\item[(i)] \(M\) is an admissible event group in \(G\)
  (Definition~\ref{def:moves}).
\item[(ii)] The induced rule \(p_M\) admits an admissible direct derivation
  \(G\derives{(p_M,m)}H\) (Definition~\ref{def:admissible-dpo-derivation}).
\item[(iii)] \(G\oplus\Delta_M\) is defined and \(\Pi\)-valid, every move of
  \(M\) is local, and \(\gout_M(x)\le c_G(x)\) at every locus.
\end{enumerate}
When these hold, \(H\cong G\oplus\Delta_M\) is unique up to isomorphism,
order-independent, has \(c_H=c_G+\Delta_M\), and conserves \(\Ie\) and \(Q\).
If any clause of (iii) fails, no successor is produced.
\end{theorem}

\begin{proof}
(i)\(\Leftrightarrow\)(iii) is Definition~\ref{def:moves} with integrality
supplied by hypothesis.  (i)\(\Leftrightarrow\)(ii) and the identification
\(H\cong G\oplus\Delta_M\) are Proposition~\ref{prop:induced-rule}.  Uniqueness
up to isomorphism is Proposition~\ref{prop:welldef}.  Order independence and the
all-or-nothing outcome are Proposition~\ref{prop:rule-flow-coherence}.  The two
conservation identities are Lemma~\ref{lem:group-conservation}.
\end{proof}

\begin{remark}
The three levels are the \emph{same} decision: a mapped reaction supplies
\(\delta_{p,m}\), a rule \(d_x^\pm\), and an annotation \(M\), all on the four
coordinates of Definition~\ref{def:lsg}.  Any disagreement is therefore an
occupancy mismatch at a named locus, grounding the diagnostic codes of
Section~\ref{sec:experiments}.
\end{remark}

In field terms the commit reads \(\ell'_v=\ell_v+\Delta_M(\operatorname{lp}(v))\)
and likewise for \(r_v\), \(s_e\), and \(p_e\), so the successor fields are
integral exactly when \(\Delta_M\) is and nonnegative under condition~(2).
Figure~\ref{fig:group-gates}D--E shows why integrality and common-prestate
availability are separate checks.
Table~\ref{tab:electron-flow-schema} gives the serialized fields.

\subsection{Conservation and atomicity}

\begin{lemma}
\label{lem:group-conservation}
For any event group \(M\),
\[
  \sum_{x \in \Loc(G)} w(x) \Delta_M(x) = 0.
\]
Consequently, any admissible commit conserves the represented electron inventory, \(\Ie(G') = \Ie(G)\), and preserves total formal charge by Proposition~\ref{prop:conservation}.
\end{lemma}

\begin{proof}
Each move contributes
\[
  w(x)\Bigl(-\frac{k}{w(x)}\Bigr) + w(y)\Bigl(+\frac{k}{w(y)}\Bigr) = -k + k = 0.
\]
Summing over \(M\) gives
\(\Ie(G')-\Ie(G)=\sum_x w(x)\Delta_M(x)=0\).  Atom preservation gives charge
conservation by Proposition~\ref{prop:conservation}.
\end{proof}

\begin{remark}
\label{rem:fishhooks}
A single homolytic fishhook gives
\(\Delta(\sigma_{uv})=-\frac12\) and is rejected.  The coupled pair
\[
  \sigma(u,v) \xrightarrow{1e^-} \boldsymbol{\cdot}(u), \quad \sigma(u,v) \xrightarrow{1e^-} \boldsymbol{\cdot}(v)
\]
gives \(\Delta(\sigma_{uv})=-1\) and
\(\Delta(\rho_u)=\Delta(\rho_v)=1\), so it is integral.
\end{remark}

\begin{example}
\label{ex:beta-scission-group}
The \(\beta\)-scission of Example~\ref{ex:beta-scission-lsg} is
\[
  \boldsymbol{\cdot}(o) \xrightarrow{1e^-} \pi(o,c_1), \quad \sigma(c_1,c_2) \xrightarrow{1e^-} \pi(o,c_1), \quad \sigma(c_1,c_2) \xrightarrow{1e^-} \boldsymbol{\cdot}(c_2).
\]
Its nonzero delta is
\(\Delta(\rho_o,\pi_{o,c_1},\sigma_{c_1,c_2},\rho_{c_2})
=(-1,1,-1,1)\).  The full group is admissible, whereas every nonempty proper
subset is half-integral.
\end{example}

\begin{remark}
\label{rem:composition}
Proposition~\ref{prop:induced-rule} turns each admissible event group \(M_i\)
into a rule \(G_{i-1}\derives{p_{M_i}}G_i\), so a mechanism
\((M_1,\dots,M_n)\) is the sequence
\(G_0\derives{p_{M_1}}G_1\Longrightarrow\cdots\derives{p_{M_n}}G_n\).
Its net change is \(\sum_i\Delta_{M_i}\), with \(\Ie\) conserved at each step.
This is weaker than DPO rule composition, which also preserves matches,
dependencies, structural edits, and conditions~\cite{Andersen2018RuleComposition}.

For \(S_N2\) (Figure~\ref{fig:mechanism-composition}), one group containing bond
formation and chloride expulsion commits without an internal intermediate,
whereas two singleton groups make the pentavalent state an explicit successor
of the first.  An octet-bounded \(\Pi\) rejects it and the permissive
\(\Pi_{\mathrm{exec}}\) of Equation~\eqref{eq:pi-exec} may admit it, so
atomicity does not choose the group boundary and endpoint conformance, not
policy validity, separates that successor from the declared product.  Building
the grouped composite as a \emph{rule} needs the concurrency theorem of
Section~\ref{sec:conclusion}.
\end{remark}

\subsection{A locus-sorted grammar for event groups}
\label{sec:grammar}

Admissibility gives algebraic consistency.  The grammar organizes two-electron
moves into polar classes and fishhook groups into radical macros.

\paragraph{Polar curly-arrow classes}
A polar class is the ordered source--target sort pair
\((t_i,t_j)\in\{\operatorname{lp},\sigma,\pi\}^2\) of one two-electron arrow.
The bond-free diagonal \(\operatorname{lp}\to\operatorname{lp}\) is excluded:
moving a whole pair between adjacent atoms changes both charges by two units
without editing a bond and has no elementary realization.  Its one-electron
surface shape reappears only in the alpha-resonance macro below.  The remaining
eight bond-coupled pairs (Table~\ref{tab:polar-classes},
Figure~\ref{fig:polar-grammar}) exhaust the polar classes and recast the
\texttt{PMechDB} taxonomy~\cite{Tavakoli2024}.  They classify arrows, not event
groups.  Their \(k=2\) deltas are integral, but locality, gross availability,
and group-successor validity remain independent.

\begin{figure*}[htbp]
  \centering
  \resizebox{\textwidth}{!}{
\begin{tikzpicture}[
  font=\Large,
  iocbadge/.append style={font=\bfseries\normalsize},
  iotitle/.append style={font=\bfseries\large},
  mov/.style={font=\normalsize, align=center, text=black!70, text width=40mm}
]
\def\cW{5.0}   
\def\rH{3.35}  

\begin{scope}[shift={(0,0)}]
  \iosub{-2.05,1.55}{C1}{$\operatorname{lp}\!\to\!\pi$}
  \node[catom] (n) at (-0.85,0.15) {\aN};
  \node[catom] (c) at ( 0.75,0.15) {\aC};
  \node[cchg] at ($(c)+(0.32,0.32)$) {\chplus};
  \node[celp] at ($(n)+(-0.07,0.42)$){}; \node[celp] at ($(n)+(0.13,0.42)$){};
  \draw[cbond] (n)--(c);
  \draw[cpair] ($(n)+(0.03,0.44)$) to[bend left=22] ($(n)!0.5!(c)+(0,0.34)$);
  \node[mov] at (0,-1.05)
    {mesomeric donation};
\end{scope}

\begin{scope}[shift={(\cW,0)}]
  \iosub{-2.05,1.55}{C2}{$\operatorname{lp}\!\to\!\sigma$}
  \node[catom] (b) at (-1.4,0.15) {B};
  \node[cchg] at ($(b)+(0.30,0.34)$) {\chminus};
  \node[celp] at ($(b)+(-0.07,0.42)$){}; \node[celp] at ($(b)+(0.13,0.42)$){};
  \node[catom] (h) at (0.2,0.15) {\aH};
  \node[catom] (a) at (1.4,0.15) {A};
  \draw[cbond] (h)--(a);
  \draw[cpair] ($(b)+(0.06,0.44)$) to[bend left=20] ($(b)!0.68!(h)+(0,0.28)$);
  \node[mov] at (0,-1.05)
    {$S_N2$ / proton transfer};
\end{scope}

\begin{scope}[shift={(2*\cW,0)}]
  \iosub{-2.05,1.55}{C3}{$\pi\!\to\!\operatorname{lp}$}
  \node[catom] (c) at (-0.45,0.15) {\aC};
  \node[catom] (o) at ( 0.95,0.15) {\aO};
  \draw[cbond] ($(c)+(0.22,0.07)$)--($(o)+(-0.20,0.07)$);
  \draw[cbond] ($(c)+(0.22,-0.07)$)--($(o)+(-0.20,-0.07)$);
  \node[celp] at ($(o)+(0.30,0.11)$){}; \node[celp] at ($(o)+(0.30,-0.11)$){};
  \draw[cpair] ($(c)!0.5!(o)+(0,0.22)$) to[bend left=26] ($(o)+(0.10,0.52)$);
  \node[celp] at ($(o)+(-0.04,0.60)$){}; \node[celp] at ($(o)+(0.18,0.62)$){};
  \node[cchg] at ($(c)+(-0.04,0.54)$) {\chplus};
  \node[cchg] at ($(o)+(0.40,0.46)$) {\chminus};
  \node[mov] at (0,-1.05)
    {carbonyl ionization};
\end{scope}

\begin{scope}[shift={(3*\cW,0)}]
  \iosub{-2.05,1.55}{C4}{$\pi\!\to\!\pi$}
  \node[catom] (c1) at (-1.35,0.15) {\aC};
  \node[catom] (c2) at (-0.05,0.15) {\aC};
  \node[catom] (c3) at ( 1.25,0.15) {\aC};
  \draw[cbond] ($(c1)+(0.22,0.07)$)--($(c2)+(-0.22,0.07)$);
  \draw[cbond] ($(c1)+(0.22,-0.07)$)--($(c2)+(-0.22,-0.07)$);
  \draw[cbond] (c2)--(c3);
  \node[cchg] at ($(c3)+(0.30,0.32)$) {\chplus};
  \draw[cpair] ($(c1)!0.5!(c2)+(0,0.24)$) to[bend left=24] ($(c2)!0.5!(c3)+(0,0.24)$);
  \node[mov] at (0,-1.05)
    {allylic resonance};
\end{scope}

\begin{scope}[shift={(0,-\rH)}]
  \iosub{-2.05,1.55}{C5}{$\pi\!\to\!\sigma$}
  \node[catom] (c1) at (-1.55,0.15) {\aC};
  \node[catom] (c2) at (-0.45,0.15) {\aC};
  \draw[cbond] ($(c1)+(0.22,0.07)$)--($(c2)+(-0.22,0.07)$);
  \draw[cbond] ($(c1)+(0.22,-0.07)$)--($(c2)+(-0.22,-0.07)$);
  \node[catom] (br1) at (0.8,0.55) {\aBr};
  \node[catom] (br2) at (1.7,-0.18) {\aBr};
  \draw[cbond] (br1)--(br2);
  \draw[cpair] ($(c1)!0.5!(c2)+(0,0.22)$) to[bend left=16] ($(br1)+(-0.18,-0.16)$);
  \node[mov] at (0,-1.05)
    {electrophilic addition};
\end{scope}

\begin{scope}[shift={(\cW,-\rH)}]
  \iosub{-2.05,1.55}{C6}{$\sigma\!\to\!\operatorname{lp}$}
  \node[catom] (c) at (-0.7,0.15) {\aC};
  \node[catom] (br) at (0.9,0.15) {\aBr};
  \draw[cbond] (c)--(br);
  \node[cchg] at ($(c)+(-0.05,0.42)$) {\chplus};
  \node[celp] at ($(br)+(0.34,0.11)$){}; \node[celp] at ($(br)+(0.34,-0.11)$){};
  \draw[cpair] ($(c)!0.55!(br)+(0,0.14)$) to[bend left=24] ($(br)+(0.05,0.54)$);
  \node[celp] at ($(br)+(-0.06,0.62)$){}; \node[celp] at ($(br)+(0.16,0.62)$){};
  \node[cchg] at ($(br)+(0.40,0.46)$) {\chminus};
  \node[mov] at (0,-1.05)
    {ionization / E1};
\end{scope}

\begin{scope}[shift={(2*\cW,-\rH)}]
  \iosub{-2.05,1.55}{C7}{$\sigma\!\to\!\pi$}
  \node[catom] (h)  at (-1.7,0.78) {\aH};
  \node[catom] (c1) at (-0.7,0.15) {\aC};
  \node[catom] (c2) at ( 0.55,0.15) {\aC};
  \node[catom] (lg) at ( 1.7,0.15) {LG};
  \draw[cbond] (h)--(c1); \draw[cbond] (c1)--(c2); \draw[cbond] (c2)--(lg);
  \draw[cpair] ($(h)!0.5!(c1)+(0.02,0.05)$) to[bend left=30] ($(c1)!0.5!(c2)+(0,0.30)$);
  \node[mov] at (0,-1.05)
    {$\beta$-elimination (E2)};
\end{scope}

\begin{scope}[shift={(3*\cW,-\rH)}]
  \iosub{-2.05,1.55}{C8}{$\sigma\!\to\!\sigma$}
  \node[catom] (h)  at (-0.75,0.95) {\aH};
  \node[catom] (c1) at (-0.75,0.12) {\aC};
  \node[catom] (c2) at ( 0.9,0.12) {\aC};
  \node[cchg] at ($(c2)+(0.32,0.32)$) {\chplus};
  \draw[cbond] (h)--(c1); \draw[cbond] (c1)--(c2);
  \draw[cpair] ($(h)!0.55!(c1)+(0.06,0)$) to[bend left=44] ($(c2)+(-0.02,0.42)$);
  \node[mov] at (0,-1.05)
    {1,2-shift};
\end{scope}

\begin{pgfonlayer}{background}
  \foreach \i in {0,1,2,3}{
    \draw[iopanel] ($(\i*\cW,0)+(-2.05,-1.42)$) rectangle ($(\i*\cW,0)+(2.15,1.85)$);
    \draw[iopanel] ($(\i*\cW,-\rH)+(-2.05,-1.42)$) rectangle ($(\i*\cW,-\rH)+(2.15,1.85)$);
  }
\end{pgfonlayer}
\end{tikzpicture}}
  \caption{The eight bond-coupled polar transitions of Table~\ref{tab:polar-classes}, each the single fundamental
  source\(\to\)target move on its prototype, labeled by arrow pattern
  \(t_i\!\to\!t_j\) and a representative step.  Blue curly arrows are
  two-electron moves, filled dots lone pairs, and \(\oplus/\ominus\) the exposed
  charges.  Heteroatoms use CDK element colors, with placeholders for a
  Br\o{}nsted base (B), acid (A), and leaving group (LG).  The gallery is
  complete: every polar class has a bond endpoint.}
  \label{fig:polar-grammar}
\end{figure*}
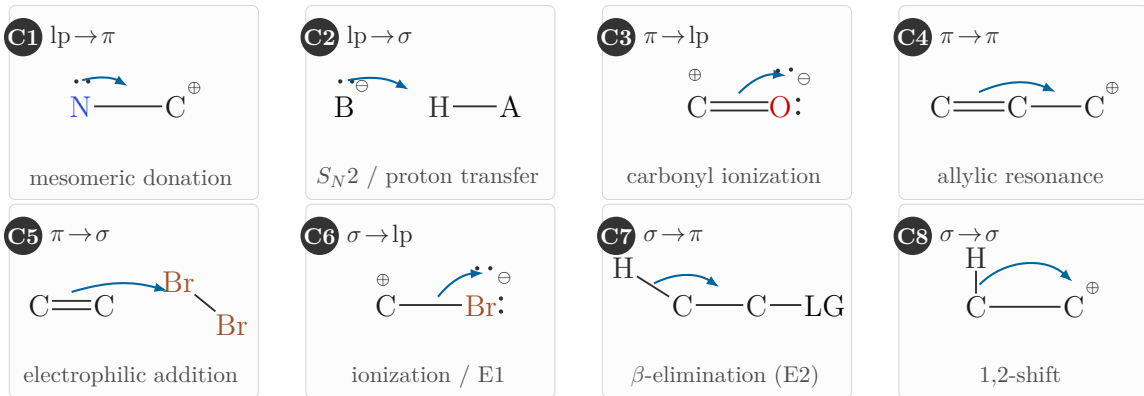

\paragraph{Radical fishhook macros}
Figure~\ref{fig:radical-grammar} and Table~\ref{tab:macros}
define six bond-centered fishhook groups and one alpha-resonance
relocation covering the seven \texttt{RMechDB} classes~\cite{Tavakoli2023}.
They are a template library, not a decomposition of every admissible group.
Each bond-centered macro couples at least two fishhooks into an integral delta.
Homolysis/recombination and radical addition/\(\beta\)-scission are inverse
pairs.  The other three macros are self-inverse, enabling reverse replay without
duplicate templates.  Hydrogen transfer uses an explicit mapped H vertex.

Alpha resonance (M7) is the one macro that is a \emph{declared shorthand} rather
than a move multiset.  Its compact source annotation sends one electron from a
lone-pair donor \(d\) to a radical acceptor \(a\), but literal
\(\operatorname{lp}(d)\to\operatorname{lp}(a)\) would be half-integral.  The
executable meaning is the three-move group
\begin{equation}
 M_{d\rightsquigarrow a}=\bigl\{\,
   \operatorname{lp}(d)\xrightarrow{1}\rho_d,\quad
   \operatorname{lp}(d)\xrightarrow{1}\operatorname{lp}(a),\quad
   \rho_a\xrightarrow{1}\operatorname{lp}(a)
 \,\bigr\},
 \label{eq:lp-radical-moves}
\end{equation}
whose net delta is integral,
\begin{equation}
 \Delta_{M_{d\rightsquigarrow a}}
 =-\mathbf 1_{\operatorname{lp}(d)}+\mathbf 1_{\rho_d}
  -\mathbf 1_{\rho_a}+\mathbf 1_{\operatorname{lp}(a)}.
 \label{eq:lp-radical-expansion}
\end{equation}
The two departures from \(\operatorname{lp}(d)\) sum to \(-1\) and the two
arrivals at \(\operatorname{lp}(a)\) to \(+1\), with weighted sum
\(-2+1-1+2=0\).  The pre-state must have \(\ell_d,r_a\ge1\) and
\(\{d,a\}\in E(G)\), and execution commits this expansion rather than the
surface arrow.  M7 therefore reuses the excluded
\(\operatorname{lp}\to\operatorname{lp}\) shape only as one-electron notation,
so the ``eight plus seven'' templates do not double-count a move class.

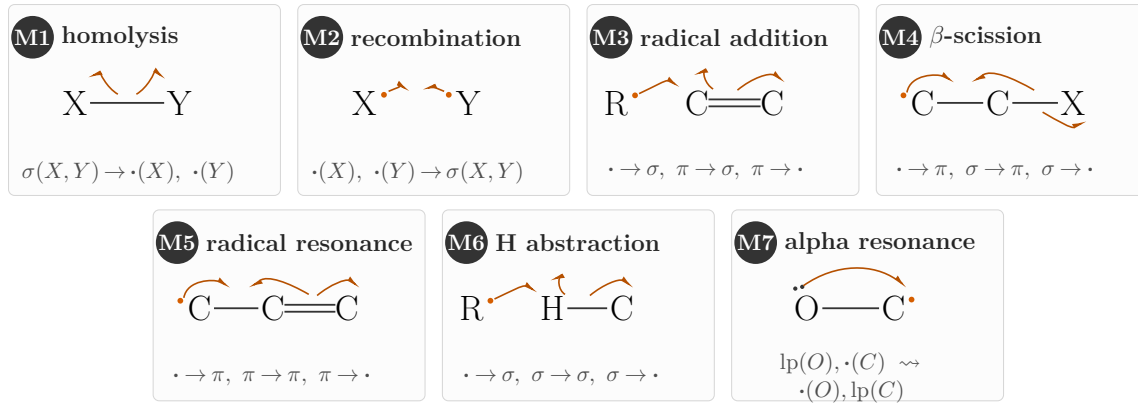
\begin{figure*}[htbp]
  \centering
  \resizebox{\textwidth}{!}{
\begin{tikzpicture}[
  font=\Large,
  iocbadge/.append style={font=\bfseries\normalsize},
  iotitle/.append style={font=\bfseries\large},
  atom/.style={font=\LARGE, inner sep=1.2pt},
  bond/.style={line width=0.9pt, black!80},
  fish/.style={-{Latex[harpoon,length=2.0mm]}, thick, ioVermilion!85!black}, 
  rad/.style={circle, fill=ioVermilion, inner sep=1.0pt},     
  mov/.style={font=\normalsize, align=center, text=black!70, text width=40mm}
]
\def\cW{5.0}   
\def\rH{3.55}  

\begin{scope}[shift={(0,0)}]
  \iosub{-1.90,1.55}{M1}{homolysis}
  \node[atom] (x) at (-0.9,0.15) {X};
  \node[atom] (y) at ( 0.9,0.15) {Y};
  \draw[bond] (x)--(y);
  \draw[fish] (-0.16,0.30) to[out=150,in=30]  ($(x)+(0.22,0.42)$);
  \draw[fish] ( 0.16,0.30) to[out=30,in=150]  ($(y)+(-0.22,0.42)$);
  \node[mov] at (0,-1.05)
    {$\sigma(X,Y)\!\to\!\boldsymbol{\cdot}(X),\ \boldsymbol{\cdot}(Y)$};
\end{scope}

\begin{scope}[shift={(\cW,0)}]
  \iosub{-1.90,1.55}{M2}{recombination}
  \node[atom] (x) at (-0.9,0.15) {X};
  \node[atom] (y) at ( 0.9,0.15) {Y};
  \node[rad] at ($(x)+(0.34,0.14)$) {};
  \node[rad] at ($(y)+(-0.34,0.14)$) {};
  \draw[fish] ($(x)+(0.42,0.20)$) to[out=25,in=155] (-0.10,0.36);
  \draw[fish] ($(y)+(-0.42,0.20)$) to[out=155,in=25] ( 0.10,0.36);
  \node[mov] at (0,-1.05)
    {$\boldsymbol{\cdot}(X),\ \boldsymbol{\cdot}(Y)\!\to\!\sigma(X,Y)$};
\end{scope}

\begin{scope}[shift={(2*\cW,0)}]
  \iosub{-1.90,1.55}{M3}{radical addition}
  \node[atom] (r)  at (-1.55,0.15) {R};
  \node[atom] (c1) at (-0.15,0.15) {C};
  \node[atom] (c2) at ( 1.15,0.15) {C};
  \node[rad] at ($(r)+(0.33,0.13)$) {};
  \draw[bond] ($(c1)+(0.2, 0.06)$)--($(c2)+(-0.2, 0.06)$);
  \draw[bond] ($(c1)+(0.2,-0.06)$)--($(c2)+(-0.2,-0.06)$);
  \draw[fish] ($(r)+(0.40,0.16)$)  to[out=25,in=155]  ($(c1)+(-0.28,0.28)$); 
  \draw[fish] ($(c1)+(0.28,0.24)$) to[out=150,in=30]  ($(c1)+(-0.06,0.46)$); 
  \draw[fish] ($(c1)!0.5!(c2)+(0.05,0.22)$) to[out=40,in=150] ($(c2)+(0.24,0.36)$); 
  \node[mov] at (0,-1.05)
    {$\boldsymbol{\cdot}\!\to\!\sigma,\ \pi\!\to\!\sigma,\ \pi\!\to\!\boldsymbol{\cdot}$};
\end{scope}

\begin{scope}[shift={(3*\cW,0)}]
  \iosub{-1.90,1.55}{M4}{$\beta$-scission}
  \node[atom] (c1) at (-1.25,0.15) {C};
  \node[atom] (c2) at ( 0.05,0.15) {C};
  \node[atom] (x)  at ( 1.35,0.15) {X};
  \node[rad] at ($(c1)+(-0.33,0.13)$) {};
  \draw[bond] (c1)--(c2);
  \draw[bond] (c2)--(x);
  \draw[fish] ($(c1)+(-0.27,0.18)$) to[out=80,in=155] ($(c1)!0.45!(c2)+(-0.02,0.36)$); 
  \draw[fish] ($(c2)!0.5!(x)+(-0.02,0.22)$) to[out=150,in=35] ($(c1)!0.6!(c2)+(0.06,0.34)$); 
  \draw[fish] ($(c2)!0.55!(x)+(0.06,-0.16)$) to[out=-30,in=210] ($(x)+(0.24,-0.30)$); 
  \node[mov] at (0,-1.05)
    {$\boldsymbol{\cdot}\!\to\!\pi,\ \sigma\!\to\!\pi,\ \sigma\!\to\!\boldsymbol{\cdot}$};
\end{scope}

\begin{scope}[shift={(0.5*\cW,-\rH)}]
  \iosub{-1.90,1.55}{M5}{radical resonance}
  \node[atom] (c1) at (-1.25,0.15) {C};
  \node[atom] (c2) at ( 0.02,0.15) {C};
  \node[atom] (c3) at ( 1.3,0.15) {C};
  \node[rad] at ($(c1)+(-0.33,0.13)$) {};
  \draw[bond] (c1)--(c2);
  \draw[bond] ($(c2)+(0.2, 0.06)$)--($(c3)+(-0.2, 0.06)$);
  \draw[bond] ($(c2)+(0.2,-0.06)$)--($(c3)+(-0.2,-0.06)$);
  \draw[fish] ($(c1)+(-0.27,0.18)$) to[out=80,in=155] ($(c1)!0.45!(c2)+(-0.02,0.36)$); 
  \draw[fish] ($(c2)!0.5!(c3)+(-0.02,0.24)$) to[out=150,in=35] ($(c1)!0.62!(c2)+(0.06,0.34)$); 
  \draw[fish] ($(c2)!0.55!(c3)+(0.06,0.22)$) to[out=40,in=150] ($(c3)+(0.24,0.36)$); 
  \node[mov] at (0,-1.05)
    {$\boldsymbol{\cdot}\!\to\!\pi,\ \pi\!\to\!\pi,\ \pi\!\to\!\boldsymbol{\cdot}$};
\end{scope}

\begin{scope}[shift={(1.5*\cW,-\rH)}]
  \iosub{-1.90,1.55}{M6}{H abstraction}
  \node[atom] (r) at (-1.55,0.15) {R};
  \node[atom] (h) at (-0.15,0.15) {H};
  \node[atom] (c) at ( 1.05,0.15) {C};
  \node[rad] at ($(r)+(0.33,0.13)$) {};
  \draw[bond] (h)--(c);
  \draw[fish] ($(r)+(0.40,0.16)$)  to[out=25,in=155] ($(h)+(-0.24,0.28)$); 
  \draw[fish] ($(h)+(0.22,0.24)$)  to[out=150,in=30] ($(h)+(-0.04,0.46)$); 
  \draw[fish] ($(h)!0.5!(c)+(0.05,0.22)$) to[out=40,in=150] ($(c)+(0.22,0.36)$); 
  \node[mov] at (0,-1.05)
    {$\boldsymbol{\cdot}\!\to\!\sigma,\ \sigma\!\to\!\sigma,\ \sigma\!\to\!\boldsymbol{\cdot}$};
\end{scope}

\begin{scope}[shift={(2.5*\cW,-\rH)}]
  \iosub{-1.90,1.55}{M7}{alpha resonance}
  \node[atom] (o) at (-0.75,0.15) {O};
  \node[atom] (c) at ( 0.75,0.15) {C};
  \draw[bond] (o)--(c);
  \fill[black!80] ($(o)+(-0.20,0.34)$) circle (1.05pt);
  \fill[black!80] ($(o)+(-0.08,0.37)$) circle (1.05pt);
  \node[rad] at ($(c)+(0.31,0.15)$) {};
  \draw[fish] ($(o)+(-0.06,0.39)$) to[out=35,in=145] ($(c)+(0.24,0.36)$);
  \node[mov] at (0,-1.05)
    {$\operatorname{lp}(O),\boldsymbol{\cdot}(C)
      \rightsquigarrow\boldsymbol{\cdot}(O),\operatorname{lp}(C)$};
\end{scope}

\begin{pgfonlayer}{background}
  \foreach \i in {0,1,2,3}{
    \draw[iopanel] ($(\i*\cW,0)+(-2.05,-1.45)$) rectangle ($(\i*\cW,0)+(2.65,1.85)$);
  }
  \foreach \i in {0.5,1.5,2.5}{
    \draw[iopanel] ($(\i*\cW,-\rH)+(-2.05,-1.45)$) rectangle ($(\i*\cW,-\rH)+(2.65,1.85)$);
  }
\end{pgfonlayer}
\end{tikzpicture}}
  \caption{The seven radical macros of Table~\ref{tab:macros} (M1--M7), as minimal arrow-pushing schematics with
  their locus-sorted move multisets.  M1--M6 are bond-centered coupled groups,
  each needing at least two fishhooks that commit as one integral group
  (Remark~\ref{rem:fishhooks}).  M7 is the explicitly expanded alpha-resonance
  relocation, a shorthand with integral four-coordinate delta
  Equation~\eqref{eq:lp-radical-expansion} that does not edit the bond.
  Vermilion single-barbed fishhooks are \(1e^-\) moves, vermilion dots
  radicals.  R is a radical partner and X, Y, C bond atoms.  Representative
  steps: M1 peroxide O--O homolysis, M2 radical--radical coupling, M3 alkene
  propagation, M4 alkoxyl \(\beta\)-scission, M5 allyl resonance, M6
  hydrogen-atom transfer, M7 alpha resonance.  Class names follow the
  \texttt{RMechDB} taxonomy \cite{Tavakoli2023}.  The drawings and executable
  semantics are ours, and the involution pairing is the table's last column.}
  \label{fig:radical-grammar}
\end{figure*}

\subsection{Formal relation to half-edge and electron-feature graphs}
\label{sec:efg-llg-relation}

HEGs retain electron identity, EFGs retain anonymous feature objects, and LLGs
store feature-fiber cardinalities while distinguishing \(\sigma\) from \(\pi\)
\cite{holzschuh2026halfedge}.  Refining \textsf{Pair} into
\(\textsf{Pair}_{\sigma},\textsf{Pair}_{\pi}\) and restricting to fiber-bijective
morphisms gives a span
\(\mathbf{EFG}\xleftarrow{F}\mathbf{EFG}^{\mathrm{fib},\Pi}_{\sigma/\pi}
\xrightarrow{K}\mathbf{LLG}_{\Pi}^{=}\)
in which \(F\) forgets bond kind and \(K\) forgets feature identity.  Because they
discard incomparable information, neither target refines the other and the
\(\sigma/\pi\) split must precede condensation.  This makes precise the two
differences stated in Section~\ref{sec:related}, the determinate \(\pi_e\)
transition locus and the executable rather than existential reading of an
annotation.  Section~\ref{si:efg-llg} gives the correspondence table, functor
details, and a worked condensation.

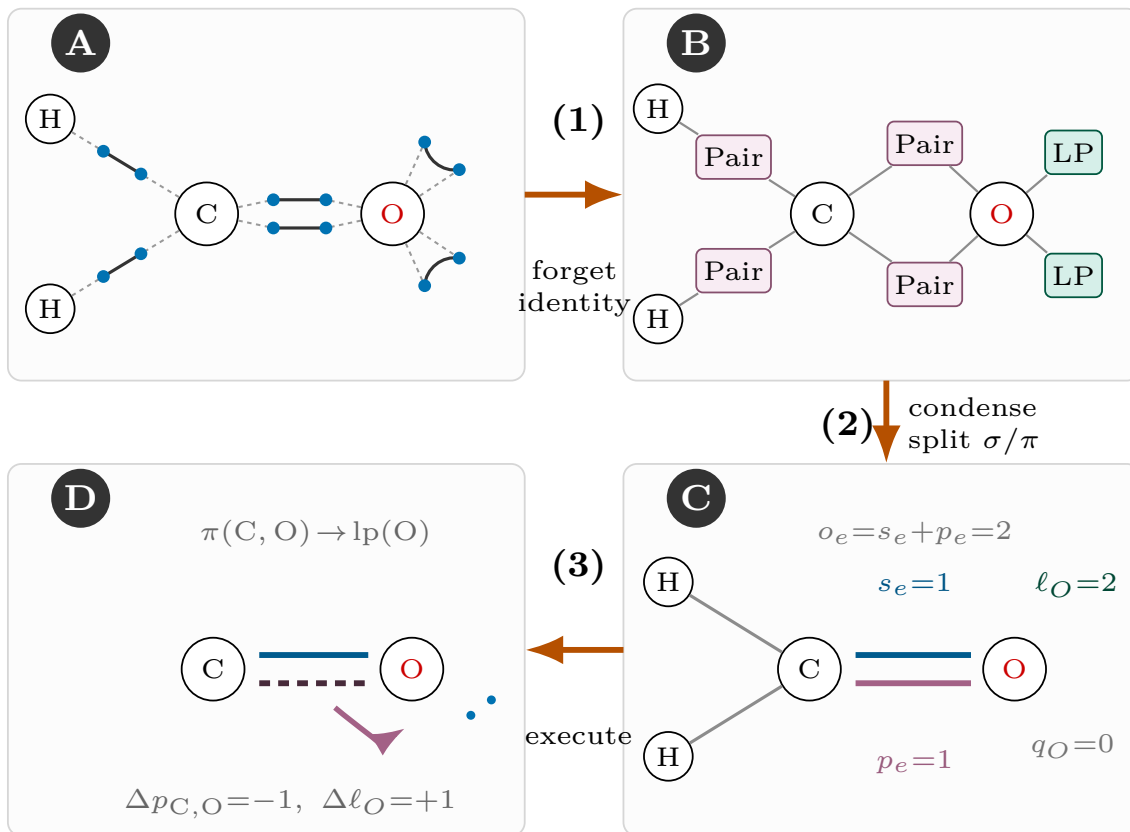
\begin{figure*}[htbp]
  \centering
  \resizebox{\textwidth}{!}{
\begin{tikzpicture}[
  font=\tiny,
  atom/.style={circle, draw, fill=white, minimum size=4.9mm, inner sep=0.4pt},
  hatom/.style={circle, draw, fill=white, minimum size=3.8mm, inner sep=0.4pt,
                font=\tiny},
  bond/.style={line width=0.7pt, black!80},
  eh/.style={circle, fill=ioBlue, inner sep=1.0pt},
  lpdot/.style={circle, fill=ioBlue, inner sep=0.7pt},
  featbox/.style={draw, rounded corners=1.5pt, inner sep=1.6pt, font=\tiny,
                  minimum height=3.4mm},
  pairbox/.style={featbox, fill=ioPurple!14, draw=ioPurple!65!black},
  lpbox/.style={featbox, fill=ioGreen!16, draw=ioGreen!55!black},
  fld/.style={font=\tiny},
  arr/.style={-{Latex[length=3mm]}, line width=1.3pt, ioVermilion!85!black},
  move/.style={-{Latex[length=2.6mm]}, very thick, ioPurple!80!black},
  tie/.style={line width=0.5pt, black!45},
  psi/.style={line width=0.45pt, black!40, dash pattern=on 1pt off 1pt}
]
\def\ax{0}    \def\ay{3.55}
\def\bx{4.80} \def\by{3.55}
\def\cx{4.80} \def\cy{0}
\def\dx{0}    \def\dy{0}

\begin{scope}[shift={(\ax,\ay)}]
  \iosub{-1.65,1.50}{A}{}
  \node[atom] (aC) at (-0.5,0) {C};
  \node[atom] (aO) at ( 0.95,0) {\textcolor{cdkO}{O}};
  \node[hatom] (aH1) at (-1.72, 0.74) {H};
  \node[hatom] (aH2) at (-1.72,-0.74) {H};
  \foreach \Hy in {aH1,aH2}{
    \coordinate (p) at ($(aC)!0.42!(\Hy)$);
    \coordinate (q) at ($(aC)!0.66!(\Hy)$);
    \draw[psi] (aC)--(p); \draw[psi] (\Hy)--(q);
    \draw[bond] (p)--(q);
    \node[eh] at (p) {}; \node[eh] at (q) {};
  }
  \foreach \dyy in {0.11,-0.11}{
    \coordinate (ca) at ($(aC)+(0.52,\dyy)$);
    \coordinate (oa) at ($(aO)+(-0.52,\dyy)$);
    \draw[psi] (aC)--(ca); \draw[psi] (aO)--(oa);
    \draw[bond] (ca)--(oa);
    \node[eh] at (ca) {}; \node[eh] at (oa) {};
  }
  \coordinate (aLPu1) at ($(aO)+(0.520,0.345)$);
  \coordinate (aLPu2) at ($(aO)+(0.250,0.560)$);
  \draw[psi] (aO)--(aLPu1); \draw[psi] (aO)--(aLPu2);
  \draw[bond] (aLPu1) to[bend left=48] (aLPu2);
  \node[eh] at (aLPu1) {}; \node[eh] at (aLPu2) {};
  \coordinate (aLPd1) at ($(aO)+(0.520,-0.345)$);
  \coordinate (aLPd2) at ($(aO)+(0.250,-0.560)$);
  \draw[psi] (aO)--(aLPd1); \draw[psi] (aO)--(aLPd2);
  \draw[bond] (aLPd1) to[bend right=48] (aLPd2);
  \node[eh] at (aLPd1) {}; \node[eh] at (aLPd2) {};
\end{scope}

\begin{scope}[shift={(\bx,\by)}]
  \iosub{-1.65,1.50}{B}{}
  \node[atom] (bC) at (-0.5,0) {C};
  \node[atom] (bO) at ( 0.90,0) {\textcolor{cdkO}{O}};
  \node[hatom] (bH1) at (-1.78, 0.82) {H};
  \node[hatom] (bH2) at (-1.78,-0.82) {H};
  \node[pairbox] (bPH1) at ($(bC)!0.54!(bH1)$) {Pair};
  \node[pairbox] (bPH2) at ($(bC)!0.54!(bH2)$) {Pair};
  \draw[tie] (bC)--(bPH1); \draw[tie] (bPH1)--(bH1);
  \draw[tie] (bC)--(bPH2); \draw[tie] (bPH2)--(bH2);
  \node[pairbox] (bP1) at (0.3,0.55) {Pair};
  \node[pairbox] (bP2) at (0.3,-0.55) {Pair};
  \draw[tie] (bC)--(bP1); \draw[tie] (bP1)--(bO);
  \draw[tie] (bC)--(bP2); \draw[tie] (bP2)--(bO);
  \node[lpbox] (bL1) at (1.46,0.48) {LP};
  \node[lpbox] (bL2) at (1.46,-0.48) {LP};
  \draw[tie] (bO)--(bL1); \draw[tie] (bO)--(bL2);
\end{scope}

\begin{scope}[shift={(\cx,\cy)}]
  \iosub{-1.65,1.50}{C}{}
  \node[atom] (cC) at (-0.6,0) {C};
  \node[atom] (cO) at ( 1.00,0) {\textcolor{cdkO}{O}};
  \node[hatom] (cH1) at (-1.7, 0.68) {H};
  \node[hatom] (cH2) at (-1.7,-0.68) {H};
  \draw[bond,black!45] (cC)--(cH1); \draw[bond,black!45] (cC)--(cH2);
  \draw[line width=1.3pt,ioBlue!80!black]   ($(cC)+(0.36,0.11)$)--($(cO)+(-0.34,0.11)$);
  \draw[line width=1.3pt,ioPurple!80!black] ($(cC)+(0.36,-0.11)$)--($(cO)+(-0.34,-0.11)$);
  \node[fld,ioBlue!75!black]   at (0.22,0.66) {$s_e{=}1$};
  \node[fld,ioPurple!75!black] at (0.22,-0.72) {$p_e{=}1$};
  \node[fld,ioGreen!45!black] at (1.48,0.66) {$\ell_O{=}2$};
  \node[fld,black!55] at (1.45,-0.6) {$q_O{=}0$};
  \node[fld,black!55] at (0.22,1.06) {$o_e{=}s_e{+}p_e{=}2$};
\end{scope}

\begin{scope}[shift={(\dx,\dy)}]
  \iosub{-1.65,1.50}{D}{}
  \node[atom] (dC) at (-0.45,0) {C};
  \node[atom] (dO) at ( 1.10,0) {\textcolor{cdkO}{O}};
  \draw[line width=1.3pt,ioBlue!80!black]   ($(dC)+(0.36,0.11)$)--($(dO)+(-0.34,0.11)$);
  \draw[line width=1.3pt,ioPurple!35!black,dash pattern=on 2.4pt off 1.8pt]
        ($(dC)+(0.36,-0.11)$)--($(dO)+(-0.34,-0.11)$);
  \draw[move] (0.5,-0.3) to[out=-40,in=210] ($(dO)+(-0.05,-0.46)$);
  \node[lpdot] at ($(dO)+(0.46,-0.36)$) {}; \node[lpdot] at ($(dO)+(0.62,-0.24)$) {};
  \node[fld,black!60] at (0.35,1.06) {$\pi(\mathrm{C,O})\!\to\!\mathrm{lp}(\mathrm{O})$};
  \node[fld,black!60] at (0.15,-1.05) {$\Delta p_{\mathrm{C,O}}{=}{-}1,\ \Delta\ell_O{=}{+}1$};
\end{scope}

\draw[arr] ($(\ax,\ay)+(1.98,0.15)$) -- ($(\bx,\by)+(-2.05,0.15)$);
\node[font=\scriptsize\bfseries] at ($({0.5*(\ax+\bx)},\ay)+(0,0.72)$) {(1)};
\node[align=center,font=\tiny,text width=10mm]
  at ($({0.5*(\ax+1.98+\bx-2.05)},\ay)+(0,-0.58)$)
  {forget\\ identity};

\draw[arr] ($(\bx,\by)+(0,-1.30)$) -- ($(\cx,\cy)+(0,1.60)$);
\node[font=\scriptsize\bfseries] at ($(\bx,{0.5*(\by+\cy)})+(-0.30,0.12)$) {(2)};
\node[align=left,font=\tiny,text width=17mm] at ($(\bx,{0.5*(\by+\cy)})+(1.02,0.12)$)
  {condense\\ split $\sigma$/$\pi$};

\draw[arr] ($(\cx,\cy)+(-2.05,0.15)$) -- ($(\dx,\dy)+(1.98,0.15)$);
\node[font=\scriptsize\bfseries] at ($({0.5*(\cx+\dx)},\cy)+(0,0.80)$) {(3)};
\node[align=center,font=\tiny] at ($({0.5*(\cx+\dx)},\cy)+(0,-0.52)$)
  {execute};

\begin{pgfonlayer}{background}
  \foreach \px/\py in {\ax/\ay, \bx/\by, \cx/\cy, \dx/\dy}{
    \draw[iopanel] ($(\px,\py)+(-2.05,-1.30)$) rectangle ($(\px,\py)+(1.98,1.60)$);
  }
\end{pgfonlayer}
\end{tikzpicture}}
  \caption{From explicit electrons to an executable LLG state, on formaldehyde,
  following the numbered transitions clockwise from~(A).
  \textbf{(A)} The HEG retains individual electrons.  Dashed lines are the
  ownership map \(\psi\).  \textbf{(B)} Coarsening drops electron identity but
  keeps anonymous \textsf{Pair}/\textsf{LP} objects, so the two carbonyl pairs
  become indistinguishable.  \textbf{(C)} The span of
  Equation~\eqref{eq:span-efg-llg} condenses feature fibers to scalar node and
  edge labels, separating \(\sigma\)- and \(\pi\)-components.  \textbf{(D)} A
  locus-sorted event consumes the carbonyl \(\pi\) and creates an oxygen lone
  pair.  The delta \(\Delta_M\) is validated before committing
  \(c_{G'}=c_G+\Delta_M\).}
  \label{fig:abstraction-transition}
\end{figure*}

\section{Experiments and results}
\label{sec:experiments}

\subsection{Graph rewriting experiments}

\subsubsection{Data and protocol}

These experiments use 39,732 balanced reactions from
Ref.~\cite{Laffitte2026}.  Each has paired partial and reference-complete AAMs,
which support rule replay and independent ITS comparison.  The expansion
benchmark includes only balanced reactions with valid atom
maps~\cite{gonzalezlaffitte2024}.  The 5,426 \texttt{RMechDB}
records~\cite{Tavakoli2023} additionally test radical rule replay and
map-expansion coverage.

\begin{enumerate}
    \item[--] \textbf{Rule replay.} Extract one rule from each fully mapped
    reference and enumerate its forward and inverse applications.  Success
    requires exact recovery of the standardized reference endpoint.  We compare
    unique ITS outcomes.

    \item[--] \textbf{Partial-AAM expansion.} Expand each partial map with the
    method of Ref.~\cite{Laffitte2026}.  On the general corpus, success requires
    the induced ITS to be isomorphic to the independent reference ITS. \texttt{RMechDB} lacks independent full maps, so its result measures
    completion coverage, not mapping accuracy.
\end{enumerate}

\subsubsection{Rule replay}
\label{sec:res-rewriting}

Replay is paired: each rule is applied exhaustively, without per-direction
timeout, to its own originating reactant and product, so the two
representations differ only in label algebra.  Both recover the standardized
reference endpoint on all 39,732 records in both directions
(Figure~\ref{fig:benchmarks}A) and differ only in the admitted outcome set,
where LLG removes 96 forward (0.10\%) and 818 inverse (0.69\%) unique ITS
outcomes, an exclusion rate about sevenfold higher inverse, on baselines of
2.42 and 2.99 outcomes per record.  Since the rule sets are identical outside
the resource coordinates, every exclusion is necessarily a resource-availability
or \(\Pi\)-validity refusal rather than a structural one.  Whether each is
chemically correct is undetermined, as the excluded outcomes carry no
independent validity labels.

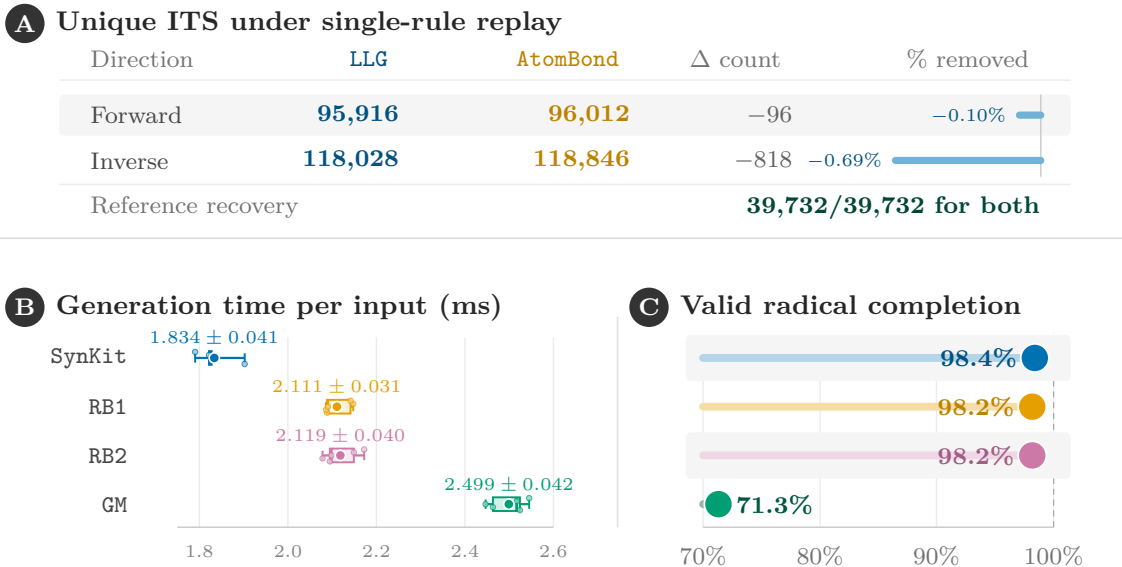
\begin{figure*}[htbp]
\centering
\resizebox{\textwidth}{!}{
\providecommand{\runboxplot}[9]{%
  \pgfmathsetmacro{\minx}{2.00+(#2-1.75)*5.176470588}
  \pgfmathsetmacro{\qonex}{2.00+(#3-1.75)*5.176470588}
  \pgfmathsetmacro{\medianx}{2.00+(#4-1.75)*5.176470588}
  \pgfmathsetmacro{\qthreex}{2.00+(#5-1.75)*5.176470588}
  \pgfmathsetmacro{\maxx}{2.00+(#6-1.75)*5.176470588}
  \foreach \mean/\sd in {#7}{
    \pgfmathsetmacro{\meanx}{2.00+(\mean-1.75)*5.176470588}
    \draw[draw=#9, line width=0.75pt] (\minx,#1)--(\maxx,#1);
    \draw[draw=#9, line width=0.65pt]
      (\minx,{#1-0.055})--(\minx,{#1+0.055})
      (\maxx,{#1-0.055})--(\maxx,{#1+0.055});
    \filldraw[fill=#9!16, draw=#9, line width=0.65pt]
      (\qonex,{#1-0.080}) rectangle (\qthreex,{#1+0.080});
    \draw[draw=#9, line width=0.9pt]
      (\medianx,{#1-0.080})--(\medianx,{#1+0.080});
    \foreach \runtime/\offset in {#8}{
      \pgfmathsetmacro{\runx}{2.00+(\runtime-1.75)*5.176470588}
      \filldraw[fill=#9!45, draw=#9, line width=0.35pt, opacity=0.78]
        (\runx,{#1+\offset}) circle[radius=0.031];
    }
    \filldraw[fill=#9, draw=white, line width=0.55pt]
      (\meanx,#1) circle[radius=0.060];
    \node[val, anchor=south, text=#9] at (\meanx,{#1+0.135})
      {\(\pgfmathprintnumber[fixed,fixed zerofill,precision=3]{\mean}
      \mathbin{\pm}
      \pgfmathprintnumber[fixed,fixed zerofill,precision=3]{\sd}\)};
  }
}
\begin{tikzpicture}[
  font=\small,
  row/.style={font=\scriptsize, anchor=east, text=black!72},
  tick/.style={font=\tiny, text=black!48},
  tickbig/.style={font=\scriptsize, text=black!55},
  val/.style={font=\tiny\bfseries, inner sep=0pt},
  valbig/.style={font=\scriptsize\bfseries, inner sep=0pt},
  guide/.style={draw=black!8, line width=0.45pt},
  axis/.style={draw=black!20, line width=0.5pt},
  band/.style={fill=black!4, rounded corners=1.6pt},
  stem/.style={line width=2.6pt, line cap=round},
  method/.style={circle, draw=white, line width=0.6pt,
    minimum size=3.3mm, inner sep=0pt},
]

\iosub{0.05,6.30}{A}{Unique ITS under single-rule replay}

\fill[band] (0.62,4.82) rectangle (12.45,5.30);
\draw[guide] (0.62,5.47)--(12.45,5.47);
\draw[guide] (0.62,4.25)--(12.45,4.25);

\node[tickbig, anchor=west] at (0.85,5.72) {Direction};
\node[tickbig, anchor=east, text=ioBlue!72!black] at (4.60,5.72)
  {\code{LLG}};
\node[tickbig, anchor=east, text=ioOrange!82!black] at (7.30,5.72)
  {\code{AtomBond}};
\node[tickbig, anchor=east] at (9.20,5.72) {\(\Delta\) count};
\node[tickbig, anchor=east] at (12.10,5.72) {\% removed};

\draw[axis] (12.10,4.34)--(12.10,5.30);

\node[row, anchor=west] at (0.85,5.06) {Forward};
\node[valbig, anchor=east, text=ioBlue!72!black] at (4.60,5.06) {95{,}916};
\node[valbig, anchor=east, text=ioOrange!82!black] at (7.30,5.06) {96{,}012};
\node[valbig, anchor=east, text=black!62] at (9.20,5.06) {\(-96\)};
\draw[ioBlue!55, line width=2.4pt, line cap=round] (12.10,5.06)--(11.85,5.06);
\node[val, anchor=east, text=ioBlue!72!black] at (11.70,5.06) {\(-0.10\%\)};

\node[row, anchor=west] at (0.85,4.53) {Inverse};
\node[valbig, anchor=east, text=ioBlue!72!black] at (4.60,4.53) {118{,}028};
\node[valbig, anchor=east, text=ioOrange!82!black] at (7.30,4.53) {118{,}846};
\node[valbig, anchor=east, text=black!62] at (9.20,4.53) {\(-818\)};
\draw[ioBlue!55, line width=2.4pt, line cap=round] (12.10,4.53)--(10.40,4.53);
\node[val, anchor=east, text=ioBlue!72!black] at (10.25,4.53) {\(-0.69\%\)};

\node[tickbig, anchor=west] at (0.85,3.98) {Reference recovery};
\node[valbig, anchor=east, text=ioGreen!45!black] at (12.10,3.98)
  {39{,}732/39{,}732 for both};

\iosub{0.05,2.98}{B}{Generation time per input (ms)}

\foreach \lab/\x in {1.8/2.259,2.0/3.294,2.2/4.329,2.4/5.365,2.6/6.400}{
  \draw[guide] (\x,0.24)--(\x,2.46);
  \node[tick, anchor=north] at (\x,0.16) {\lab};
}
\draw[axis] (2.00,0.24)--(6.40,0.24);

\node[row] at (1.55,2.22) {\texttt{SynKit}};
\node[row] at (1.55,1.65) {\texttt{RB1}};
\node[row] at (1.55,1.08) {\texttt{RB2}};
\node[row] at (1.55,0.51) {\texttt{GM}};

\runboxplot{2.22}
  {1.7903683897}{1.8216719691}{1.8263382844}{1.8276346832}{1.9023535011}
  {1.8336733655/0.0413115740}
  {1.9023535011/-0.070,1.8276346832/-0.035,1.8263382844/0,
   1.8216719691/0.035,1.7903683897/0.070}
  {ioBlue}
\runboxplot{1.65}
  {2.0862292070}{2.0884055751}{2.0901184532}{2.1415064890}{2.1476139488}
  {2.1107747346/0.0309479911}
  {2.0884055751/-0.070,2.0862292070/-0.035,2.0901184532/0,
   2.1476139488/0.035,2.1415064890/0.070}
  {ioOrange}
\runboxplot{1.08}
  {2.0782304679}{2.0948104113}{2.0980970131}{2.1490105925}{2.1723800391}
  {2.1185057048/0.0401032056}
  {2.0948104113/-0.070,2.0782304679/-0.035,2.0980970131/0,
   2.1490105925/0.035,2.1723800391/0.070}
  {ioPurple}
\runboxplot{0.51}
  {2.4466387733}{2.4632253237}{2.5168177655}{2.5242681186}{2.5450829157}
  {2.4992065793/0.0421334127}
  {2.5242681186/-0.070,2.4632253237/-0.035,2.4466387733/0,
   2.5168177655/0.035,2.5450829157/0.070}
  {ioGreen}

\draw[draw=black!12, line width=0.5pt] (7.15,0.24)--(7.15,2.70);

\iosub{7.35,2.98}{C}{Valid radical completion}

\foreach \lab/\x in {70/8.15,80/9.52,90/10.88,100/12.25}{
  \draw[guide] (\x,0.24)--(\x,2.46);
  \node[tickbig, anchor=north] at (\x,0.16) {\lab\%};
}
\draw[axis] (8.15,0.24)--(12.25,0.24);
\draw[dash pattern=on 2pt off 1.6pt, draw=black!35]
  (12.25,0.24)--(12.25,2.52);

\foreach \y in {2.22,1.08}{ \fill[band] (7.95,\y-0.27) rectangle (12.45,\y+0.27); }

\draw[stem, ioBlue!30]   (8.15,2.22)--(12.03,2.22);
\node[method, fill=ioBlue] at (12.03,2.22){};
\node[valbig, anchor=east, text=ioBlue!72!black] at (11.83,2.22) {98.4\%};
\draw[stem, ioOrange!35]  (8.15,1.65)--(12.00,1.65);
\node[method, fill=ioOrange] at (12.00,1.65){};
\node[valbig, anchor=east, text=ioOrange!82!black] at (11.80,1.65) {98.2\%};
\draw[stem, ioPurple!30]  (8.15,1.08)--(12.00,1.08);
\node[method, fill=ioPurple] at (12.00,1.08){};
\node[valbig, anchor=east, text=ioPurple!78!black] at (11.80,1.08) {98.2\%};
\draw[stem, ioGreen!45!black!40] (8.15,0.51)--(8.33,0.51);
\node[method, fill=ioGreen] at (8.33,0.51){};
\node[valbig, anchor=west, text=ioGreen!55!black] at (8.53,0.51) {71.3\%};

\draw[draw=black!15, line width=0.5pt] (-0.10,3.62)--(13.10,3.62);
\end{tikzpicture}}
\caption{\textbf{(A)} Unique ITS counts from single-rule forward and inverse
replay.  Columns compare \code{LLG}, \code{AtomBond}, and their difference.
Both recover all 39,732 references.  \textbf{(B)} Mean AAM-expansion generation
time per attempted input in each of five runs.
Boxes span the 25th--75th percentiles, center lines show medians, and whiskers
show the full observed range.  Small points show all five run-level means,
large points show their across-run means, and labels report mean \(\pm\) sample
s.d.  All methods recover every reference ITS.  \textbf{(C)} Fraction
of 5,426 \texttt{RMechDB} inputs yielding valid total AAMs.  The dashed line
marks 100\%, and \texttt{SynKit} yields 5,340 (98.4\%).}
\label{fig:benchmarks}
\end{figure*}

Radical replay has no legacy baseline, since atom--bond labels cannot encode
side-specific radical and lone-pair states.  Under the same stereo-independent
normalization it recovers all 5,426 references in both directions.

\subsubsection{Partial atom-map expansion}
\label{sec:res-aam}

All four methods recover the independent ITS for all 39,732 general records
(Figure~\ref{fig:benchmarks}B).  Across five runs, mean generation times are
$1.834\pm0.041$~ms per input for \texttt{SynKit}, $2.499\pm0.042$ for
\texttt{GM}, $2.111\pm0.031$ for \texttt{RB1}, and $2.119\pm0.040$ for
\texttt{RB2} (mean \(\pm\) sample s.d.).  Because the external methods run in a
reconstructed historical environment, these timings are descriptive rather
than a controlled speedup.

On the radical corpus, the minimal \texttt{SynKit} expansion returns 5,341 outputs, of
which 5,340 are valid total AAMs (98.4\%) (Figure~\ref{fig:benchmarks}C).
Coverage is 71.3\% for \texttt{GM} and 98.2\% for both \texttt{RB1} and
\texttt{RB2}.  These are coverage rates, not mapping accuracies, because no
independent full AAM exists.  Radical runtime is omitted because only
\texttt{SynKit} transports the required radical attribute.  The separate
conversion used for transition construction returns a total AAM for all 5,426
inputs, which is likewise a coverage result outside the minimal-expansion
comparison.

\subsection{Arrow-pushing experiments}

\subsubsection{Data and protocol}
\label{sec:arrow-protocol}

These experiments use 95,888 polar \texttt{PMechDB}
records~\cite{Tavakoli2024} and the 5,426 radical \texttt{RMechDB} records.
The controlled set contains 80 reviewed records from each corpus.  Both tasks
evaluate supplied electron-flow annotations as synchronized event groups rather
than by reaction-class or macro labels.

\begin{itemize}
    \item[--] \textbf{Controlled conformance.} Replay the 160 reviewed records
    and seven deterministic variants of each, giving 1,120 generated negative
    tests.  Every reviewed record must execute and reach its declared endpoint.
    Every generated variant must be rejected with its prescribed diagnostic.

    \item[--] \textbf{Transition construction.} Expand the AAM and resolve each
    published arrow endpoint to an LLG locus.  Success requires an admissible
    event group satisfying endpoint conformance,
    \(G\oplus\Delta_M\cong H_\star\).
\end{itemize}

\subsubsection{Controlled conformance}

The reviewed set contains 160 base cases, with 80 polar and 80 radical
elementary steps.  Strict replay accepts all 160 and reaches every declared
product.  For each base case, seven deterministic operators generate one
negative variant apiece by altering electron count, fishhook pairing, locus
existence or locality, source availability, update order, endpoint state, or
the declared product.  The checker rejects all synthetic tests (\(160\times7=1{,}120\)) generated
variants with their prescribed diagnostics
(see Figure~\ref{fig:transition-conformance-controls}).

\begin{figure*}[htbp]
\centering
\resizebox{\textwidth}{!}{
\begin{tikzpicture}[
  font=\small,
  atom/.style={catom, font=\normalsize},
  bond/.style={cbond},
  map/.style={font=\tiny, text=black!55, inner sep=0pt},
  delta/.style={font=\scriptsize, text=black!78, align=center},
  accept/.style={rounded corners=2pt, draw=ioGreen!58!black,
    fill=ioGreen!9, text=ioGreen!42!black, font=\scriptsize\bfseries,
    inner xsep=5pt, inner ysep=2.5pt},
  reject/.style={rounded corners=2pt, draw=ioVermilion!72!black,
    fill=ioVermilion!7, text=ioVermilion!76!black,
    font=\scriptsize\bfseries, inner xsep=5pt, inner ysep=2.5pt},
  okchip/.style={rounded corners=1.6pt, fill=ioGreen!52!black, text=white,
    font=\scriptsize\bfseries, inner sep=0pt, minimum size=3.6mm},
  nochip/.style={rounded corners=1.6pt, fill=ioVermilion!70!black, text=white,
    font=\scriptsize\bfseries, inner sep=0pt, minimum size=3.6mm},
  dupmap/.style={rounded corners=1pt, draw=ioVermilion!65,
    fill=ioVermilion!8, text=ioVermilion!82!black,
    font=\tiny\bfseries, inner sep=1.2pt},
]

\begin{scope}[shift={(0,0)}]
  \iosub{0.05,3.00}{A}{Polar: \(2e^-\)}
  \node[okchip] at (6.05,3.00) {\(\checkmark\)};
  \node[atom] (an) at (1.45,1.95) {\aN};
  \node[map] at ($(an)+(-0.27,0.29)$) {1};
  \node[atom] (ah1) at (0.72,1.40) {\aH};
  \node[atom] (ah2) at (2.18,1.40) {\aH};
  \node[atom] (ah3) at (1.45,0.98) {\aH};
  \draw[bond] (an)--(ah1);
  \draw[bond] (an)--(ah2);
  \draw[bond] (an)--(ah3);
  \node[celp] (alp1) at ($(an)+(-0.11,0.43)$) {};
  \node[celp] (alp2) at ($(an)+(0.11,0.43)$) {};
  \node[atom] at (2.65,1.75) {\(+\)};
  \node[atom] (ap) at (3.22,1.75) {\aH\(^{+}\)};
  \node[map] at ($(ap)+(0,-0.27)$) {4};
  \draw[cpair] ($(alp1)!0.5!(alp2)+(0,0.06)$)
    to[out=55,in=125] ($(ap)+(-0.14,0.27)$);
  \node[atom] at (4.03,1.75) {\(\longrightarrow\)};
  \node[atom] (anp) at (5.05,1.75) {\aN\aH\( _4^{+}\)};
  \node[delta] at (3.20,0.80)
    {\(\operatorname{lp}(1)\xrightarrow{\,2e^-\,}\sigma_{1,4}\)\\[-1pt]
     \(\Delta\ell_1=-1,\quad\Delta s_{1,4}=+1\)};
\end{scope}

\begin{scope}[shift={(6.75,0)}]
  \iosub{0.05,3.00}{B}{Polar: \(1e^-\)}
  \node[nochip] at (6.05,3.00) {\(\times\)};
  \node[atom] (bn) at (1.45,1.95) {\aN};
  \node[map] at ($(bn)+(-0.27,0.29)$) {1};
  \node[atom] (bh1) at (0.72,1.40) {\aH};
  \node[atom] (bh2) at (2.18,1.40) {\aH};
  \node[atom] (bh3) at (1.45,0.98) {\aH};
  \draw[bond] (bn)--(bh1);
  \draw[bond] (bn)--(bh2);
  \draw[bond] (bn)--(bh3);
  \node[celp] (blp1) at ($(bn)+(-0.11,0.43)$) {};
  \node[celp] (blp2) at ($(bn)+(0.11,0.43)$) {};
  \node[atom] at (2.65,1.75) {\(+\)};
  \node[atom] (bp) at (3.22,1.75) {\aH\(^{+}\)};
  \node[map] at ($(bp)+(0,-0.27)$) {4};
  \draw[cfish] ($(blp1)!0.5!(blp2)+(0,0.06)$)
    to[out=55,in=125] ($(bp)+(-0.14,0.27)$);
  \node[atom] at (4.03,1.75) {\(\nrightarrow\)};
  \node[delta] at (3.20,0.80)
    {\(\operatorname{lp}(1)\xrightarrow{\,1e^-\,}\sigma_{1,4}\)\\[-1pt]
     \(\Delta\ell_1=-\tfrac12,\quad\Delta s_{1,4}=+\tfrac12\)};
\end{scope}

\begin{scope}[shift={(0,-3.22)}]
  \iosub{0.05,3.00}{C}{Radical: paired}
  \node[okchip] at (6.05,3.00) {\(\checkmark\)};
  \node[atom] (cr1) at (0.62,1.75) {R};
  \node[atom] (co1) at (1.55,1.75) {\aO};
  \node[map] at ($(co1)+(0,-0.28)$) {1};
  \node[atom] (co2) at (2.72,1.75) {\aO};
  \node[map] at ($(co2)+(0,-0.28)$) {2};
  \node[atom] (cr2) at (3.65,1.75) {R};
  \draw[bond] (cr1)--(co1);
  \draw[bond] (co1)--(co2);
  \draw[bond] (co2)--(cr2);
  \coordinate (cmid) at ($(co1)!0.5!(co2)$);
  \draw[cfish] ($(cmid)+(0,0.03)$)
    to[out=118,in=35] ($(co1)+(0.08,0.33)$);
  \draw[cfish] ($(cmid)+(0,0.03)$)
    to[out=62,in=145] ($(co2)+(-0.08,0.33)$);
  \node[atom] at (4.25,1.75) {\(\longrightarrow\)};
  \node[atom] (cpo1) at (4.93,1.75) {R--\aO};
  \node[crad] at ($(cpo1)+(0.47,0.20)$) {};
  \node[atom] at (5.62,1.75) {\(+\)};
  \node[atom] (cpo2) at (6.18,1.75) {\aO--R};
  \node[crad] at ($(cpo2)+(-0.43,0.20)$) {};
  \node[delta] at (3.20,0.80)
    {\(\sigma_{1,2}\xrightarrow{\,1e^-\,}\rho_1,\quad
       \sigma_{1,2}\xrightarrow{\,1e^-\,}\rho_2\)\\[-1pt]
     \(\Delta s_{1,2}=-1,\quad\Delta r_1=\Delta r_2=+1\)};
\end{scope}

\begin{scope}[shift={(6.75,-3.22)}]
  \iosub{0.05,3.00}{D}{Radical: unpaired}
  \node[nochip] at (6.05,3.00) {\(\times\)};
  \node[atom] (dr1) at (0.62,1.75) {R};
  \node[atom] (do1) at (1.55,1.75) {\aO};
  \node[map] at ($(do1)+(0,-0.28)$) {1};
  \node[atom] (do2) at (2.72,1.75) {\aO};
  \node[map] at ($(do2)+(0,-0.28)$) {2};
  \node[atom] (dr2) at (3.65,1.75) {R};
  \draw[bond] (dr1)--(do1);
  \draw[bond] (do1)--(do2);
  \draw[bond] (do2)--(dr2);
  \coordinate (dmid) at ($(do1)!0.5!(do2)$);
  \draw[cfish] ($(dmid)+(0,0.03)$)
    to[out=118,in=35] ($(do1)+(0.08,0.33)$);
  \draw[ioVermilion!65, line width=0.8pt]
    ($(dmid)+(0.22,0.48)$)--($(dmid)+(0.55,0.81)$);
  \draw[ioVermilion!65, line width=0.8pt]
    ($(dmid)+(0.55,0.48)$)--($(dmid)+(0.22,0.81)$);
  \node[atom] at (4.25,1.75) {\(\nrightarrow\)};
  \node[delta] at (3.20,0.80)
    {\(\sigma_{1,2}\xrightarrow{\,1e^-\,}\rho_1\)\\[-1pt]
     \(\Delta s_{1,2}=-\tfrac12,\quad\Delta r_1=+1\)};
\end{scope}

\begin{scope}[shift={(0,-7.29)}]
  \iosub{0.05,3.88}{E}{Polar \(\cdot\) ID 49{,}985}
  \node[okchip] at (6.05,3.88) {\(\checkmark\)};
  \begin{scope}[yshift=-0.38cm]
  \node[atom] (b10)   at (1.95,2.41) {CH};      
  \node[atom] (b11)   at (2.75,2.87) {\aC};     
  \node[atom] (bn82a) at (3.55,2.41) {\aN H};   
  \node[atom] (b73)   at (3.55,1.49) {\aC};     
  \node[atom] (bn82b) at (2.75,1.03) {\aN H};   
  \node[atom] (b61)   at (1.95,1.49) {\aC};     

  \draw[cbond, double=white, double distance=1.5pt] (b10)--(b11);
  \draw[cbond] (b11)--(bn82a);
  \draw[cbond] (bn82a)--(b73);
  \draw[cbond] (b73)--(bn82b);
  \draw[cbond] (bn82b)--(b61);
  \draw[cbond] (b61)--(b10);

  \node[atom] (bo12) at (2.75,3.68) {\aO\(^{-}\)};
  \draw[cbond] (b11)--(bo12);
  \node[atom] (bo74) at (4.42,1.00) {\aO};
  \draw[cbond, double=white, double distance=1.5pt] (b73)--(bo74);
  \node[atom] (bo62) at (1.08,1.00) {\aO};
  \draw[cbond, double=white, double distance=1.5pt] (b61)--(bo62);

  \node[map] at ($(b10)+(-0.30,0.14)$) {10};
  \node[map] at ($(b11)+(-0.34,0.02)$) {11};
  \node[dupmap] at ($(bn82a)+(0.34,0.24)$) {82};
  \node[map] at ($(b73)+(0.34,0.14)$) {73};
  \node[dupmap] at ($(bn82b)+(0.02,-0.42)$) {82};
  \node[map] at ($(b61)+(-0.34,0.06)$) {61};
  \node[map] at ($(bo12)+(-0.36,0.02)$) {12};
  \node[map] at ($(bo74)+(0.30,0.02)$) {74};
  \node[map] at ($(bo62)+(-0.30,0.02)$) {62};
  \end{scope}
\end{scope}

\begin{scope}[shift={(6.75,-7.29)}]
  \iosub{0.05,3.88}{F}{Radical \(\cdot\) ID 4{,}852}
  \node[nochip] at (6.05,3.88) {\(\times\)};
  \node[atom] (ccr) at (0.96,2.57) {\aC};
  \node[map] at ($(ccr)+(0,-0.32)$) {21};
  \draw[cbond] (ccr)--(0.35,3.05);
  \draw[cbond] (ccr)--(0.35,2.10);
  \node[font=\scriptsize] at (0.19,3.13) {R};
  \node[font=\scriptsize] at (0.19,2.02) {R};
  \node[atom] (ch20) at (1.78,2.57) {\aH};
  \node[map] at ($(ch20)+(0,-0.32)$) {20};
  \draw[cbond] (ccr)--(ch20);
  \node[font=\normalsize] at (2.18,2.57) {\(+\)};
  \node[atom] (coh) at (2.72,2.57) {\(\boldsymbol{\cdot}\)\aO H};
  \node[map] at ($(coh)+(0,-0.32)$) {10};
  \node[font=\large] at (3.36,2.57) {\(\longrightarrow\)};

  \node[atom] (ccp) at (4.08,2.57) {\aC\(\boldsymbol{\cdot}\)};
  \node[map] at ($(ccp)+(0,-0.32)$) {21};
  \draw[cbond] (ccp)--(3.57,3.05);
  \draw[cbond] (ccp)--(3.57,2.10);
  \node[font=\scriptsize] at (3.43,3.13) {R};
  \node[font=\scriptsize] at (3.43,2.02) {R};
  \node[font=\normalsize] at (4.62,2.57) {\(+\)};
  \node[atom] (cwh) at (5.03,2.57) {\aH};
  \node[map] at ($(cwh)+(0,-0.32)$) {20};
  \node[atom] (cwo) at (5.78,2.57) {\aO H};
  \node[map] at ($(cwo)+(0,-0.32)$) {10};
  \draw[cbond] (cwh)--(cwo);

  \node[font=\scriptsize, text=ioVermilion!78!black] at (1.90,1.10)
    {\(\boldsymbol{(20,21)\rightharpoonup2}\)};
  \node[font=\scriptsize, text=black!38] at (3.40,1.10)
    {\(\Longrightarrow\)};
  \node[font=\scriptsize, text=ioGreen!45!black] at (4.90,1.10)
    {\(\boldsymbol{(20,21)\rightharpoonup21}\)};
\end{scope}

\begin{pgfonlayer}{background}
  \draw[iopanel] (-0.25,0.30) rectangle (6.55,3.35);
  \draw[iopanel] (6.50,0.30) rectangle (13.30,3.35);
  \draw[iopanel] (-0.25,-2.92) rectangle (6.55,0.13);
  \draw[iopanel] (6.50,-2.92) rectangle (13.30,0.13);
  \draw[iopanel] (-0.25,-7.31) rectangle (6.55,-3.09);
  \draw[iopanel] (6.50,-7.31) rectangle (13.30,-3.09);
\end{pgfonlayer}
\end{tikzpicture}}
\caption{Conformance controls (A--D) and construction-audit source defects
(E--F).  In A--D, columns compare base cases and generated variants, while rows
distinguish polar and radical events.  \textbf{(A--B)} Two-electron protonation
is integral.  Changing it to \(1e^-\) is not.  \textbf{(C--D)} Paired homolysis
fishhooks are integral.  Deleting one is not.  \textbf{(E)} Polar ID 49,985 duplicates map 82
on symmetry-equivalent nitrogens, resolved by preprocessing without endpoint
edits.  \textbf{(F)} Radical ID 4,852 targets absent map 2 and requires manual reassignment to map 21.
Panels E--F reproduce source annotations.  Table~\ref{tab:radical-arrow-review} gives the full review.}
\label{fig:transition-conformance-controls}
\end{figure*}
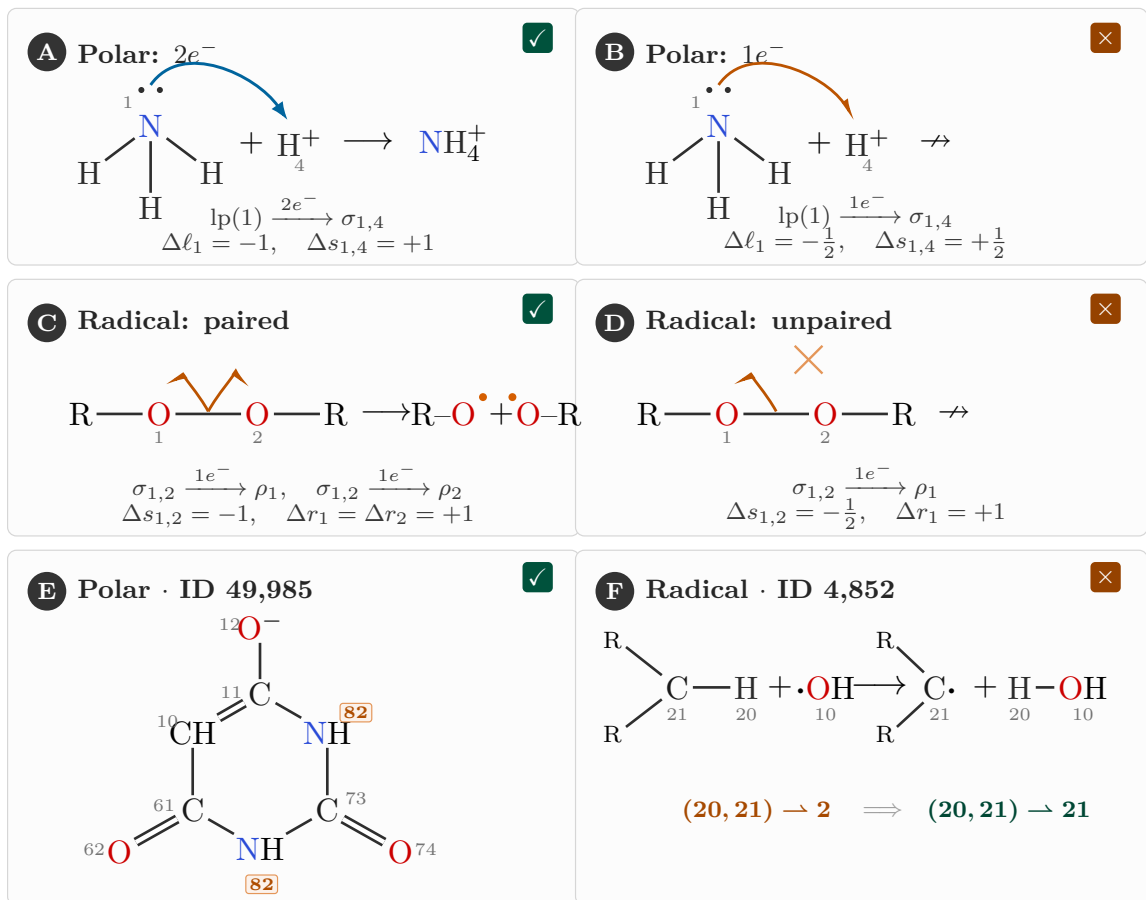

\subsubsection{Transition construction}
\label{sec:res-construction}

Construction proceeds in two stages: partial-AAM expansion, then resolution of
each published arrow endpoint to an LLG locus and assembly into event groups,
scored by the criterion of Section~\ref{sec:arrow-protocol}.  Expansion here uses
the guarded conversion of Section~\ref{sec:res-aam}, not the minimal path
benchmarked there, and succeeds on all 101,314 records.

In 3,274 polar records the annotation reuses one positive map on two
symmetry-equivalent atoms, an assignment that is unambiguous precisely because
the positions are equivalent
(Figure~\ref{fig:transition-conformance-controls}E).  Preprocessing removes the
duplicate label and lets expansion assign a fresh identifier, leaving both
endpoint structures unchanged, after which all 95,888 construct.  In the radical
corpus, direct resolution fails for 11 records.  Ten require manual annotation
corrections: six map-locus edits, two endpoint-orbit alignments, one fishhook
replaced by a paired arrow, and one mixed paired/fishhook group
(Table~\ref{tab:radical-arrow-review}).  Applying them raises construction to
5,425/5,426 radical and 101,313/101,314 total records.

ID 2,207 instead requires an atom-map correction.  Its fishhooks break
\((12,13)\), form \((9,13)\), and leave the radical at map 12, whereas the
endpoint maps H15 to the new O9--H bond.  Transferring H15 would break
O14--H15 and leave the radical at map 14.  Thus no arrow edit can satisfy the
current endpoint; H13 must map to O9 and O14--H15 remain intact.  This proposed
map correction is not applied.  Figure~\ref{fig:si-radical-arrow-review}
contrasts it with the arrow correction for ID 2,300.

\section{Conclusion, limitations, and outlook}
\label{sec:conclusion}

The Lewis-labeled graph supplies the state space missing between molecular
graphs, rewrite rules, and arrow-pushing annotations.  Its four resource fields
\((\ell_v,r_v,s_e,p_e)\) determine bond order, formal charge, and represented
electron inventory.  Rule matching tests the resources a transformation
consumes, and an electron-flow annotation becomes a factorization of the
endpoint change into admissible deltas on those same fields.
Theorem~\ref{thm:coherence} makes the continuity precise: the three levels are
one decision on one set of coordinates, not three procedures that happen to
agree.  The construction tests resource availability without demanding equality
of derived formal charge, makes open-shell state explicit without expanding the
atom--bond skeleton, and derives the atomic coupling of bond-centered fishhooks
from integrality rather than from drawing convention.

The experiments bound these consequences.  Electron-aware rules retain complete
bidirectional reference recovery while admitting slightly fewer unique ITS
outcomes in paired replay.  Partial-map expansion recovers every independent
reference ITS.  Ten radical annotations require manual correction, after which
101,313/101,314 records yield endpoint-consistent transitions, while the remaining
record requires an unapplied endpoint-map correction.  All 160 reviewed steps
are accepted, and all 1,120 generated variants rejected under their declared
diagnostics.  This is
representational and operational consistency, not chemical precision: the
excluded outcomes carry no validity labels, paired replay is not an \(N\)-rule
by \(M\)-molecule screen, and the negative variants are synthetic.

The boundary is substantive.  LLG is a localized two-center Lewis model, so
multicenter bonding, delocalized spin coupling, excited states, fractional
occupation, and general organometallic electron counting require new state
fields.  Imported lone-pair counts remain conditional on the perception policy,
aromatic systems need an explicit Kekul\'e policy, and the eight polar classes
with seven radical macros form a template library covering the source taxonomies
rather than a proved decomposition.  No result here shows that every admissible
integral event group factors through them, and supplying such a closure theorem,
or a counterexample bounding it, is open.  A field-by-field ablation is likewise
unavailable, since removing a coordinate changes the state schema rather than a
parameter.  Conformance checking decides only whether an annotation is local,
resource-admissible, integral, policy-consistent, and endpoint-reaching.  It
settles neither kinetics, thermodynamic preference, solvent response,
stereoselectivity, nor the preferred mechanism, and stereochemistry lies outside
the replay benchmark by construction.

Formally, the next steps are to characterize resource-constrained LLG rewriting
as an attributed rewriting category, to state when the condensation functor
\(K\) preserves and reflects admissible derivations, and, building on the
additive composition of Remark~\ref{rem:composition}, to prove a concurrency
theorem under which sequential mechanisms compose into single
resource-constrained rules and factor back into admissible steps.  Empirically,
they are a paired comparison with expanded electron-feature graphs under
identical search conditions, then validation on an independently reviewed,
frozen corpus.  Only then should the conformance checker serve as the acceptance
layer for mechanism enumeration or ranking, its certificates and trajectory
graphs marking the boundary between generation and conformance checking.

\section{Declarations}

\subsection{Availability of data and materials}
The \synkit source code is available under the MIT License at
\url{https://github.com/TieuLongPhan/SynKit}.  Documentation is available at
\url{https://tieulongphan.github.io/SynKit/}.  Archived releases are available
at \url{https://doi.org/10.5281/zenodo.15269901}.  Every result reported here was
produced with version \code{1.6.0}. The
reaction corpora remain under the licenses of their providers and are not
redistributed here, so \texttt{PMechDB}~\cite{Tavakoli2024} and \texttt{RMechDB}~\cite{Tavakoli2023} must be obtained
from the original sources.

\subsection{Conflict of interest}
The author declares no conflict of interest.

\subsection{Funding}
This work was supported by the European Union's Horizon Europe research and
innovation programme under Marie Sk{\l}odowska--Curie grant
agreement No. 101072930 (TACsy, Training Alliance for Computational Systems
Chemistry). Views and opinions expressed are however those of the author(s) only and do not necessarily reflect those of the European Union. Neither the European Union nor the granting authority can be held responsible for them

\subsection{Author contributions}
T.-L.P.: Conceptualization, Methodology, Software, Validation, Formal analysis,
Investigation, Data curation, Visualization, Writing (original draft), and
Writing (review and editing).

\subsection{Acknowledgements}
Not applicable.

\bibliographystyle{unsrtnat}
\bibliography{references}

\clearpage
\renewcommand{\thesection}{S\arabic{section}}
\setcounter{section}{0}
\setcounter{figure}{0}
\setcounter{table}{0}
\setcounter{equation}{0}
\renewcommand{\thefigure}{S\arabic{figure}}
\renewcommand{\thetable}{S\arabic{table}}
\renewcommand{\theequation}{S\arabic{equation}}
\renewcommand{\theHsection}{S\arabic{section}}
\renewcommand{\theHsubsection}{\theHsection.\arabic{subsection}}
\renewcommand{\theHsubsubsection}{\theHsubsection.\arabic{subsubsection}}
\renewcommand{\theHfigure}{S\arabic{figure}}
\renewcommand{\theHtable}{S\arabic{table}}
\renewcommand{\theHequation}{S\arabic{equation}}
\section*{Supporting Information}

\section{Locus-sorted event-group grammar}
\label{si:locus-grammar}

Tables~\ref{tab:polar-classes}--\ref{tab:macros} list each locus pattern, its
mechanistic reading, and its inverse.  Polar rows classify individual
two-electron moves.  A complete group may contain several rows and remains
subject to common-prestate availability and \(\Pi_{\mathrm{exec}}\).

\begin{table}[htbp]
\caption{Eight polar two-electron locus transitions.  The classes are the
ordered pairs over \(\{\operatorname{lp},\sigma,\pi\}\) with at least one bond
endpoint, each classifying a single curly arrow rather than a complete step.
The diagonal \(\operatorname{lp}\!\to\!\operatorname{lp}\) is excluded because a
two-electron pair relocation that edits no bond has no elementary chemical
realization.}
\label{tab:polar-classes}
\centering
\small
\begin{tabularx}{\textwidth}{@{}l >{\raggedright\arraybackslash}X l@{}}
\toprule
Arrow pattern & Typical context & Inverse \\
\midrule
\(\operatorname{lp}\!\to\!\pi\) &
  nucleophilic addition &
  \(\pi\!\to\!\operatorname{lp}\) \\
\(\operatorname{lp}\!\to\!\sigma\) &
  \(S_N2\) or proton transfer &
  \(\sigma\!\to\!\operatorname{lp}\) \\
\(\pi\!\to\!\operatorname{lp}\) &
  ionization or retro-addition &
  \(\operatorname{lp}\!\to\!\pi\) \\
\(\pi\!\to\!\pi\) &
  conjugated or pericyclic shift &
  self \\
\(\pi\!\to\!\sigma\) &
  electrophilic addition &
  \(\sigma\!\to\!\pi\) \\
\(\sigma\!\to\!\operatorname{lp}\) &
  heterolysis or E1 &
  \(\operatorname{lp}\!\to\!\sigma\) \\
\(\sigma\!\to\!\pi\) &
  elimination or retro-cycloaddition &
  \(\pi\!\to\!\sigma\) \\
\(\sigma\!\to\!\sigma\) &
  1,2-migration or hydride transfer &
  self \\
\bottomrule
\end{tabularx}
\end{table}

\begin{table}[htbp]
\caption{Seven radical event macros.}
\label{tab:macros}
\centering
\small
\begin{tabularx}{\textwidth}{@{}
  >{\raggedright\arraybackslash}p{2.6cm}
  >{\raggedright\arraybackslash}p{4.1cm}
  >{\raggedright\arraybackslash}X
  >{\raggedright\arraybackslash}p{2.4cm}@{}}
\toprule
Macro & Fishhook group & Incidence & Inverse \\
\midrule
Homolysis &
  \(\sigma\!\to\!\boldsymbol{\cdot}\ (\times 2)\) &
  both broken-bond atoms &
  recombination \\
Recombination &
  \(\boldsymbol{\cdot}\!\to\!\sigma\ (\times 2)\) &
  both new-bond atoms &
  homolysis \\
Radical addition &
  \(\boldsymbol{\cdot}\!\to\!\sigma,\ \pi\!\to\!\sigma,\ \pi\!\to\!\boldsymbol{\cdot}\) &
  old \(\pi\) and new \(\sigma\) share a centre &
  beta scission \\
Beta scission &
  \(\boldsymbol{\cdot}\!\to\!\pi,\ \sigma\!\to\!\pi,\ \sigma\!\to\!\boldsymbol{\cdot}\) &
  old \(\sigma\) and new \(\pi\) share a centre &
  radical addition \\
Radical resonance &
  \(\boldsymbol{\cdot}\!\to\!\pi,\ \pi\!\to\!\pi,\ \pi\!\to\!\boldsymbol{\cdot}\) &
  migrating \(\pi\) shares a centre &
  self \\
H abstraction &
  \(\boldsymbol{\cdot}\!\to\!\sigma,\ \sigma\!\to\!\sigma,\ \sigma\!\to\!\boldsymbol{\cdot}\) &
  old and new bonds share the mapped H &
  self \\
Alpha resonance
(LP/radical relocation) &
  annotated \(\operatorname{lp}_d\!\to\!\operatorname{lp}_a\), expanded to
  \(\operatorname{lp}_d\!\to\!\boldsymbol{\cdot}_d,\
  \operatorname{lp}_d\!\to\!\operatorname{lp}_a,\
  \boldsymbol{\cdot}_a\!\to\!\operatorname{lp}_a\) &
  adjacent \(d\ne a\), \(\ell_d,r_a\ge1\).  The expansion, not the surface
  arrow, carries the integral four-field delta of
  Equation~\eqref{eq:lp-radical-expansion} &
  self \\
\bottomrule
\end{tabularx}
\vspace{2pt}

\parbox{\textwidth}{\footnotesize Conjugated H abstraction expands the simple
group with \(\sigma\!\to\!\pi\) and
\(\boldsymbol{\cdot}\!\to\!\pi\).  Reversing the four moves gives the inverse.}
\end{table}

\section{Electron-flow event schema}

\begin{table}[htbp]
\caption{Electron-flow record schema.}
\label{tab:electron-flow-schema}
\centering
\begin{tabularx}{\textwidth}{@{}lXX@{}}
\toprule
Object & Carries & Acceptance condition \\
\midrule
Electron locus & Kind and support in \(V\cup E\) &
  Vertex for lone pair/radical and edge for \(\sigma/\pi\) \\
Curly arrow & Source and target loci &
  Two electrons with valid local incidence \\
Fishhook & Source and target loci &
  One electron in an integral group after shorthand expansion \\
Event group & Move multiset &
  Complete macro from one pre-state with an integral atomic delta \\
Transition record & Groups and endpoint &
  All groups admissible and final state equal to endpoint \\
\bottomrule
\end{tabularx}
\end{table}

Positive atom maps are support keys, not coordinates of \(a(v)\).  The locus
alphabet is \(\{\operatorname{lp},\sigma,\pi,\boldsymbol{\cdot}\}\).
\(\sigma\) and \(\pi\) remain distinct whenever they enable different moves.

\section{Composite rules and mechanism sequences}
\label{si:composition}

Remark~\ref{rem:composition} treats a mechanism as a derivation sequence and
notes that additive composition is weaker than DPO rule composition.
Figure~\ref{fig:mechanism-composition} draws the \(S_N2\) case for comparison.

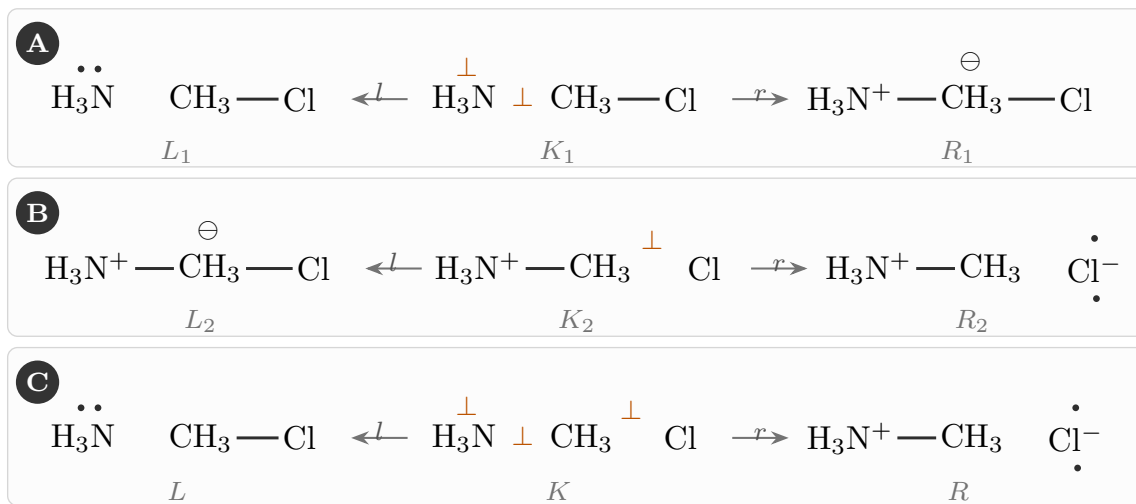
\begin{figure}[!htb]
  \centering
  \resizebox{\textwidth}{!}{
\begin{tikzpicture}[
  font=\small,
  grp/.style={catom, font=\small},
  flag/.style={font=\scriptsize, text=ioVermilion!85!black},
  mor/.style={-{Stealth[length=1.9mm]}, line width=0.62pt, draw=black!55},
  ml/.style={font=\scriptsize\itshape, text=black!68, inner sep=1pt},
  free/.style={font=\scriptsize, text=ioVermilion!88!black},
  sub/.style={font=\scriptsize, text=black!55},
]

\begin{scope}[shift={(0,0)}]
  \iosub{-0.15,0.76}{A}{}
  \node[grp](aln) at (0.5,0){\(\mathrm{H_3N}\)}; \node[celp] at ($(aln)+(-0.02,0.32)$){}; \node[celp] at ($(aln)+(0.16,0.32)$){};
  \node[grp](alc) at (1.75,0){\(\mathrm{CH_3}\)}; \node[grp](alcl) at (2.8,0){Cl}; \draw[cbond](alc)--(alcl);
  \draw[mor](3.95,0)--(3.35,0) node[ml,midway,above=-2pt]{\(l\)};
  \node[grp](akn) at (4.55,0){\(\mathrm{H_3N}\)}; \node[free] at ($(akn)+(0.03,0.34)$){\(\bot\)};
  \node[grp](akc) at (5.8,0){\(\mathrm{CH_3}\)}; \node[grp](akcl) at (6.85,0){Cl}; \draw[cbond](akc)--(akcl);
  \node[free] at ($(akn)!0.5!(akc)+(0,0.02)$){\(\bot\)};
  \draw[mor](7.4,0)--(8.0,0) node[ml,midway,above=-2pt]{\(r\)};
  \node[grp](arn) at (8.65,0){\(\mathrm{H_3N^+}\)}; \node[grp](arc) at (9.95,0){\(\mathrm{CH_3}\)}; \node[grp](arcl) at (11.05,0){Cl};
  \draw[cbond](arn)--(arc); \draw[cbond](arc)--(arcl);
  \node[cchg] at ($(arc)+(-0.02,0.4)$){\chminus};
  \node[sub] at (1.5,-0.56){\(L_1\)}; \node[sub] at (5.55,-0.56){\(K_1\)}; \node[sub] at (9.8,-0.56){\(R_1\)};
\end{scope}

\begin{scope}[shift={(0,-1.80)}]
  \iosub{-0.15,0.76}{B}{}
  \node[grp](bln) at (0.55,0){\(\mathrm{H_3N^+}\)}; \node[grp](blc) at (1.85,0){\(\mathrm{CH_3}\)}; \node[grp](blcl) at (2.95,0){Cl};
  \draw[cbond](bln)--(blc); \draw[cbond](blc)--(blcl);
  \node[cchg] at ($(blc)+(-0.02,0.4)$){\chminus};
  \draw[mor](4.1,0)--(3.5,0) node[ml,midway,above=-2pt]{\(l\)};
  \node[grp](bkn) at (4.7,0){\(\mathrm{H_3N^+}\)}; \node[grp](bkc) at (6.0,0){\(\mathrm{CH_3}\)}; \node[grp](bkcl) at (7.1,0){Cl};
  \draw[cbond](bkn)--(bkc);
  \node[free] at ($(bkc)!0.5!(bkcl)+(0,0.3)$){\(\bot\)};
  \draw[mor](7.6,0)--(8.2,0) node[ml,midway,above=-2pt]{\(r\)};
  \node[grp](brn) at (8.85,0){\(\mathrm{H_3N^+}\)}; \node[grp](brc) at (10.15,0){\(\mathrm{CH_3}\)}; \draw[cbond](brn)--(brc);
  \node[grp](brcl) at (11.25,0){\(\mathrm{Cl^-}\)}; \node[celp] at ($(brcl)+(0,0.32)$){}; \node[celp] at ($(brcl)+(0.02,-0.32)$){};
  \node[sub] at (1.75,-0.56){\(L_2\)}; \node[sub] at (5.75,-0.56){\(K_2\)}; \node[sub] at (9.95,-0.56){\(R_2\)};
\end{scope}

\begin{scope}[shift={(0,-3.60)}]
  \iosub{-0.15,0.76}{C}{}
  \node[grp](cln) at (0.5,0){\(\mathrm{H_3N}\)}; \node[celp] at ($(cln)+(-0.02,0.32)$){}; \node[celp] at ($(cln)+(0.16,0.32)$){};
  \node[grp](clc) at (1.75,0){\(\mathrm{CH_3}\)}; \node[grp](clcl) at (2.8,0){Cl}; \draw[cbond](clc)--(clcl);
  \draw[mor](3.95,0)--(3.35,0) node[ml,midway,above=-2pt]{\(l\)};
  \node[grp](ckn) at (4.55,0){\(\mathrm{H_3N}\)}; \node[free] at ($(ckn)+(0.03,0.34)$){\(\bot\)};
  \node[grp](ckc) at (5.8,0){\(\mathrm{CH_3}\)}; \node[grp](ckcl) at (6.85,0){Cl};
  \node[free] at ($(ckn)!0.5!(ckc)+(0,0.02)$){\(\bot\)};
  \node[free] at ($(ckc)!0.5!(ckcl)+(0,0.3)$){\(\bot\)};
  \draw[mor](7.4,0)--(8.0,0) node[ml,midway,above=-2pt]{\(r\)};
  \node[grp](crn) at (8.65,0){\(\mathrm{H_3N^+}\)}; \node[grp](crc) at (9.95,0){\(\mathrm{CH_3}\)}; \draw[cbond](crn)--(crc);
  \node[grp](crcl) at (11.05,0){\(\mathrm{Cl^-}\)}; \node[celp] at ($(crcl)+(0,0.32)$){}; \node[celp] at ($(crcl)+(0.02,-0.32)$){};
  \node[sub] at (1.5,-0.56){\(L\)}; \node[sub] at (5.55,-0.56){\(K\)}; \node[sub] at (9.8,-0.56){\(R\)};
\end{scope}

\begin{pgfonlayer}{background}
  \draw[iopanel] (-0.3,-0.72) rectangle (11.8,0.96);
  \draw[iopanel] ($(0,-1.80)+(-0.3,-0.72)$) rectangle ($(0,-1.80)+(11.8,0.96)$);
  \draw[iopanel] ($(0,-3.60)+(-0.3,-0.72)$) rectangle ($(0,-3.60)+(11.8,0.96)$);
\end{pgfonlayer}
\end{tikzpicture}}
  \caption{Composition of the \(S_N2\) step in
  \(\ce{NH3 + CH3Cl -> CH3NH3+ + Cl-}\).
  \textbf{(A)} \(p_1\) forms the N--C bond, giving the non-\(\Pi\)-valid shared
  object \(R_1=L_2\). \textbf{(B)} \(p_2\) expels chloride.
  \textbf{(C)} Their composite keeps this intermediate internal and maps valid
  reactants directly to valid products.}
  \label{fig:mechanism-composition}
\end{figure}

\section{Relation to half-edge and electron-feature graphs}
\label{si:efg-llg}

Figure~\ref{fig:abstraction-transition} compares LLGs with the HEG and EFG
constructions of Ref.~\cite{holzschuh2026halfedge}.  HEGs retain electron
identity.  EFGs retain anonymous feature objects.  LLGs store feature-fiber
cardinalities and distinguish \(\sigma\) from \(\pi\).  Table~\ref{tab:efg-llg-notation}
gives the correspondence.

\begin{table}[htbp]
\caption{Correspondence between electron-feature and LLG quantities.  Fiber
cardinality is denoted by vertical bars, and the \(\sigma/\pi\) refinement is made
before condensation.}
\label{tab:efg-llg-notation}
\centering
\small
\begin{tabularx}{\textwidth}{@{}lXX@{}}
\toprule
Chemical quantity & HEG/EFG realization & LLG realization \\
\midrule
Neutral valence electrons at \(v\) & atom attribute \(\operatorname{val}(a(v))\) & \(z_v\) \\
Lone pairs at \(v\) & fiber \(\textsf{LP}(v)\) & \(\ell_v=|\textsf{LP}(v)|\) \\
Radical electrons at \(v\) & fiber \(\textsf{Rad}(v)\) & \(r_v=|\textsf{Rad}(v)|\) \\
Shared pairs on \(\{u,v\}\) & fiber \(\textsf{Pair}(u,v)\), refined by kind &
\(\begin{aligned}
s_{uv}&=|\textsf{Pair}_{\sigma}(u,v)|,\\
p_{uv}&=|\textsf{Pair}_{\pi}(u,v)|
\end{aligned}\) \\
Total bond order & shared-pair multiplicity & \(o_{uv}=s_{uv}+p_{uv}\) \\
Formal charge at \(v\) & \(\operatorname{ch}(v)\), derived from incidence & \(q_v\) from Equation~\eqref{eq:charge} \\
\bottomrule
\end{tabularx}
\end{table}

\subsection{Categorical relation}
Let
\(\mathbf{EFG}^{\mathrm{fib},\Pi}_{\sigma/\pi}\) refine \(\textsf{Pair}\) into
\(\textsf{Pair}_{\sigma}\) and \(\textsf{Pair}_{\pi}\), retain objects whose
condensed labels satisfy \(\Pi\), and restrict morphisms to maps bijective on
each feature fiber.  Fiber cardinality then defines a condensation functor
\(K\), while \(F\) merges the two pair sorts:
\begin{equation}
\label{eq:span-efg-llg}
  \mathbf{EFG} \xleftarrow{\quad F \quad}
  \mathbf{EFG}^{\mathrm{fib},\Pi}_{\sigma/\pi}
  \xrightarrow{\quad K \quad} \mathbf{LLG}_{\Pi}^{=}.
\end{equation}
\(F\) forgets bond kind, while \(K\) forgets feature identity.  A span alone
would not show that neither target is a quotient of the other.  That follows
because the two functors discard incomparable information.  \(K\) retains the
\(\sigma/\pi\) distinction \(F\) discards, so \(K\) cannot factor through
\(F\).  In contrast, \(F\) retains the feature identity, that is which
\textsf{Pair} reacted, which
\(K\) discards, so \(F\) cannot factor through \(K\).  Neither is a refinement
of the other, and the \(\sigma/\pi\) refinement must precede condensation.
Functoriality on this restricted class is routine, both maps fixing the atom set
and acting on fibers.  We claim nothing for morphisms outside it.

\subsection{Determinate bond loci}
An unrefined EFG multiple bond contains interchangeable \(\textsf{Pair}\)
objects, so selecting one reacting pair is non-canonical.  An LLG stores
\((s_e,p_e)\), making \(\pi_e\) a determinate transition locus.

\subsection{Executable transition semantics}
Ref.~\cite{holzschuh2026halfedge} proves existence of electron-move
explanations for balanced mapped reactions, and gives double-pushout rules over
electron-feature ACSets that it proves to be quotients of electron-anonymous
half-edge rules.  Both are statements about what rules and explanations exist.
Our transition relation instead tests whether a specified group is executable
from a given state.  It checks availability, integrality, \(\Pi\), and the
declared endpoint.

Condensation loses electron and feature identity but leaves an atom--bond graph
with scalar attributes.  This is a representational claim about graph-object
count, not a runtime or memory claim.

\section{Audit of radical transition annotations}
\label{si:radical-arrow-review}

Direct resolution fails for the 11 radical records listed in
Table~\ref{tab:radical-arrow-review}.  Ten become endpoint-conformant only after
manual annotation corrections.  ID 2,207 is different: its fishhooks are
coherent with transfer of H13, but the mapped endpoint assigns H15 to the new
O9--H bond.  It therefore requires a hydrogen-map correction rather than an
arrow edit, and that correction is not applied.  Figure~\ref{fig:si-radical-arrow-review}
expands ID 2,207 and the nonstandard mixed-group correction for ID 2,300.  In
the latter, two fishhooks form the radical component and two paired moves
complete the bond rearrangement.

\begin{figure}[htbp]
\centering
\resizebox{\textwidth}{!}{
\begin{tikzpicture}[
  font=\small,
  map/.style={font=\scriptsize, text=black!52, inner sep=0.4pt},
  note/.style={font=\scriptsize, text=black!58, align=center},
  newbond/.style={draw=black!28, dash pattern=on 2pt off 1.5pt,
    line width=0.8pt},
]

\begin{scope}[shift={(0,4.55)}]
  \draw[iopanel] (0,0) rectangle (7.0,4.2);
  \iosub{0.12,4.08}{A}{ID 2,207 \(\cdot\) published arrows}

  \node[font=\scriptsize] (ar) at (0.54,2.10) {R};
  \node[catom] (a9) at (1.18,2.10) {\aO};
  \node[crad] (a9rad) at ($(a9)+(0.24,0.31)$) {};
  \node[catom] (a13) at (2.35,2.10) {\aH};
  \node[catom] (a12) at (3.35,2.10) {\aC};
  \node[catom] (a14) at (4.50,2.10) {\aO};
  \node[catom] (a15) at (5.60,2.10) {\aH};
  \node[font=\scriptsize] (arr) at (3.35,2.88) {R};

  \draw[cbond] (ar)--(a9);
  \draw[cbond] (a13)--(a12);
  \draw[cbond] (a12)--(a14);
  \draw[cbond] (a14)--(a15);
  \draw[cbond] (a12)--(arr);
  \draw[newbond] (a9)--(a13);

  \coordinate (a913) at ($(a9)!0.53!(a13)$);
  \coordinate (a1213) at ($(a12)!0.50!(a13)$);
  \draw[cfish] (a9rad.center) to[bend left=24] ($(a913)+(0,0.10)$);
  \draw[cfish] ($(a1213)+(0.04,0.08)$)
    to[bend right=38] ($(a913)+(0.03,-0.08)$);
  \draw[cfish] ($(a1213)+(-0.04,-0.08)$)
    to[bend right=28] ($(a12)+(-0.06,0.35)$);

  \node[map] at ($(a9)+(0,-0.34)$) {9};
  \node[map] at ($(a13)+(0,-0.34)$) {13};
  \node[map] at ($(a12)+(0,-0.34)$) {12};
  \node[map] at ($(a14)+(0,-0.34)$) {14};
  \node[map] at ($(a15)+(0,-0.34)$) {15};
  \node[note] at (3.50,0.55)
    {Arrow-implied endpoint: \(\mathrm{O}_9{-}\mathrm{H}_{13}\), radical at 12\\
     Mapped endpoint: \(\mathrm{O}_9{-}\mathrm{H}_{15}\), radical at 12
     \(\;\textcolor{ioVermilion!78!black}{\ne}\)};
\end{scope}

\begin{scope}[shift={(7.35,4.55)}]
  \draw[iopanel] (0,0) rectangle (7.0,4.2);
  \iosub{0.12,4.08}{B}{ID 2,207 \(\cdot\) proposed atom map}

  \node[font=\scriptsize] (br) at (0.54,2.10) {R};
  \node[catom] (b9) at (1.18,2.10) {\aO};
  \node[crad] (b9rad) at ($(b9)+(0.24,0.31)$) {};
  \node[catom] (b13) at (2.35,2.10) {\aH};
  \node[catom] (b12) at (3.35,2.10) {\aC};
  \node[catom] (b14) at (4.50,2.10) {\aO};
  \node[catom] (b15) at (5.60,2.10) {\aH};
  \node[font=\scriptsize] (brr) at (3.35,2.88) {R};

  \draw[cbond] (br)--(b9);
  \draw[cbond] (b13)--(b12);
  \draw[cbond] (b12)--(b14);
  \draw[cbond] (b14)--(b15);
  \draw[cbond] (b12)--(brr);
  \draw[newbond] (b9)--(b13);

  \coordinate (b913) at ($(b9)!0.53!(b13)$);
  \coordinate (b1213) at ($(b12)!0.50!(b13)$);
  \draw[cfish] (b9rad.center) to[bend left=24] ($(b913)+(0,0.10)$);
  \draw[cfish] ($(b1213)+(0.04,0.08)$)
    to[bend right=38] ($(b913)+(0.03,-0.08)$);
  \draw[cfish] ($(b1213)+(-0.04,-0.08)$)
    to[bend right=28] ($(b12)+(-0.06,0.35)$);

  \node[map] at ($(b9)+(0,-0.34)$) {9};
  \node[map] at ($(b13)+(0,-0.34)$) {13};
  \node[map] at ($(b12)+(0,-0.34)$) {12};
  \node[map] at ($(b14)+(0,-0.34)$) {14};
  \node[map] at ($(b15)+(0,-0.34)$) {15};
  \node[note] at (3.50,0.55)
    {Arrow-implied endpoint: \(\mathrm{O}_9{-}\mathrm{H}_{13}\), radical at 12\\
     Corrected map: \(\mathrm{O}_9{-}\mathrm{H}_{13}\), radical at 12
     \(\;\textcolor{ioGreen!45!black}{\checkmark}\)};
\end{scope}

\begin{scope}[shift={(0,0)}]
  \draw[iopanel] (0,0) rectangle (7.0,4.2);
  \iosub{0.12,4.08}{C}{ID 2,300 \(\cdot\) published annotation}
  \begin{scope}[yshift=0.45cm]

  \node[catom] (c1) at (1.05,1.18) {\aC H$_2$};
  \node[crad] (c1rad) at ($(c1)+(0.31,0.31)$) {};
  \node[catom] (c2) at (1.68,2.30) {\aO};
  \node[catom] (c3) at (3.00,2.95) {\aH};
  \node[catom] (c4) at (4.30,2.30) {\aO};
  \node[catom] (c5) at (4.92,1.18) {\aO};
  \node[crad] at ($(c5)+(0.25,0.31)$) {};

  \draw[cbond] (c1)--(c2);
  \draw[cbond] (c2)--(c3);
  \draw[cbond] (c4)--(c5);
  \draw[newbond] (c3)--(c4);

  \coordinate (c12) at ($(c1)!0.50!(c2)$);
  \coordinate (c23) at ($(c2)!0.50!(c3)$);
  \coordinate (c34) at ($(c3)!0.50!(c4)$);
  \draw[cfish] (c1rad.center) to[bend right=23] ($(c12)+(-0.08,-0.02)$);
  \draw[cfish] ($(c23)+(0.04,0.06)$)
    to[bend left=28] ($(c34)+(-0.05,0.07)$);
  \draw[cfish] ($(c4)+(-0.02,0.27)$)
    to[bend left=26] ($(c34)+(0.06,-0.06)$);
  \draw[cfish] ($(c23)+(-0.05,-0.05)$)
    to[bend left=42] ($(c12)+(0.06,0.03)$);

  \node[map] at ($(c1)+(0,-0.38)$) {1};
  \node[map] at ($(c2)+(-0.28,0.02)$) {2};
  \node[map] at ($(c3)+(0,0.32)$) {3};
  \node[map] at ($(c4)+(0.28,0.02)$) {4};
  \node[map] at ($(c5)+(0,-0.38)$) {5};
  \end{scope}
  \node[note] at (3.50,0.52)
    {Arrow-implied endpoint: \(\mathrm{C}_1{=}\mathrm{O}_2\) and
     \(\mathrm{H}_3{-}\mathrm{O}_4{-}\mathrm{O}_5\)\\
     Mapped endpoint: \(\mathrm{H}_3{-}\mathrm{O}_2{-}\mathrm{O}_4
     + \mathrm{C}_1{=}\mathrm{O}_5\)
     \(\;\textcolor{ioVermilion!78!black}{\ne}\)};
\end{scope}

\begin{scope}[shift={(7.35,0)}]
  \draw[iopanel] (0,0) rectangle (7.0,4.2);
  \iosub{0.12,4.08}{D}{ID 2,300 \(\cdot\) reviewed annotation}
  \begin{scope}[yshift=0.45cm]

  \node[catom] (d1) at (1.05,1.18) {\aC H$_2$};
  \node[crad] (d1rad) at ($(d1)+(0.31,0.31)$) {};
  \node[catom] (d2) at (1.68,2.30) {\aO};
  \node[catom] (d3) at (3.00,2.95) {\aH};
  \node[catom] (d4) at (4.30,2.30) {\aO};
  \node[catom] (d5) at (4.92,1.18) {\aO};
  \node[crad] (d5rad) at ($(d5)+(0.25,0.31)$) {};

  \draw[cbond] (d1)--(d2);
  \draw[cbond] (d2)--(d3);
  \draw[cbond] (d4)--(d5);
  \draw[newbond] (d1)--(d5);

  \coordinate (d15) at ($(d1)!0.50!(d5)$);
  \draw[cfish] (d1rad.center) to[bend right=28] ($(d15)+(-0.10,-0.07)$);
  \draw[cfish] (d5rad.center) to[bend left=28] ($(d15)+(0.10,0.07)$);

  \node[map] at ($(d1)+(0,-0.38)$) {1};
  \node[map] at ($(d2)+(-0.28,0.02)$) {2};
  \node[map] at ($(d3)+(0,0.32)$) {3};
  \node[map] at ($(d4)+(0.28,0.02)$) {4};
  \node[map] at ($(d5)+(0,-0.38)$) {5};
  \end{scope}
  \node[note] at (3.50,0.52)
    {Arrow-implied endpoint: radical \(\mathrm{C}_1{=}\mathrm{O}_5\) component\\
     Mapped endpoint: radical \(\mathrm{C}_1{=}\mathrm{O}_5\) component
     \(\;\textcolor{ioGreen!45!black}{\checkmark}\)};
\end{scope}

\end{tikzpicture}}
\caption{Local radical-arrow review.  Red single-barbed arrows are fishhooks.
Dashed lines are destination bonds.  \textbf{(A)} ID 2,207 conflicts with its
current mapped endpoint.  \textbf{(B)} The same arrows match the proposed H13
hydrogen-map correction, which is not applied.  \textbf{(C,D)} ID 2,300 before
and after annotation review.  Paired moves are listed in
Table~\ref{tab:radical-arrow-review}.}
\label{fig:si-radical-arrow-review}
\end{figure}
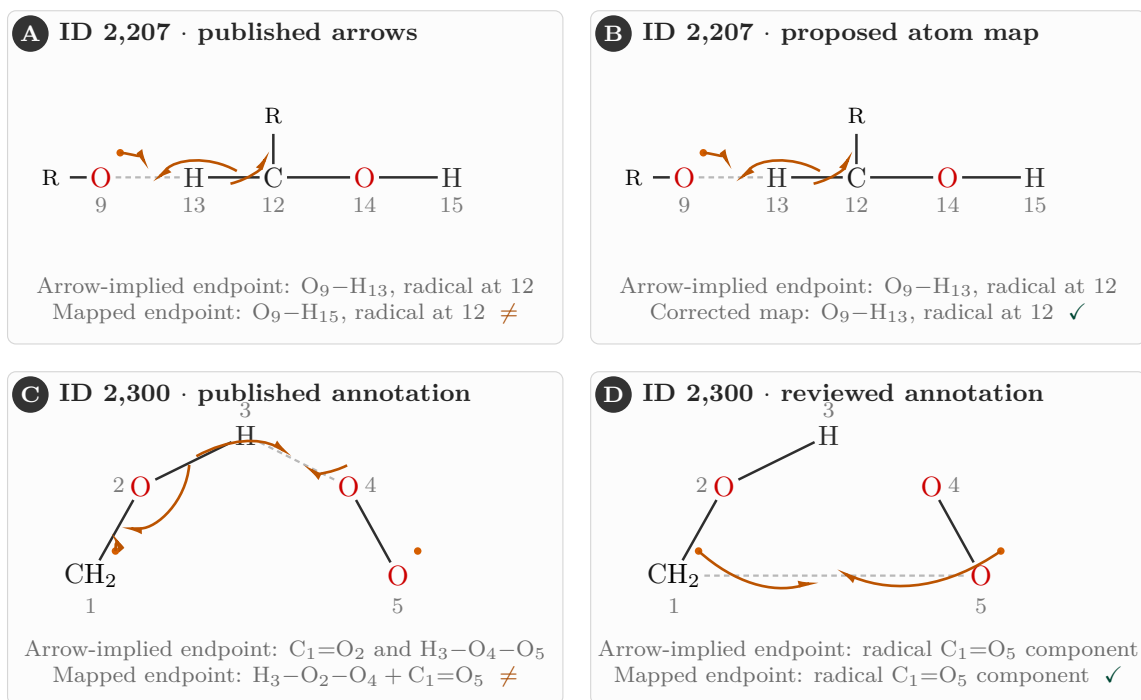
\clearpage

\centering
\small
\renewcommand{\arraystretch}{1.24}
\setlength{\tabcolsep}{5pt}
\setlength{\LTcapwidth}{\textwidth}
\begin{longtable}{@{}>{\raggedleft\arraybackslash}p{1.3cm}
  >{\raggedright\arraybackslash}p{6.6cm}
  >{\raggedright\arraybackslash}p{6.6cm}@{}}
\caption{Published and reviewed radical transition annotations by record ID.
Parentheses denote bond loci.  \(\rightharpoonup\) is a one-electron fishhook
and \(\longrightarrow\) a paired move.  Boldface marks the changed arrow in
both columns.  \(\dagger\) denotes unchanged arrows after endpoint-orbit
alignment.  For \(2{,}207^*\), the arrows are unchanged.  The proposed endpoint
atom-map correction assigns H13 to O9 and retains O14--H15.  It is not applied.}
\label{tab:radical-arrow-review}\\
\toprule
ID & Published annotation & Reviewed annotation \\
\midrule
\endfirsthead
\multicolumn{3}{@{}l}{\small Table~\ref{tab:radical-arrow-review} continued}\\[3pt]
\toprule
ID & Published annotation & Reviewed annotation \\
\midrule
\endhead
\midrule
\multicolumn{3}{r@{}}{\small Continued on next page}
\endfoot
\bottomrule
\endlastfoot
279 &
\(\begin{aligned}[t]
2&\rightharpoonup(2,14)\\
(14,15)&\rightharpoonup15\\
\boldsymbol{(14,15)}&\boldsymbol{\rightharpoonup15}\\
(14,15)&\rightharpoonup(2,14)
\end{aligned}\) &
\(\begin{aligned}[t]
2&\rightharpoonup(2,14)\\
(14,15)&\rightharpoonup(2,14)\\
\boldsymbol{(14,15)}&\boldsymbol{\rightharpoonup15}
\end{aligned}\) \\
\addlinespace[2pt]\cmidrule(lr){1-3}\addlinespace[2pt]
404 &
\(\begin{aligned}[t]
\boldsymbol{5}&\boldsymbol{\rightharpoonup(5,21)}\\
(21,22)&\rightharpoonup(5,22)\\
(21,22)&\rightharpoonup21
\end{aligned}\) &
\(\begin{aligned}[t]
\boldsymbol{5}&\boldsymbol{\rightharpoonup(5,22)}\\
(21,22)&\rightharpoonup(5,22)\\
(21,22)&\rightharpoonup21
\end{aligned}\) \\
\addlinespace[2pt]\cmidrule(lr){1-3}\addlinespace[2pt]
540 &
\(\boldsymbol{1\rightharpoonup(1,2)}\) &
\(\boldsymbol{1\longrightarrow(1,2)}\) \\
\addlinespace[2pt]\cmidrule(lr){1-3}\addlinespace[2pt]
\(1{,}310^\dagger\) &
\(\begin{aligned}[t]
\boldsymbol{13}&\boldsymbol{\rightharpoonup(6,13)}\\
(5,6)&\rightharpoonup(6,13)\\
(5,6)&\rightharpoonup5
\end{aligned}\) &
\(\begin{aligned}[t]
\boldsymbol{13}&\boldsymbol{\rightharpoonup(6,13)}\\
(5,6)&\rightharpoonup(6,13)\\
(5,6)&\rightharpoonup5
\end{aligned}\) \\
\addlinespace[2pt]\cmidrule(lr){1-3}\addlinespace[2pt]
\(1{,}317^\dagger\) &
\(\begin{aligned}[t]
\boldsymbol{19}&\boldsymbol{\rightharpoonup(1,19)}\\
(1,2)&\rightharpoonup(1,19)\\
(1,2)&\rightharpoonup2
\end{aligned}\) &
\(\begin{aligned}[t]
\boldsymbol{19}&\boldsymbol{\rightharpoonup(1,19)}\\
(1,2)&\rightharpoonup(1,19)\\
(1,2)&\rightharpoonup2
\end{aligned}\) \\
\addlinespace[2pt]\cmidrule(lr){1-3}\addlinespace[2pt]
2,046 &
\(\boldsymbol{21\rightharpoonup10}\) &
\(\boldsymbol{10\rightharpoonup21}\) \\
\addlinespace[2pt]\cmidrule(lr){1-3}\addlinespace[2pt]
\(2{,}207^*\) &
\(\begin{aligned}[t]
\boldsymbol{9}&\boldsymbol{\rightharpoonup(9,13)}\\
\boldsymbol{(12,13)}&\boldsymbol{\rightharpoonup(9,13)}\\
\boldsymbol{(12,13)}&\boldsymbol{\rightharpoonup12}
\end{aligned}\) &
\(\begin{aligned}[t]
\boldsymbol{9}&\boldsymbol{\rightharpoonup(9,13)}\\
\boldsymbol{(12,13)}&\boldsymbol{\rightharpoonup(9,13)}\\
\boldsymbol{(12,13)}&\boldsymbol{\rightharpoonup12}
\end{aligned}\) \\
\addlinespace[2pt]\cmidrule(lr){1-3}\addlinespace[2pt]
2,300 &
\(\begin{aligned}[t]
\boldsymbol{1}&\boldsymbol{\rightharpoonup(1,2)}\\
\boldsymbol{(2,3)}&\boldsymbol{\rightharpoonup(3,4)}\\
\boldsymbol{4}&\boldsymbol{\rightharpoonup(3,4)}\\
\boldsymbol{(2,3)}&\boldsymbol{\rightharpoonup(1,2)}
\end{aligned}\) &
\(\begin{aligned}[t]
\boldsymbol{1}&\boldsymbol{\rightharpoonup(1,5)}\\
\boldsymbol{5}&\boldsymbol{\rightharpoonup(1,5)}\\
\boldsymbol{(1,2)}&\boldsymbol{\longrightarrow(2,4)}\\
\boldsymbol{(4,5)}&\boldsymbol{\longrightarrow(1,5)}
\end{aligned}\) \\
\addlinespace[2pt]\cmidrule(lr){1-3}\addlinespace[2pt]
2,516 &
\(\begin{aligned}[t]
4&\rightharpoonup(4,8)\\
\boldsymbol{9}&\boldsymbol{\rightharpoonup(4,8)}
\end{aligned}\) &
\(\begin{aligned}[t]
4&\rightharpoonup(4,8)\\
\boldsymbol{8}&\boldsymbol{\rightharpoonup(4,8)}
\end{aligned}\) \\
\addlinespace[2pt]\cmidrule(lr){1-3}\addlinespace[2pt]
3,970 &
\(\begin{aligned}[t]
(20,21)&\rightharpoonup20\\
\boldsymbol{(20,21)}&\boldsymbol{\rightharpoonup2}
\end{aligned}\) &
\(\begin{aligned}[t]
(20,21)&\rightharpoonup20\\
\boldsymbol{(20,21)}&\boldsymbol{\rightharpoonup21}
\end{aligned}\) \\
\addlinespace[2pt]\cmidrule(lr){1-3}\addlinespace[2pt]
4,852 &
\(\begin{aligned}[t]
10&\rightharpoonup(10,20)\\
(20,21)&\rightharpoonup(10,20)\\
\boldsymbol{(20,21)}&\boldsymbol{\rightharpoonup2}
\end{aligned}\) &
\(\begin{aligned}[t]
10&\rightharpoonup(10,20)\\
(20,21)&\rightharpoonup(10,20)\\
\boldsymbol{(20,21)}&\boldsymbol{\rightharpoonup21}
\end{aligned}\) \\
\end{longtable}

\end{document}